%% file: ivqr_id.tex
\documentclass[11pt,letterpaper]{article}
\usepackage{docmute}

\input{settings}

	\title{Identification of multi-valued treatment effects with unobserved heterogeneity}
	
	\author{Koki Fusejima\thanks{Koki Fusejima, Institute for Advanced Study, Hitotsubashi University, 2-1 Naka, Kunitachi-shi, Tokyo 186-8601, Japan; Email: k.fusejima@r.hit-u.ac.jp} \\
    Institute for Advanced Study \\ Hitotsubashi University}
	
	\date{This version: \today}

\begin{document}

	\maketitle

\input{abstract}

\newpage

\input{introduction}
\input{model}
\input{assumption}
\input{identification}
\input{conclusion}

\def\thesection{}

\section{Appendix}

\appendix

\input{proofs}
\input{otherproofs}

\def\thesection{}

\section{Acknowledgments}
I would like to thank Katsumi Shimotsu, my advisor at The University of Tokyo, Hidehiko Ichimura, Yuichi Kitamura, Hiroaki Kaido, Takuya Ishihara, Masayuki Sawada, Ryo Imai, Ryota Yuasa, and the seminar participants at The University of Tokyo, Otaru University of Commerce, and Hitotsubashi University for their helpful comments on this research. This research is supported by JSPS KAKENHI Grant Number JP20J20046.

\bibliography{ivqr_id}

\end{document}


\maketitle

\def\thesection{\Alph{section}}

\input{abstract_supp}
\input{closedform}
\input{fullrank}
\input{without}

\bibliography{ivqr_id}

%% file: settings.tex
\usepackage[utf8]{inputenc}

\usepackage{amsthm,amsmath,amssymb}

\usepackage{letterspace}

\usepackage{graphicx}
\usepackage{setspace}
\usepackage{ascmac}
\usepackage{bm}
\usepackage{tikz}
\usetikzlibrary{intersections,calc,arrows.meta,shapes.callouts}

\usepackage{xcolor}

\usepackage{filecontents}
\makeatletter\@input{ivqr_id_supp_a.tex}\makeatother

\usepackage{nameref} 
\usepackage{zref-xr}
\zxrsetup{toltxlabel} 
\zexternaldocument*{ivqr_id_supp}

\usepackage{algorithm2e}
\usepackage{algorithmic}
\usepackage{enumitem}
\usepackage{latexsym}
\usepackage[hang,small,bf]{caption}
\usepackage[subrefformat=parens]{subcaption}
\newcommand{\R}{\mathbb{R}}

\newcommand{\mA}{\mathcal{A}}

\newcommand{\mC}{\mathcal{C}}	
\newcommand{\mD}{\mathcal{D}}

\newcommand{\mP}{\mathcal{P}}

\newcommand{\mT}{\mathcal{T}}

\newcommand{\mW}{\mathcal{W}}
\newcommand{\mX}{\mathcal{X}}
\newcommand{\mY}{\mathcal{Y}}
\newcommand{\mZ}{\mathcal{Z}}

\newcommand{\deq}{\stackrel{d}{=}}
\newcommand{\emp}{\varnothing}

\newcommand{\To}{\Rightarrow}
\newcommand{\toto}{\Leftrightarrow}
\newcommand{\eps}{\varepsilon}

\DeclareMathOperator*{\argmax}{arg\,max}

\theoremstyle{plain}
\newtheorem{theorem}{Theorem}
\newtheorem{lemma}{Lemma}

\newtheorem{assumption}{Assumption}

\theoremstyle{definition}

\newtheorem{remark}{Remark}

\usepackage[top=1.25in, bottom=1.25in, left=1.25in, right=1.25in]{geometry}

\numberwithin{equation}{section}

\usepackage{color}

\AtBeginDocument{
  \abovedisplayskip     =0.6\abovedisplayskip
  \abovedisplayshortskip=0.6\abovedisplayshortskip
  \belowdisplayskip     =0.6\belowdisplayskip
  \belowdisplayshortskip=0.6\belowdisplayshortskip}

\let\oldenumerate\enumerate
\renewcommand{\enumerate}{
\oldenumerate
\setlength{\itemsep}{1.5pt}
\setlength{\parskip}{1.5pt}
\setlength{\parsep}{1.5pt}
}

\usepackage{hyperref}
\hypersetup{
       colorlinks=false,
       citebordercolor=green,
       linkbordercolor=red,
       urlbordercolor=cyan,
}
\PassOptionsToPackage{linktocpage}{hyperref}
\usepackage{breakcites}
\usepackage{breakurl}
\usepackage[authoryear, round]{natbib}

%% file: abstract.tex
\begin{abstract}
In this paper, we establish sufficient conditions for identifying treatment effects on continuous outcomes in endogenous and multi-valued discrete treatment settings with unobserved heterogeneity. We employ the monotonicity assumption for multi-valued discrete treatments and instruments, and our identification condition has a clear economic interpretation. In addition, we identify the local treatment effects in multi-valued treatment settings and derive closed-form expressions of the identified treatment effects. We provide examples to illustrate the usefulness of our result.  
\end{abstract}

\textit{Keywords:} Treatment effect, unobserved heterogeneity, identification, endogeneity, instrumental variable

\textit{JEL classification:} C14, C21, C26

%% file: introduction.tex
\section{Introduction}
Unobserved heterogeneity in treatment effects is an essential consideration in many empirical studies in economics.
As discussed in \cite{Heckman2001}, for example, economic theory and applications strongly suggest that the causal effects of treatments or policy variables differ across individuals and subpopulations with the same characteristics.
Quantile treatment effects characterize the heterogeneous impacts of treatments on individuals with different levels of unobserved characteristics in terms of potential outcome quantiles.
Based on instrumental variable (IV) methods, local treatment effects, as introduced by \cite{Imbens1994}, are the treatment effects conditional on an unobservable subpopulation for which the instrument affects treatment states.

In this paper, we establish sufficient conditions for identifying treatment effects on continuous outcomes in endogenous and multi-valued discrete treatment settings with unobserved heterogeneity using IV methods. We use only discrete instruments for identification because instruments are discrete in many empirical applications.
Discrete treatments are implicitly or explicitly multi-valued in many applications. For example, households may receive different levels of transfers in anti-poverty programs, and students who wish to attend college have multiple options for choosing a college or major.
For the policy-maker, it is essential to compare such multi-valued treatment effects when determining which treatment level is appropriate.

In multi-valued treatment settings, \cite{CH2005} establish the identification of quantile treatment effects (QTEs; on the observed populations) with discrete instruments, and their identification results are testable in principle.
However, as we present in the next section, it is unclear how to interpret the required numerical conditions in each empirical study economically.

The main contributions of this paper are stated as follows. 
We establish sufficient conditions for identifying treatment effects in multi-valued treatment settings, which are easier to interpret economically than the identification conditions of \cite{CH2005}.
In addition, we provide closed-form expressions of the identified treatment effects.
We also establish the identification of local treatment effects in multi-valued treatment settings based on our assumptions.

To illustrate the usefulness of our results, we provide three examples based on empirical research and discuss the applicability of our identification result for these examples.
The first example is the effects of choosing different fields on earnings in postsecondary education.
The second example is the effects of expanding access to two-year colleges on student outcomes.
The third example is the effects of relocating to low-poverty neighborhoods on the outcomes of disadvantaged families living in high-poverty neighborhoods. 

We provide identification conditions that take the form of monotonicity assumptions.
The monotonicity assumption, introduced by \cite{Imbens1994} for the binary treatment case, has a clear interpretation in empirical studies because they can be motivated based on behavioral assumptions.
The monotonicity assumption we employ for multi-valued and unordered treatments is motivated by an unordered discrete choice model, where each individual chooses the treatment option with the highest utility.
Our monotonicity assumption is related to \cite{Heckman2018}'s ``unordered monotonicity" assumption for multi-valued and unordered treatments.
However, the identification condition we provide is based on a weaker assumption that holds only on a particular subset of the pairs of values the instrument can take.
As highlighted in our examples, imposing the monotonicity assumption on all the instrument pairs employed in many studies, including \cite{Heckman2018}, may be too strong when the instrument is multi-valued.

When the treatment is binary, \cite{VX2017} and \cite{Wuthrich2019} show identification of the treatment effects as closed-form expressions under 
the instrument satisfying the monotonicity assumption.
\cite{Wuthrich2019} and \cite{FVX2019} develop plug-in estimators based on the closed-form expressions of the identified treatment effects. 
Our identification analysis covers more general settings with
multi-valued treatments and instruments, regardless of whether the treatment is ordered or unordered.

The identification approach adopted here generalizes the idea of matching two distributions, which is introduced by \cite{AtheyImbens2006} and used for identifying the treatment effects in some recent studies.
For the binary treatment case, \cite{VX2017} and \cite{Wuthrich2019} establish identification with binary instruments by matching the two potential outcome distributions conditional on the same subpopulation called ``compliers'' under the monotonicity assumption.
In continuous treatment settings, \cite{torgovitsky2015}, \cite{d2015}, and \cite{ishihara2017} establish identification with binary instruments employing sufficiently large support of the treatment. 

In multi-valued treatment settings, however, we generally cannot match the potential outcome distributions for two treatment states in the same subpopulation. 
The selection mechanism becomes more complicated than in the binary treatment case, and the support of the treatment is limited compared to the continuous treatment case.
We overcome this difficulty by developing systems of equations with multiple potential outcome distributions.
These equations are derived from relationships between the compliers that can be motivated by the discrete choice model under our monotonicity assumption.
Although the distributions are conditional on different subpopulations, we show that the simultaneous equations can be solved uniquely, and the potential outcome distributions are identified under our assumptions.

We also employ our monotonicity relationships for identifying local treatment effects in multi-valued treatment settings.
We establish identification when the outcome variable is continuously distributed under our monotonicity assumption, with some additional assumptions on unobservable factors, such as the rank similarity assumption.\footnote{In Section \ref{subsec:4.5}, we briefly review the identification studies of local treatment effects in multi-valued treatment settings.}

The remainder of the paper is organized as follows:
in Section \ref{sec:2}, we introduce our basic setup and provide three real-world examples. 
In Section \ref{sec:3}, we introduce our monotonicity assumption for multi-valued and unordered treatments.
In Section \ref{sec:4}, we establish the identification of the treatment effects using our monotonicity assumption.
Section \ref{sec:6} provides conclusions with brief recommendations for estimation.
Proofs of the main results and some auxiliary results are presented in Appendices \ref{sec:a}- \ref{sec:a.5}. Some additional discussions are given in Appendices \ref{sec:a.25}-\ref{sec:g} of the Supplemental Appendix.

%% file: model.tex
\section{Basic set up and motivating examples}\label{sec:2}
In this section, we introduce our basic setup with a benchmark nonseparable model. 
We provide three real-world examples to show that our identification approach can be applied in well-known empirical settings.

\subsection{Notation and basic assumptions}\label{subsec:2.1}
Throughout this paper, we use the notations $F_A$, $Q_A$, and $f_A$ for the unconditional cumulative distribution function (cdf), quantile function (qf), and probability density function (pdf) of a scalar-valued random variable $A$, respectively. Similarly, for a set $\mD$ and random vectors $B$ and $C$, $F_{A|\mD BC}(\cdot|b,c)$,
$Q_{A|\mD BC}(\cdot|b,c)$, and $f_{A|\mD BC}(\cdot|b,c)$ denote the conditional cdf, qf, and pdf of $A$ on $\mD\cap\{(B,C)=(b,c)\}$, respectively; let $\mD^\circ$ denote the interior of $\mD$.

To introduce our basic setup, we consider the following nonseparable simultaneous equations for a continuous outcome and a multi-valued endogenous treatment:
\begin{eqnarray}
Y &=& g(T,X,U),  \label{model.1}\\
T &=& \rho(Z,X,V). \label{model.2}
\end{eqnarray}
For the outcome equation, $Y$ is the outcome, $T\in\mT$ is the multi-valued (possibly) unordered treatment where $\mT$ contains $k+1$ values, $X\in\mX\subset\R^r$ is a vector of observed covariates, and the random vector $U$ captures unobserved heterogeneity in the effect of $T$ on $Y$.
For the treatment $T$, we assume a finite collection of multiple treatment statuses (either unordered or ordered) indexed by $t\in \mT$ where, without loss of generality, $\mT=\{0,1,2,\ldots, k\}$.
For the treatment equation, $Z\in\mZ$ is a discrete instrument where $\mZ$ contains at least $k+1$ values, and the random vector $V$ captures unobserved factors affecting selection into treatment.
For simplicity, we suppress $X$ throughout the identification analysis.
All assumptions and results can be understood as conditional on $X$.
The potential outcome under each treatment level $t\in\mT$ is $Y_t=g(t,U)$, and we assume $Y_t\in\mY\subset\R$ and $E[|Y_t|]<\infty$. 
The potential treatment choice if $Z$ had been externally set to $z$ is $T(z)=\rho(z,V)$.
In this paper, for two different treatment levels $t$ and $t'$, we are interested in the sufficient conditions for identifying the average treatment effect (ATE): $E[Y_t] - E[Y_{t'}]$ and the quantile treatment effect (QTE): $Q_{Y_t} (\tau) - Q_{Y_{t'}} (\tau)$, where $\tau\in(0,1)$. 
We are also interested in the sufficient conditions for identifying the local treatment effects, and we introduce them in Section \ref{subsec:4.5}.
For notation simplicity, we assume that the supports of the distributions of $T$, $Y_t$, $Y$, and $Z$ are equal to $\mT$, $\mY$, $\mY$, and $\mZ$, respectively. 
The results in this paper do not rely on these restrictions.

For the outcome equation \eqref{model.1}, we allow $U=(U_0,\ldots,U_k)'$ to be multi-dimensional, and we assume that the potential outcome is expressed as $Y_t=g(t,U_t)$, where we define $U_t:=F_{Y_t} (Y_t)$.
$U_t$ is the rank variable that characterizes heterogeneity of outcomes for individuals with the same observed characteristics by relative ranking in terms of potential outcomes. 
For the rank variable, we assume the ``rank similarity" introduced in \cite{CH2005}.
The rank similarity is an assumption that weakens the ``rank invariance" assumption. 
Rank invariance assumes that $U$ is a scalar error term and that the rank variables satisfy $U_t=U_{t'}=U$ for any $t\neq t'$. 
However, rank similarity allows the rank variables to deviate from a common ranking $U$.

For the treatment equation \eqref{model.2}, we allow $V=(V_0,\ldots,V_k)'$ to be multi-dimensional for the unordered treatment $T$, and each element $V_t$ represents unobserved individual preference heterogeneity from choosing $T=t$.
The treatment decision can be explained by an unordered discrete choice model, where each individual chooses the treatment option with the highest indirect utility, 
\begin{equation}\label{eq:chcf}
T(z) = \rho(z,V) = \argmax_{t\in\mT} I_t(z,V_t),
\end{equation}
where $I_t$ is the indirect utility of choosing $T=t$.\footnote{We can justify the utility maximization model of the (potential) treatment choice as in \eqref{eq:chcf} under the rank similarity assumption with an argument similar to Example 2 of \cite{chernozhukov2013quantile}, where QTE is used to examine the effects of participating in a 401(k) plan.}
\cite{chesher2013instrumental} also employs a similar unordered choice model where $V$ is allowed to be multi-dimensional.

When the unobserved factor $V$ is a scalar random variable, the two equations \eqref{model.1} and \eqref{model.2} form the triangular model of \cite{chesher2005nonparametric}.
\cite{chesher2005nonparametric} studies interval identification of the endogenous nonseparable triangular model with a discrete ordered treatment.
Our research is related to \cite{chesher2005nonparametric}, but our treatment equation for unordered treatments is fundamentally different from the triangular model that captures a single source of unobserved heterogeneity.

To identify these treatment effects, it suffices to identify the conditional mean and qf of the potential outcomes.
We identify these under the following set of assumptions. \cite{CH2005}, \cite{VX2017}, and \cite{Wuthrich2019} employ a similar set of assumptions. 
\begin{assumption}[Instrument independence and rank similarity]\label{asp:ivqr}
The following conditions hold:
\begin{enumerate}
\item[(i)] Potential outcomes: For each $t\in\mT$, $Y_t$ is expressed as $Y_t=g(t,U_t)$ for some unknown function $g$ and $U_t=F_{Y_t} (Y_t)$, and $F_{Y_t} (\cdot)$ is continuous.
\item[(ii)] Independence: $\{U_t\}_{t=0}^k$ are independent of $Z$.
\item[(iii)] Selection: $T$ is expressed as $T = \rho(Z, V)$ for some unknown function $\rho$ and random vector $V$.
\item[(iv)] Rank similarity: Conditional on $(Z,V) = (z,v)$, $\{U_t\}_{t=0}^k$ are identically distributed.
\item[(v)] Outcome support: The closure of $\mY^\circ$ is equal to $\mY$, and $F_{Y_t}(\mY^\circ)$ does not depend on $t\in\mT$.
\end{enumerate}
\end{assumption}
\noindent Assumption \ref{asp:ivqr} (i) states an expression for the potential outcome with the rank variable and imposes continuity of the potential outcome cdf.
Under Assumption \ref{asp:ivqr} (i), $F_{Y_t} (y)$ is strictly increasing in $y\in\mY^\circ$, and $Q_{Y_t} (\tau)$ is strictly increasing in $\tau\in(0,1)$.\footnote{Lemmas \ref{lem:3b} and \ref{lem:3c} in Appendix \ref{sec:a.5} prove these properties.}
We do not assume that $Q_{Y_t} (\cdot)$ is continuous and allow $F_{Y_t} (\cdot)$ to have flat intervals.
Under Assumption \ref{asp:ivqr} (i), the rank variable $U_t$ follows the uniform distribution on $(0,1)$, and $Y_t$ and $Q_{Y_t}(U_t)$ are identically distributed.
Hence, we can interpret the QTE as treatment effects on individuals with the same level of unobserved heterogeneity at some level $U_t=\tau$.\footnote{We employ a slightly different definition for the rank variable from the original definition of \cite{CH2005}. 
\cite{CH2005} define the rank variable $U_t$ as a uniformly distributed random variable on $(0,1)$ that satisfies $Y_t=Q_{Y_t}(U_t)$, and they directly assume that $Q_{Y_t} (\tau)$ is strictly increasing in $\tau\in(0,1)$.
The difference does not matter in our settings because $Y_t$ and $Q_{Y_t}(U_t)$ are identically distributed.}
Assumption \ref{asp:ivqr} (ii) imposes conditional independence between the potential outcomes and the instrument. 
Assumption \ref{asp:ivqr} (iii) states a general selection equation where the random vector $V$ captures unobserved factors affecting selection into treatment. 
Assumption \ref{asp:ivqr} (iv) is the rank similarity assumption.
The rank similarity is arguably strong, but this condition has essential implications for identification and is consistent with many empirical situations.
Assumption \ref{asp:ivqr} (v) is assumed to simplify the proofs in the main paper, and we relax Assumption \ref{asp:ivqr} (v) in Appendix \ref{sec:g}. 
See also Remark \ref{rem:clfm} for related discussions.

The main statistical implication of Assumption \ref{asp:ivqr} is that, for each $t\in\mT$, the qf of $Y_t$ satisfies the following nonlinear moment equation (\cite{CH2005} Theorem 1): 
\begin{equation}\label{eq:ti}
\sum_{t=0}^k F_{Y|TZ}(Q_{Y_t} (\tau)|t,z)p_t(z)=\tau,
\end{equation}
where $p_t(z)$ is defined as 
\begin{equation}\label{eq:pt}
p_t(z):=P(T=t|Z=z).
\end{equation}
\cite{CH2005} show that $Q_{Y_t} (\tau )$'s are identified if the following $(k+1)\times(k+1)$ matrix $\Pi'(y_0,\ldots,y_k)$ is full rank for all the values $(y_0,\ldots,y_k)$ in a set of potential solutions to the moment equations (\ref{eq:ti}):
\begin{equation}\label{eq:ch2}
\Pi'(y_0,\ldots,y_k):=
\begin{pmatrix}
f_{Y|TZ}(y_0|0,z_0)p_0(z_0) & \cdots & f_{Y|TZ}(y_k|k,z_0)p_k(z_0) \\
\vdots & \ddots & \vdots \\
f_{Y|TZ}(y_0|0,z_k)p_0(z_k) & \cdots & f_{Y|TZ}(y_k|k,z_k)p_k(z_k)
\end{pmatrix},
\end{equation}
where $\{z_0,z_1\ldots,z_k\}\subset\mZ$. 
The identification condition of \cite{CH2005} is, in principle, directly testable.
However, in each empirical study, it is not easy to check this numerical condition, which takes the form of matrices of the outcome conditional densities.
It is unclear how to interpret the requirements for the endogenous variable and the instruments implied in this condition. 
In this paper, we establish sufficient conditions for identifying treatment effects in multi-valued treatment settings, which are easier to interpret economically than the identification conditions of \cite{CH2005}.
We also provide closed-form expressions of the identified treatment effects that could be used for constructing plug-in estimators.\footnote{This idea resembles that of \cite{Das2005}, who develops an estimation strategy based on the closed-form expression of the regression function with discrete endogenous treatments in the nonparametric regression model with an additive error term.}

\subsection{Examples}\label{subsec:2.2ex}
In this section, we introduce three examples based on empirical research. 
Throughout this paper, we consider Example I as a running example.
In Section \ref{sec:5}, we discuss the content of our assumptions and the applicability of our identification result for Examples II and III.

\subsubsection{Example I}
The first example is the effects of choosing different fields on earnings in postsecondary education.
In postsecondary education, almost all students have to choose a field of study, and earnings differ across not only universities but also fields.
\cite{kirkeboen2016field} study the identification and estimation of local average treatment effects (LATEs) of choosing different fields on earnings in Norway's postsecondary educational system.
\cite{kirkeboen2016field} find that a centralized admission process in Norway randomizes applicants into different groups, and applicants in each group are much more likely to receive an offer for each field.
Based on this process, \cite{kirkeboen2016field} use the predicted offers for each field as an instrument.
This process also provides information on individuals' ranking of fields, and the identification analysis of \cite{kirkeboen2016field} depends on each individual's next best alternative, that is, the field one would prefer if one's preferred field would not be feasible.
Our results can be applied when we can find a certain instrument that effectively randomizes students into different groups, and information on individuals' ranking of fields is not necessary for our identification approach.

\subsubsection{Example II}
The second example is the effects of expanding access to two-year colleges on student outcomes.
Two-year community colleges will increase the flow of young people into higher education, and expanding access to two-year colleges is expected to positively affect educational attainment and earnings.
However, higher enrollment rates in two-year colleges may adversely affect student outcomes because college applicants are discouraged from paying higher tuition and enrolling directly in four-year colleges.
\cite{Mountjoy2019} and \cite{ferreyra2022labor} estimate the effects of expanding access to two-year colleges in the United States and Colombia using instruments, respectively.
The treatment is the college applicant's decision to start college at a two-year or four-year institution or not to enroll in college.
With this multi-valued treatment, they compare the two opposing effects: the positive effect on new two-year entrants who otherwise would not have enrolled in any college, and the negative effect on two-year entrants who otherwise would have started directly at a four-year institution.
For the instruments, they use the distance to the nearest college.
\cite{Mountjoy2019} directly uses the distance as a continuous instrument, and \cite{ferreyra2022labor} use a discrete instrument that indicates whether the nearest college is located within a certain distance radius.
We employ discrete instruments and show that our identification method can be applied to discrete instruments.

\subsubsection{Example III}
The third example is Moving to Opportunity (MTO), a housing experiment implemented between 1994 and 1998. 
The MTO experiment was designed to evaluate the effects of relocating to low-poverty neighborhoods on the outcomes of disadvantaged families living in high-poverty urban neighborhoods in the United States. 
This project randomly assigned housing vouchers from the Section 8 program that could be used to subsidize housing costs.
Eligible families were placed in one of the following three assignment groups: experimental group, to which Section 8 housing vouchers were assigned but restricted to use their vouchers in a low-poverty neighborhood; Section 8 group, to which regular Section 8 housing vouchers were assigned without any restriction on their place of use; or the control group to which no voucher was assigned.
Impact evaluations were conducted in 2002, 2009, and 2010.
See \cite{Orr2003}, \cite{Sanbonmatsu2011}, and \cite{Shroder2012}
for the detailed information on this project. 
For the recent studies that find evidence of neighborhood effects on adult employment, \cite{aliprantis2019evidence} estimate LATEs for moving to a higher-quality neighborhood under an ordered treatment model using neighborhood quality as an observed continuous measure of the treatment variable.

Under an unordered treatment model, \cite{Pinto2015} applies the identification results of \cite{Heckman2018} for estimating the conditional means of the potential outcomes on the compliers.
We show that, under some additional assumptions on unobservable factors, such as the rank similarity assumption, the ATEs and QTEs are nonparametrically identified when the outcome variable is continuously distributed. 

%% file: assumption.tex
\section{The (generalized) monotonicity assumption}\label{sec:3}
In this section, we introduce the monotonicity assumption we employ for the multi-valued treatment case.
We illustrate that the monotonicity assumption is motivated by economic analysis using a discrete choice index model introduced in Section \ref{sec:2}.

\subsection{Monotonicity assumption in multi-valued treatment settings}\label{subsec:3.1}
The monotonicity assumption is first introduced by \cite{Imbens1994} for the binary treatment case, and \cite{Heckman2018} generalize this assumption to  unordered multi-valued treatment settings.
We define a binary variable $D_t:=1\{T=t\}$, where $1\{\mA\}$ is the indicator function of a set $\mA$ and $D_t$ is an indicator function of each treatment level. 
Then, the observed outcome can be represented as $Y=\sum_{t=0}^kY_tD_t$.
We also define $D_t(z):=1\{T(z)=t\}$ as an indicator function of each potential treatment state if $Z$ had been externally set to $z$.
We define $\mP:=\{(z,z')\in\mZ^2:z\neq z'\}$ as a set of pairs of different values that the instrument can take.
The monotonicity assumption imposes restrictions on these pairs. 

For the binary treatment case, the monotonicity assumption requires that either $T(z)\leq T({z'})$ and $P(\{T(z)=0, T(z')=1\})>0$ or $T(z)\geq T({z'})$ and $P(\{T(z)=1, T(z')=0\})>0$ hold almost surely for each $(z,z')\in\mP$.
This condition implies that either $P(\{T(z)=1, T(z')=0\})=0$ or $P(\{T(z)=0, T(z')=1\})=0$ holds for each $(z,z')\in\mP$.
Under the monotonicity assumption, individuals who change their choice respond in only one direction to a change in $Z$, and the group with positive probability is called ``compliers."

\cite{Heckman2018} generalize this argument to multi-valued treatment settings. 
For each $z\in\mZ$ and $t\in\mT$, they assume the monotonicity assumption for binary treatments on each binary indicator $D_t(z)$, and either $D_t(z)\leq D_t(z')$ or $D_t(z)\geq D_t(z')$ holds almost surely for each pair $(z,z')\in\mP$.
They call this assumption ``unordered monotonicity'' because this condition can be assumed on the unordered treatments.
However, as we see in Section \ref{subsec:3.2}, imposing such conditions on all the pairs $(z,z')\in\mP$ may be too strong when the instrument is multi-valued. 
Hence, we employ a weaker assumption that imposes such conditions only on a subset of $\mP$. 
That subset is determined differently in each situation.
Recently, a similar problem has been discussed when there are multiple instruments. 
\cite{MTW2019,mogstad2020policy} and \cite{goff2020vector} consider the binary treatment case, and \cite{Mountjoy2019} considers the multi-valued unordered treatment case.
As discussed in Section \ref{sec:5}, \cite{Mountjoy2019} employs the same type of monotonicity assumption with multiple continuous instruments.
We consider a (possibly) scalar multi-valued instrument and weaken the monotonicity assumption from another perspective.

We define the compliers for the multi-valued treatment case as $\mC_{z,z'}^t:=\{D_t(z)=0, D_t(z')=1\}$. 
Our monotonicity assumption is characterized by inequalities such as $D_t(z)\leq D_t(z')$ and $D_t(z)\geq D_t(z')$. 
We employ the following monotonicity assumption:
\begin{assumption}[Instrument independence and monotonicity]\label{asp:mt}
There exists a subset $\Lambda$ of $\mP$ such that the following conditions hold for each $\lambda=(z,z')\in\Lambda$:
\begin{enumerate}
\item[(i)] Independence: $(Y_t, T(z))$ for $t\in\mT$ and $z\in\lambda$ are jointly independent of $Z$.
\item[(ii)] Monotonicity inequalities: Either $D_t(z)\leq D_t(z')$ or $D_t(z)\geq D_t(z')$ holds almost surely for each $t\in\mT$.
\item[(iii)] Instrument relevance: Either $P(\mC^t_{z,z'})>0$ or $P(\mC^t_{z',z})>0$ holds for each $t\in\mT$.
\item[(iv)] Sufficient support: The support of the conditional distribution of $Y_t$ on either $\mC^t_{z,z'}$ or $\mC^t_{z',z}$ is $\mY$.
\end{enumerate}
\end{assumption}

\noindent We call this subset $\Lambda$ ``monotonicity subset."
When $T$ is binary, Assumptions \ref{asp:mt} (i)--(iii) are the monotonicity assumption in \cite{Imbens1994}. 
Assumption \ref{asp:mt} (i) strengthens Assumption \ref{asp:ivqr} (ii) and assumes that the potential outcome and treatment are jointly independent of the instrument. 
Assumption \ref{asp:mt} (ii) assumes that monotonicity inequalities hold on a particular subset $\Lambda$ of $\mP$.
Assumption \ref{asp:mt} (iii) is an instrument relevance condition and assumes that the compliers always exist. 
Under these conditions, when monotonicity inequalities hold on $(z,z')\in\Lambda$, we exclude the cases where neither $D_t(z)<D_t(z')$ nor $D_t(z)>D_t(z')$ can happen for some $t\in\mT$.
Assumption \ref{asp:mt} (iv) strengthens the instrument relevance condition and assumes that the compliers are sufficiently large. 
\cite{VX2017} employ a similar condition for the binary treatment case.
Under Assumption \ref{asp:mt} (iv), each conditional cdf of the potential outcome on compliers strictly increases on $\mY^\circ$.\footnote{We show this statement in Lemma \ref{lem:3b} in Appendix \ref{sec:a.5}.}

\subsection{Motivating the monotonicity assumption}\label{subsec:3.2}
In this section, we illustrate that economic analysis implies the monotonicity assumption we employ for the multi-valued treatment case.
We consider Example I and use a discrete choice index model introduced in Section \ref{sec:2} to motivate the monotonicity assumption.
\cite{Mountjoy2019} also employs a discrete choice index model to motivate his monotonicity assumption for continuous instruments.
We take an approach similar to \cite{Mountjoy2019} for the choice index model.\footnote{\cite{Heckman2018} and \cite{Pinto2015} consider a general utility maximization problem and employ revealed preference analysis to motivate their unordered monotonicity assumption.
The generalized models also imply our monotonicity assumption.}

We consider a setting where individuals choose between not taking any postsecondary education or completing some postsecondary education and choose between $k$ different fields of study labeled as $1,\ldots,k$.
Let $Y$ denote observed earnings.
For the treatment $T$, let $T=0$ denote not taking any postsecondary education, and for $j\in\{1,\ldots,k\}$, $T=j$ denotes completing field $j$.
Suppose that individuals are randomly assigned to one of the following $(k+1)$ groups; individuals assigned to group $0$ have no cost reduction and for $j\in\{1,\ldots,k\}$, the cost of choosing field $j$ is decreased for individuals assigned to group $j$.
Let the instrument $Z$ represent group assignment that takes values on $\mZ=\{0,1,\ldots,k\}$, where $Z=j$ denotes assignment to group $j$.
In Example I, Assumption \ref{asp:mt} (i) holds because vouchers are randomly assigned.
Suppose that the treatment $T$ and the instrument $Z$ are sufficiently correlated, and we assume Assumptions \ref{asp:mt} (iii) and (iv) unless stated otherwise.

We first consider a case where individuals choose between three alternatives.
As in \cite{Mountjoy2019}, we assume additive separability of the utility functions in unobservable components for simplicity and ease of visualization.
As in \eqref{eq:chcf} in Section \ref{sec:2}, each individual chooses the treatment option with the highest indirect utility:
\begin{equation*}
T(z) = \argmax_{t\in\mT} I_t(z,V_t),
\end{equation*}
where the indirect utilities for each treatment option are defined as follows:
\[
I_0 = 0,\quad I_1 = V_1 - \mu_1(Z), \text{ and } I_2 = V_2 - \mu_2(Z).
\]
The utility of not taking any postsecondary education is normalized to zero.
$V_t$ is an individual's gross utility from choosing field $t$ and represents unobserved individual preference heterogeneity.
$\mu_t(Z)$ is the cost of choosing field $t$, and $I_t$ is the net utility of choosing field $t$.
The potential treatments under this choice index model are expressed as follows:
\begin{eqnarray}
D_0(z) &=& 1\{V_1 < \mu_1(z), V_2 < \mu_2(z)\}, \label{eq:te1} \\
D_1(z) &=& 1\{V_1 > \mu_1(Z), V_2 - V_1 < \mu_2(z) - \mu_1(z)\}, \label{eq:te2} \\
D_2(z) &=& 1\{V_2 > \mu_2(Z), V_2 - V_1 > \mu_2(z) - \mu_1(z)\}. \label{eq:te3}
\end{eqnarray}
Figure \ref{fig:ex1.1} (a) shows how these treatment choice equations \eqref{eq:te1}-\eqref{eq:te3} partition the two-dimensional space of unobserved preferences $(V_1,V_2)$.
Individuals who choose $T=0$ have low preferences for fields 1 and 2 relative to their costs, while those who choose $T=1$ or $T=2$ have higher preferences for their treatment choice.

For the cost functions, we can naturally assume the following relationships from the group assignment of the instrument:
\begin{eqnarray}
\mu_1(0) &=& \mu_1(2)>\mu_1(1), \label{eq:cf1} \\
\mu_2(0) &=& \mu_2(1)>\mu_2(2). \label{eq:cf2}
\end{eqnarray}
Relationship \eqref{eq:cf1} holds because the cost of choosing field 1 decreases for individuals assigned to group 1. 
Relationship \eqref{eq:cf2} holds for a similar reason.
Applying the restrictions on the cost functions \eqref{eq:cf1} and \eqref{eq:cf2} to the treatment choice equations \eqref{eq:te1}-\eqref{eq:te3} generates eight monotonicity inequalities summarized in Table \ref{tab:mtc1}.
\begin{table}[htb]
\caption{Monotonicity inequalities of Example I for the case of $k=2$}\label{tab:mtc1}
\centering
\begin{tabular}{|c| c c c c|} \hline
\multicolumn{1}{|c|}{} & \multicolumn{4}{|c|}{$\mT$} \\ \hline
\multicolumn{1}{|c|}{} & & $0$ & $1$ & $2$ \\ 
& $(1,0)$ & $\,\,\,\,\,\,\,\,D_0(1)\leq D_0(0)\,\,\,\,\,\,\,\,$ & $\,\,\,\,\,\,\,\,D_1(1)\textcolor{red}{\geq} D_1(0)\,\,\,\,\,\,\,\,$ & $\,\,\,\,\,\,\,\,D_2(1)\leq D_2(0)\,\,\,\,\,\,\,\,$ \\ 
$\mP$ & $(2,0)$ & $\,\,D_0(2)\leq D_0(0)\,\,$ & $\,\,D_1(2)\leq D_1(0)\,\,$ & $\,\,D_2(2)\textcolor{red}{\geq} D_2(0)\,\,$ \\ 
& $(1,2)$ & $\,\,\,\,$ & $\,\,D_1(1)\geq D_1(2)\,\,$ & $\,\,D_2(1)\leq D_2(2)\,\,$ \\ \hline
\end{tabular}
\end{table}
From Table \ref{tab:mtc1}, Assumption \ref{asp:mt} (ii) holds for 
$\Lambda=\{(1,0),(2,0)\}$.
For the inequalities in Table \ref{tab:mtc1}, \cite{kirkeboen2016field} also employ $D_1(1)\geq D_1(0)$ and $D_2(2)\geq D_2(0)$, and we obtain the other inequalities from the restrictions on the cost functions \eqref{eq:cf1} and \eqref{eq:cf2}.
As we discuss in Section \ref{sec:5}, \cite{Mountjoy2019} provides the same type of inequalities as the monotonicity inequalities for $(1,0)$ and $(2,0)$ with continuous instruments.

Figures \ref{fig:ex1.1} (b)--(d) visualize the monotonicity inequalities summarized in Table \ref{tab:mtc1}.
Figure \ref{fig:ex1.1} (b) illustrates how a shift in the instrument from $Z=0$ to $Z=1$ induces the three monotonicity inequalities for $(1,0)$.
Because this shift in $Z$ decreases the cost of choosing field 1 from $\mu_1(0)$ to $\mu_1(1)$ but does not change the cost of choosing field 2, some individuals find field 1 more attractive and change their choice from either $T=0$ or $T=2$ to $T=1$, but no individuals find field 1 less attractive and leave from the $T=1$ group.
Hence, the expansion of the $T=1$ region induces $D_1(1)\geq D_1(0)$ and, at the same time, the shrinkage of both the $T=0$ and $T=2$ regions induces $D_0(1)\leq D_0(0)$ and $D_2(1)\leq D_2(0)$.
Analogously, Figure \ref{fig:ex1.1} (c) illustrates that a shift in the instrument from $Z=0$ to $Z=2$ induces the three monotonicity inequalities for $(2,0)$.

Figure \ref{fig:ex1.1} (d) illustrates that a shift in the instrument from $Z=1$ to $Z=2$ does not induce any monotonicity inequality between $D_0(1)$ and $D_0(2)$, and $(1,2)$ is not contained in the monotonicity subset.
Because this shift in $Z$ increases the cost of choosing field 1 from $\mu_1(1)$ to $\mu_1(2)$ and decreases the cost of choosing field 2 from $\mu_2(1)$ to $\mu_2(2)$, some individuals who find field 1 less attractive move on to the $T=0$ group and, at the same time, some individuals who find field 2 more attractive leave from the $T=0$ group.
Hence, $D_0(1)>D_0(2)$ and $D_0(1)<D_0(2)$ can both happen depending on whether more individuals are induced into or out from the $T=0$ group.

\begin{figure}[htbp]
\begin{tabular}{cc}
\centering
\begin{minipage}[t]{0.45\hsize}
\centering
\scalebox{0.95}{
\begin{tikzpicture}
\coordinate(O)at(0,0);
\coordinate(XS)at(0,0);
\coordinate(XL)at(6,0);
\coordinate(YS)at(0,0);
\coordinate(YL)at(0,6);
\draw[->,>=stealth,semithick](XS)--(XL)node[below left]{$V_1$};
\draw[->,>=stealth,semithick](YS)--(YL)node[below left]{$V_2$};

\coordinate(P)at(3.5,3.5);
\draw[thick]($(XS)!(P)!(XL)$)node[below]{$\mu_1(z)$}--(P)--($(YS)!(P)!(YL)$)node[left]{$\mu_2(z)$};
\draw[thick,domain=3.5:6]plot(\x,\x)node[left]{$V_2-V_1=\mu_2(z)-\mu_1(z)$};

\draw(2,2)node{$D_0(z)=1$};
\draw(5,2)node{$D_1(z)=1$};
\draw(2,5)node{$D_2(z)=1$};
\end{tikzpicture}
}
\subcaption{Treatment choices for $Z=z$}
\end{minipage} &
\begin{minipage}[t]{0.45\hsize}
\centering
\scalebox{0.95}{
\begin{tikzpicture}
\coordinate(O)at(0,0);
\coordinate(XS)at(0,0);
\coordinate(XL)at(6,0);
\coordinate(YS)at(0,0);
\coordinate(YL)at(0,6);
\draw[->,>=stealth,semithick](XS)--(XL)node[below left]{$V_1$};
\draw[->,>=stealth,semithick](YS)--(YL)node[below left]{$V_2$};

\coordinate(P0)at(3.5,3.5);
\draw[dashed,thick]($(XS)!(P0)!(XL)$)node[below]{$\mu_1(0)$}--(P0)--($(YS)!(P0)!(YL)$);
\draw[dashed,thick,domain=3.5:6]plot(\x,\x);

\coordinate(P1)at(2,3.5);
\draw[thick]($(XS)!(P1)!(XL)$)node[below]{$\mu_1(1)$}--(P1)--($(YS)!(P1)!(YL)$)node[left]{$\mu_2(1)$};
\draw[thick,domain=2:4.5]plot(\x,\x+1.5)node[left]{$V_2-V_1=\mu_2(1)-\mu_1(1)$};

\draw(1,2)node{$D_0(1)=1$};
\draw(4.5,2)node{$D_1(1)=1$};
\draw(1.5,5)node{$D_2(1)=1$};
\draw[<-,>=stealth,thick](2.1,1)--(3.4,1);
\draw[<-,>=stealth,thick](3.1,4.5)--(4.4,4.5);
\end{tikzpicture}
}
\subcaption{Instrument shift from $Z=0$ to 1}
\end{minipage} \\
\begin{minipage}[t]{0.45\hsize}
\centering
\scalebox{0.95}{
\begin{tikzpicture}
\coordinate(O)at(0,0);
\coordinate(XS)at(0,0);
\coordinate(XL)at(6,0);
\coordinate(YS)at(0,0);
\coordinate(YL)at(0,6);
\draw[->,>=stealth,semithick](XS)--(XL)node[below left]{$V_1$};
\draw[->,>=stealth,semithick](YS)--(YL)node[below left]{$V_2$};

\coordinate(P0)at(3.5,3.5);
\draw[dashed,thick]($(XS)!(P0)!(XL)$)--(P0)--($(YS)!(P0)!(YL)$)node[left]{$\mu_2(0)$};
\draw[dashed,thick,domain=3.5:6]plot(\x,\x);

\coordinate(P2)at(3.5,2);
\draw[thick]($(XS)!(P2)!(XL)$)node[below]{$\mu_1(2)$}--(P2)--($(YS)!(P2)!(YL)$)node[left]{$\mu_2(2)$};
\draw[thick,domain=3.5:6]plot(\x,\x-1.5)node[above left]{$V_2-V_1=\mu_2(2)-\mu_1(2)$};

\draw(2,1.5)node{$D_0(2)=1$};
\draw(5,1.5)node{$D_1(2)=1$};
\draw(2,4.2)node{$D_2(2)=1$};
\draw[<-,>=stealth,thick](1,2.1)--(1,3.4);
\draw[<-,>=stealth,thick](4.5,3.1)--(4.5,4.4);
\end{tikzpicture}
}
\subcaption{Instrument shift from $Z=0$ to 2}
\end{minipage} &
\begin{minipage}[t]{0.45\hsize}
\centering
\scalebox{0.95}{
\begin{tikzpicture}
\coordinate(O)at(0,0);
\coordinate(XS)at(0,0);
\coordinate(XL)at(6,0);
\coordinate(YS)at(0,0);
\coordinate(YL)at(0,6);
\draw[->,>=stealth,semithick](XS)--(XL)node[below left]{$V_1$};
\draw[->,>=stealth,semithick](YS)--(YL)node[below left]{$V_2$};

\coordinate(P1)at(2,3.5);
\draw[dashed,thick]($(XS)!(P1)!(XL)$)node[below]{$\mu_1(1)$}--(P1)--($(YS)!(P1)!(YL)$)node[left]{$\mu_2(1)$};
\draw[dashed,thick,domain=2:4.5]plot(\x,\x+1.5);

\coordinate(P2)at(3.5,2);
\draw[thick]($(XS)!(P2)!(XL)$)node[below]{$\mu_1(2)$}--(P2)--($(YS)!(P2)!(YL)$)node[left]{$\mu_2(2)$};
\draw[thick,domain=3.5:6]plot(\x,\x-1.5);

\draw(1,1)node{$D_0(2)=1$};
\draw(5,1.5)node{$D_1(2)=1$};
\draw(1.5,4.5)node{$D_2(2)=1$};
\draw[->,>=stealth,thick](2.1,1)--(3.4,1);
\draw[<-,>=stealth,thick](1,2.1)--(1,3.4);
\end{tikzpicture}
}
\subcaption{Instrument shift from $Z=1$ to 2}
\end{minipage} 
\end{tabular}
\caption{Visualization of shifts in the instrument}\label{fig:ex1.1}
\end{figure}

We generalize the preceding argument to general $k\in\mT$.
Similar to the case of $k=2$, a discrete choice model with $k+1$ treatment options generates the following monotonicity inequalities:
\begin{equation}\label{eq:exi3}
D_{i}({i})\geq D_{i}({0})\text{ and }D_j({i})\leq D_j({0})\text{ for } i=1,\ldots,k\text{ and }j\in\mT\setminus\{ i\}.
\end{equation}
Table \ref{tab:mtc2} summarizes (\ref{eq:exi3}). 
From Table \ref{tab:mtc2}, Assumption \ref{asp:mt} (ii) holds for $\Lambda=\{(1,0),\ldots,(k,0)\}$.
\begin{table}[htb]
\caption{Monotonicity inequalities of Example I}\label{tab:mtc2}
\begin{center}
\scalebox{0.83}{
\begin{tabular}{|c|c c c c c c|} \hline
\multicolumn{1}{|c|}{} & \multicolumn{6}{|c|}{$\mT$} \\ \hline
\multicolumn{1}{|c|}{} & & $0$ & $1$ & $_{\cdots}$ & $k-1$ & $k$ \\ 
 &${(1,0)}$& $D_0({1})\leq D_0({0})$ & $D_1({1})\textcolor{red}{\geq} D_1({0})$ & $_{\cdots}$ &$D_{k-1}({1})\leq D_{k-1}({0})$ & $D_k({1})\leq D_k({0})$ \\ 
 $\mP$&$\vdots$ & $\vdots$ & $\vdots$ & $\vdots$ & $\vdots$ & $\vdots$ \\ 
 &${(k-1,0)}$& $D_0({k-1})\leq D_0({0})$ & $D_1({k-1})\leq D_1({0})$ & $_{\cdots}$ & $D_{k-1}({k-1})\textcolor{red}{\geq} D_{k-1}({0})$ & $D_k({k-1})\leq D_k({0})$ \\ 
 &${(k,0)}$& $D_0({k})\leq D_0({0})$ & $D_1({k})\leq D_1({0})$ & $_{\cdots}$ & $D_{k-1}({k})\leq D_{k-1}({0})$ & $D_k({k})\textcolor{red}{\geq} D_k({0})$ \\ \hline
\end{tabular}
}
\end{center}
\end{table}

%% file: identification.tex
\section{Identification}\label{sec:4}
In this section, we establish the identification of the potential outcome distributions and the local treatment effects using our monotonicity assumption.
Before establishing our main results, we first introduce a map termed ``counterfactual mapping.'' 
Counterfactual mapping, developed by \cite{VX2017}, is an essential tool for identification.
We then establish the identification of counterfactual mappings.

\subsection{Counterfactual mappings and our identification challenge}\label{subsec:2.2}
In this section, we introduce counterfactual mapping in multi-valued treatment settings. 
We show that identifying the potential outcome distributions follows from identifying the counterfactual mappings.
Our key identification challenge is to recover the counterfactual mappings in multi-valued treatment settings.

For $s,t\in\mT$, define $\phi_{s,t}:\R\to\R$ as $\phi_{s,t}(y):=Q_{Y_t}(F_{Y_s}(y))$. 
$\phi_{s,t}$ is called ``counterfactual'' mapping from $Y_s$ to $Y_t$ because the potential outcomes are also called ``counterfactual outcomes." 
\cite{VX2017} define a similar mapping for the binary treatment case. 
From the definition, this mapping relates the quantiles of the distribution of $Y_s$ to that of $Y_t$. 
Under Assumption \ref{asp:ivqr} (i), this mapping is strictly increasing on $\mY^\circ$, and $\phi_{s,r}=\phi_{t,r}\circ\phi_{s,t}$ holds for $s,t,r\in\mT$. 
Moreover, under Assumption \ref{asp:ivqr} (v), 
an inverse mapping $\phi_{s,t}^{-1}$ exists on $\mY^\circ$, and $\phi_{t,s}(y)=\phi_{s,t}^{-1}(y)$ holds for $y\in\mY^\circ$. 
Similarly, we define the unconditional counterfactual mapping for $s,t\in\mT$ as $\phi_{s,t}(y):=Q_{Y_t}(F_{Y_s}(y))$.

The following lemma shows that the potential outcome cdfs and means can be written as compositions of the counterfactual mappings and observable distributions. 
\cite{VX2017} show a similar result for the binary treatment case. 

\begin{lemma}[Potential outcome cdfs and means via counterfactual mappings]\label{lem:clfm}
Suppose that Assumption \ref{asp:ivqr} holds. Define $p_t(z)$ as in (\ref{eq:pt}).
Then, the following holds:
\begin{enumerate}
\item[(a)] For each $s\in\mT$, $F_{Y_s}(y)$ for $y\in\mY^\circ$ can be expressed as  
\begin{equation}\label{eq:clfm}
F_{Y_s}(y)=\sum_{t=0}^k F_{Y|TZ}(\phi_{s,t}(y)|t,z)p_t(z).
\end{equation}
\item[(b)] For each $s\in\mT$, $E[Y_s]$ can be expressed as
\begin{equation}\label{eq:clfmasf}
E[Y_s]=\sum_{t=0}^k E[\phi_{t,s}(Y)|T=t,Z=z]p_t(z).
\end{equation}
\end{enumerate}
\end{lemma}
\noindent Lemma \ref{lem:clfm} follows from the rank similarity assumption.
For each $s,t\in\mT$, the rank variables $U_s$ and $U_t$, and hence $Y_s$ and $\phi_{t,s}(Y_t)$, are identically distributed conditional on $(T,Z)=(t,z)$. 

Lemma \ref{lem:clfm} implies that for each $s\in\mT$, $E[Y_s]$ and $Q_{Y_s} (\tau )$ for $\tau\in(0,1)$ are identified as closed-form expressions if $\phi_{s,t}$ for $t\in\mT$ is also identified as a closed-form expression.
Hence, we establish sufficient conditions to identify $\phi_{s,t}$'s and derive the closed-form expressions of $\phi_{s,t}$'s.

\begin{remark}[]\label{rem:clfm}
We only need Lemma \ref{lem:clfm} (a) to identify $E[Y_s]$ and $Q_{Y_s} (\tau )$ for $\tau\in(0,1)$. 
The identification of $\phi_{s,t}(y)$ for $y\in\mY^\circ$ suffices for the identification of the treatment effects because we assume that $F_{Y_s} (y )$ is continuous in $y\in\mY$ in Assumption \ref{asp:ivqr} (i) and that the closure of $\mY^\circ$ is equal to $\mY$ in Assumption \ref{asp:ivqr} (v).
In Appendix \ref{sec:g}, we relax Assumption \ref{asp:ivqr} (v) and derive the closed-form expressions of $\phi_{s,t}$ and $F_{Y_s} (\cdot )$ on a sufficiently large subset of $\mY$.
\end{remark}

\subsection{Preliminary identification results}\label{subsec:4.0}
In this section, we introduce some basic identification results under our monotonicity assumption.
We first establish the identification of the compliers under our monotonicity assumption.
The following lemma shows that when $D_t(z)\leq D_t({z'})$ and $P(\mC^t_{z,z'})>0$ hold almost surely for $(z,z')\in\mP$ and $t\in\mT$, each probability of $\mC^t_{z,z'}$ and the conditional cdf of $Y_t$ given $\mC^t_{z,z'}$ are identified as closed-form expressions. 
\cite{Heckman2018} show a similar result under the unordered monotonicity assumption.

\begin{lemma}[Identification of the compliers]\label{lem:idcpl}
Assume that Assumption \ref{asp:mt} holds, and that $P(D_t(z)\leq D_t({z'}))=1$ and $P(\mC^t_{z,z'})>0$ hold for $(z,z')\in\mP$ and $t\in\mT$. Define $p_t(z)$ as in (\ref{eq:pt}). 
Then $P(\mC^t_{z,z'})$ and $F_{Y_t|\mC^t_{z,z'} }(y)$ for $y\in\mY$ are identified as
\begin{equation}\label{eq:15}
P(\mC^t_{z,z'})=p_t(z')-p_t(z),\qquad p_t(z')>p_t(z),
\end{equation}
and
\begin{equation}\label{eq:16}
F_{Y_t|\mC^t_{z,z'} }(y)=\frac{F_{Y|TZ}(y|t,z')p_t(z')-F_{Y|TZ}(y|t,z)p_t(z)}{p_t(z')-p_t(z)}.
\end{equation}
\end{lemma}

With Lemma \ref{lem:idcpl} at hand, we establish the identification of the counterfactual mappings.
To provide intuition, we first review the identification results for the binary treatment case by \cite{VX2017}. 
Let the treatment be binary so that $\mT=\{0,1\}$.
Observe that
\begin{equation}\label{eq:eqcpl}
\mC^1_{z,z'}=\mC^0_{z',z}
\end{equation}
holds from the definition when $T$ is binary.
Suppose Assumptions \ref{asp:ivqr} and \ref{asp:mt} hold with $P(D_1(z)\leq D_1({z'}))=1$ and $P(\mC^1_{z,z'})>0$.
Note that under the rank similarity assumption, the rank variables
$U_0$ and $U_1$ are identically distributed conditional on $\mC^1_{z,z'}$.\footnote{We show this statement in Lemma \ref{lem:3a} in Appendix \ref{sec:a.5}.} 
Furthermore, from the definition of the counterfactual mapping, if $y\in\mY^\circ$ is the $\tau\in(0,1)$ quantile of the distribution of $Y_1$, then $\phi_{1,0}(y)$ is the $\tau$ quantile of the distribution of $Y_0$.
Therefore, given (\ref{eq:eqcpl}), we obtain the following equation for the potential outcome conditional distributions given the compliers:
\begin{equation}\label{eq:cp01}
F_{Y_1|\mC^1_{z,z'}}(y)=F_{Y_0|\mC^0_{z,z'}}(\phi_{1,0}(y))\text{ for }y\in\mY^\circ.
\end{equation}
Finally, $F_{Y_0|\mC^0_{z,z'}}$ and $F_{Y_1|\mC^1_{z,z'}}$ are identified from Lemma \ref{lem:idcpl}, and $\phi_{1,0}(y)$ for $y\in\mY^\circ$ is identified as $\phi_{1,0}(y)=Q_{Y_0|\mC^0_{z,z'}}(F_{Y_1|\mC^1_{z,z'}}(y))$ by solving (\ref{eq:cp01}) for $\phi_{1,0}$.

In the case of multi-valued treatment settings, generally, we can compare no two treatment states on the same compliers as in (\ref{eq:cp01}).
This is because the relationships between the compliers become more complicated than (\ref{eq:eqcpl}), and the compliers do not generally coincide.
We overcome this difficulty by developing systems of relationships between two or more compliers for multiple treatment states that can be solved simultaneously for the counterfactual mappings.

\subsection{Identification in multi-valued treatment settings}\label{subsec:4.1}
In this section, we establish the identification of counterfactual mappings using the relationships of compliers when the treatment is multi-valued.
Throughout the identification analysis, we consider Example I for notation simplicity. The following argument does not rely on the settings of Example I.
We first consider a case where the treatment takes three values (then we have $\mT=\{0,1,2\}$), and assume Assumptions \ref{asp:ivqr} and \ref{asp:mt} hold for the subset $\Lambda$ of $\mP$.

We first introduce ``sign treatments" that characterize the type of monotonicity relationship of each pair $(z_1,z_2)$ contained in the monotonicity subset $\Lambda$. 
Suppose that, for $(z_1,z_2)\in\Lambda$, there uniquely exists $t{(z_1,z_2)}\in\mT$ such that either
\[
D_{t{(z_1,z_2)}}({z_1})\geq D_{t{(z_1,z_2)}}({z_2})\text{ and }
D_j({z_1})\leq D_j({z_2})\,\, \text{ for }j\in\mT\setminus\{t{(z_1,z_2)}\}
\]
or
\[
D_{t{(z_1,z_2)}}({z_1})\leq D_{t{(z_1,z_2)}}({z_2})\text{ and } D_j({z_1})\geq D_j({z_2})\,\, \text{ for }j\in\mT\setminus\{t{(z_1,z_2)}\}
\]
holds almost surely. Then we call this $t{(z_1,z_2)}$ ``sign treatment of $(z_1,z_2)$," and we state ``$(z_1,z_2)$ has a sign treatment" in this paper.

When the treatment takes three values, each $(z_1,z_2)\in\Lambda$ has a sign treatment. 
This is because if the monotonicity inequalities of $(z_1,z_2)$ are all the same, then no compliers exist, and $\Lambda$ cannot contain $(z_1,z_2)$. 
Suppose that $D_j({z_1})\leq D_j({z_2})$ holds almost surely for all $j=0,1,2$. 
Then, $D_0({z_1})\geq D_0({z_2})$ holds almost surely because $D_1({z_1})\leq D_1({z_2})$ and $D_2({z_1})\leq D_2({z_2})$ imply $1-D_0({z_1})\leq 1-D_0({z_2})$. 
Hence, $D_0({z_1})= D_0({z_2})$ holds almost surely. 
Applying a similar argument to treatment states 1 and 2 gives $D_j({z_1})=D_j({z_2})$ for each $j\in\mT$, which violates Assumption \ref{asp:mt} (iii).\footnote{When the treatment takes more than three values, each $(z_1,z_2)\in\Lambda$ may not have a sign treatment.
Let $\mT=\{0,1,2,3\}$, and suppose that the following monotonicity inequalities hold for $(z_1,z_2)$:
\[
D_{0}({z_1})\geq D_{0}({z_2}),\,\, D_{1}({z_1})\geq D_{1}({z_2}),\,\,
D_{2}({z_1})\leq D_{2}({z_2}),\text{ and }
D_3({z_1})\leq D_3({z_2}).
\]
Then, from the definition, $(z_1,z_2)$ does not have a sign treatment.}

With the sign treatments, we employ the following assumption and assume that two different types of monotonicity relationships exist.
\begin{assumption}[Existence of different types of monotonicity relationships]\label{asp:apb.2}
In the case of $\mT = \{0,1,2\}$, the monotonicity subset $\Lambda$ contains two pairs of instrument values $\lambda_1$ and $\lambda_2$ such that the 
following condition holds:
\begin{enumerate}
\item[] 
There uniquely exists $t({\lambda_1})\in\mT$ and $t({\lambda_2})\in\mT$ with $\lambda_i=(\lambda_{i,1},\lambda_{i,2})$ and $t({\lambda_1})\neq t({\lambda_2})$ such that 
\[
D_{t({\lambda_i})}({\lambda_{i,1}})\geq D_{t({\lambda_i})}({\lambda_{i,2}})\text{ and }
D_j({\lambda_{i,1}})\leq D_j({\lambda_{i,2}})\,\, \text{ for }j\in \mT\setminus\{t({\lambda_i})\}
\]
hold almost surely for $i=1,2$. 
\end{enumerate}
\end{assumption}
\noindent Under Assumption \ref{asp:apb.2}, the monotonicity subset $\Lambda$ contains two pairs of instrument values $\lambda_1$ and $\lambda_2$ such that each pair $\lambda_i$ has a sign treatment, and the two sign treatments $t({\lambda_1})$ and $t({\lambda_2})$ are different. 
Then, each $t({\lambda_i})$ characterizes a type of monotonicity relationship of $\lambda_i$. 
Assumption \ref{asp:apb.2} holds in Example I.
From Table \ref{tab:mtc1} in Section \ref{subsec:3.2}, $(1,0)$ and $(2,0)$ have sign treatments $t{(1,0)}=1$ and $t{(2,0)}=2$, respectively. 
Then, $(1,0)$ and $(2,0)$ induce different types of monotonicity relationships.

We can interpret Assumption \ref{asp:apb.2} as an instrument relevance condition that requires monotonic correlation between each of the two endogenous variables $D_{t({\lambda_1})}$ and $D_{t({\lambda_2})}$ and pairs of instrument values $\lambda_1$ and $\lambda_2$. 
We illustrate this point with Example I.
For $i=1,2$, whether the group assignment is $i$ or 0 produces a monotonic effect only toward the field choice $D_{t{(i,0)}}=D_i$. 
Compared with group 0, group $i$ additionally offers a discount only to field $i$, and only the preference for field $i$ is affected by the difference between these two group assignments. 

Different types of monotonicity relationships are essential for our identification analysis.
We obtain key relationships between the compliers essential for identifying the counterfactual mappings based on the two different monotonicity relationships.
For Example I, the proof of Lemma \ref{lem:apb} in Appendix \ref{sec:a} derives the following key relationships between the compliers:
\begin{eqnarray}
F_{Y_1|\mC^1_{0,1}}(\phi_{2,1}(y))
&=& \frac{F_{Y_0|\mC^0_{1,0}}(\phi_{2,0}(y))P(\mC^0_{1,0})+F_{Y_2|\mC^2_{1,0}}(y)P(\mC^2_{1,0})}{P(\mC^1_{0,1})}, \label{eq:ab1} \\
F_{Y_2|\mC^2_{0,2}}(y)
&=& \frac{F_{Y_0|\mC^0_{2,0}}(\phi_{2,0}(y))P(\mC^0_{2,0})+F_{Y_1|\mC^1_{2,0}}(\phi_{2,1}(y))P(\mC^1_{2,0})}{P(\mC^2_{0,2})}. \label{eq:ac1}
\end{eqnarray}
Relationships \eqref{eq:ab1} and \eqref{eq:ac1} correspond to (\ref{eq:cp01}) in the binary treatment case.
Applying Lemma \ref{lem:idcpl}, all the functions in \eqref{eq:ab1} and \eqref{eq:ac1} except for the counterfactual mappings are identified.

We identify the counterfactual mappings by solving \eqref{eq:ab1} and \eqref{eq:ac1} simultaneously for each $y\in\mY^\circ$.
When we fix the value of $y$ in \eqref{eq:ab1} and \eqref{eq:ac1} at any $y=y^{f}\in\mY^\circ$, $\phi_{2,1}(y^f)$ and $\phi_{2,0}(y^f)$ are the solutions to the following nonlinear simultaneous equations of two unknown variables $y_1$ and $y_0$:
\begin{eqnarray}
F_{Y_1|\mC^1_{0,1}}(y_1)
&=& \frac{F_{Y_0|\mC^0_{1,0}}(y_0)P(\mC^0_{1,0})+F_{Y_2|\mC^2_{1,0}}(y^f)P(\mC^2_{1,0})}{P(\mC^1_{0,1})}, \label{eq:ab1.f} \\
F_{Y_2|\mC^2_{0,2}}(y^f)
&=& \frac{F_{Y_0|\mC^0_{2,0}}(y_0)P(\mC^0_{2,0})+F_{Y_1|\mC^1_{2,0}}(y_1)P(\mC^1_{2,0})}{P(\mC^2_{0,2})}. \label{eq:ac1.f}
\end{eqnarray}
We have two equations to solve two unknowns, and $\phi_{2,1}(y^f)$ and $\phi_{2,0}(y^f)$ are identified if the solution is unique at $y^{f}\in\mY^\circ$.
The two equations (\ref{eq:ab1.f}) and (\ref{eq:ac1.f}) have a unique solution for $y_1$ and $y_0$ when the conditional cdfs of the potential outcome on compliers are strictly increasing on $\mY^\circ$.
Assumption \ref{asp:mt} (iv) implies the strict monotonicity of the conditional cdfs, and $\phi_{2,1}(y^f)$ and $\phi_{2,0}(y^f)$ are identified as follows:\footnote{See Appendix \ref{sec:a.75} for the derivation of \eqref{eq:ac3,4}.}
\begin{equation}\label{eq:ac3,4}  
\phi_{2,1}(y^f)=G_{1,2}^{y^f-1}(F_{Y_2|\mC^2_{0,2}}(y^f)) \text{ and } \phi_{2,0}(y^f)=\phi_{1,0}^{y^f}(\phi_{2,1}(y^f)),
\end{equation}
where $\phi_{1,0}^{y_f}$ with its domain $\mY^f\subset\mY$ is defined as
\begin{equation}\label{eq:A}
\phi_{1,0}^{y_f}(y)
:=Q_{Y_0|\mC^0_{1,0}}\left(\frac{F_{Y_1|\mC^1_{0,1}}(y)P(\mC^1_{0,1})-F_{Y_2|\mC^2_{1,0}}(y^f)P(\mC^2_{1,0})}{P(\mC^0_{1,0})}\right)\text{ for }y\in\mY^f,
\end{equation}
and we define a function $G_{1,2}^{y^f}$ as 
\begin{equation}\label{eq:46}
G_{1,2}^{y^f}(y):=\frac{F_{Y_0|\mC^0_{2,0}}(\phi_{1,0}^{y^f}(y))P(\mC^0_{2,0})+F_{Y_1|\mC^1_{2,0}}(y)P(\mC^1_{2,0})}{P(\mC^2_{0,2})}.
\end{equation}
Other counterfactual mappings on $\mY^\circ$ are inversions or compositions of $\phi_{2,1}$ and $\phi_{2,0}$, and they are also identified as closed-form expressions. 
The following lemma shows the identification of the counterfactual mappings under Assumption \ref{asp:apb.2}:
\begin{lemma}[Identification of counterfactual mappings from monotonicity]\label{lem:apb.2}
Suppose that Assumptions \ref{asp:ivqr}-\ref{asp:apb.2} hold for the case of $\mT = \{0,1,2\}$,. 
Then, $\phi_{s,t}(y)$ for $y\in\mY^\circ$ and $s,t\in\mT$ are identified.
\end{lemma}

We provide intuition for the identification of the counterfactual mappings.
First, the two relationships \eqref{eq:excpls} and \eqref{eq:excpls2} that are key conditions for identification are based on the following two relationships between compliers:
\begin{equation}\label{eq:excpls}
\mC^1_{0,1}=\mC^0_{1,0}\cup\mC^2_{1,0}\quad\text{and}\quad\mC^0_{1,0}\cap\mC^2_{1,0}=\emp,
\end{equation}
\begin{equation}\label{eq:excpls2}
\mC^2_{0,2}=\mC^0_{2,0}\cup\mC^1_{2,0}\quad\text{and}\quad\mC^0_{2,0}\cap\mC^1_{2,0}=\emp.
\end{equation}
Relationships \eqref{eq:excpls} and \eqref{eq:excpls2} correspond to \eqref{eq:eqcpl} in the binary treatment case.
We show that two monotonicity relationships for $(1,0)$ and $(2,0)$ generate \eqref{eq:excpls} and \eqref{eq:excpls2}. 

Figure \ref{fig:ex1.2} visualizes these compliers relationships \eqref{eq:excpls} and \eqref{eq:excpls2} implied by the separable index model introduced in Section \ref{subsec:3.2}.
Figure \ref{fig:ex1.2} (a) visualizes how the compliers are generated from
a shift in the instrument from $Z=0$ to $Z=1$.
First, because $\mC^1_{0,1}$ compliers are driven by individuals who find field 1 more attractive and leave from either $T=0$ group or $T=2$ group, we have
\begin{equation}\label{eq:19}
\mC^1_{0,1}=\{D_1({1})=1,D_0({0})=1\}\cup\{D_1({1})=1,D_2({0})=1\}.
\end{equation}
Obviously, the sets $\{D_1({1})=1,D_0({0})=1\}$ and $\{D_1({1})=1,D_2({0})=1\}$ are disjoint.
Second, we show that
\begin{equation}\label{eq:case1}
\mC^0_{1,0}=\{D_1({1})=1,D_0({0})=1\}\quad\text{and}\quad\mC^2_{1,0}=\{D_1({1})=1,D_2({0})=1\}.
\end{equation}
To see this, note that $\mC^0_{1,0}$ compliers, which consist of individuals leaving from the $T=0$ group, are entirely driven by those who find field 1 more attractive.
This is because, from the restrictions on the cost functions \eqref{eq:cf1} and \eqref{eq:cf2}, a shift in the instrument from $Z=0$ to $Z=1$ decreases the cost of choosing field 1 but does not change the cost of choosing field 2.
Analogously, $\mC^2_{1,0}$ compliers, which consist of individuals leaving from the $T=2$ group, are also entirely driven by those who find field 1 more attractive.
Therefore, \eqref{eq:excpls} holds from (\ref{eq:19}) and (\ref{eq:case1}).
By analogous logic, as visualized in Figure \ref{fig:ex1.2} (b), the compliers generated from a shift in the instrument from $Z=0$ to $Z=2$ satisfy relationship \eqref{eq:excpls2}.

\begin{figure}[htbp]
\centering
\begin{minipage}[b]{0.45\linewidth}
\centering
\scalebox{0.95}{
\begin{tikzpicture}
\coordinate(O)at(0,0);
\coordinate(XS)at(0,0);
\coordinate(XL)at(6,0);
\coordinate(YS)at(0,0);
\coordinate(YL)at(0,6);
\draw[->,>=stealth,semithick](XS)--(XL)node[below left]{$V_1$};
\draw[->,>=stealth,semithick](YS)--(YL)node[below left]{$V_2$};

\coordinate(P0)at(3.5,3.5);
\draw[thick]($(XS)!(P0)!(XL)$)node[below]{$\mu_1(0)$}--(P0)--($(YS)!(P0)!(YL)$);
\draw[thick,domain=3.5:6]plot(\x,\x);

\coordinate(P1)at(2,3.5);
\draw[thick]($(XS)!(P1)!(XL)$)node[below]{$\mu_1(1)$}--(P1)--($(YS)!(P1)!(YL)$)node[left]{$\mu_2(0)$};
\draw[thick,domain=2:4.5]plot(\x,\x+1.5);

\draw(1,2)node{$D_0(1)=1$};
\draw(4.5,2)node{$D_1(0)=1$};
\draw(1.5,5)node{$D_2(1)=1$};
\draw(2.7,1.5)node{$\mC^0_{1,0}$};
\draw(4.2,5)node{$\mC^2_{1,0}$};
\draw(4.9,3.1)node{$\mC^1_{0,1}$};
\draw[->,semithick](4.9,3.4)--(3.8,4.4);
\draw[->,semithick](4.9,2.8)--(3.1,2.3);
\end{tikzpicture}
}
\subcaption{Compliers between $Z=0$ and 1}
\end{minipage}
\begin{minipage}[b]{0.45\linewidth}
\centering
\scalebox{0.95}{
\begin{tikzpicture}
\coordinate(O)at(0,0);
\coordinate(XS)at(0,0);
\coordinate(XL)at(6,0);
\coordinate(YS)at(0,0);
\coordinate(YL)at(0,6);
\draw[->,>=stealth,semithick](XS)--(XL)node[below left]{$V_1$};
\draw[->,>=stealth,semithick](YS)--(YL)node[below left]{$V_2$};

\coordinate(P0)at(3.5,3.5);
\draw[thick]($(XS)!(P0)!(XL)$)--(P0)--($(YS)!(P0)!(YL)$)node[left]{$\mu_2(0)$};
\draw[thick,domain=3.5:6]plot(\x,\x);

\coordinate(P2)at(3.5,2);
\draw[thick]($(XS)!(P2)!(XL)$)node[below]{$\mu_1(0)$}--(P2)--($(YS)!(P2)!(YL)$)node[left]{$\mu_2(2)$};
\draw[thick,domain=3.5:6]plot(\x,\x-1.5);

\draw(2,1)node{$D_0(2)=1$};
\draw(5,1.5)node{$D_1(2)=1$};
\draw(2,5.3)node{$D_2(0)=1$};
\draw(2,2.7)node{$\mC^0_{2,0}$};
\draw(4.5,3.8)node{$\mC^1_{2,0}$};
\draw(2.5,4.5)node{$\mC^2_{0,2}$};
\draw[->,semithick](2.5,4.2)--(2,3.1);
\draw[->,semithick](2.9,4.5)--(4.5,4.2);
\end{tikzpicture}
}
\subcaption{Compliers between $Z=0$ and 2}
\end{minipage}
\caption{Visualization of the compliers}
\label{fig:ex1.2}
\end{figure}

Next, we provide an intuition for identifying $\phi_{2,1}(y^f)$ and $\phi_{2,0}(y^f)$ for each $y^{f}\in\mY^\circ$ as the unique solution to the two equations (\ref{eq:ab1.f}) and (\ref{eq:ac1.f}) of two unknowns $y_1$ and $y_0$.
We compare the set of (an infinite number of) solutions to each of these equations.
The two equations have a unique solution when these two sets intersect at a point.
Figure \ref{fig:ex1.3} visualizes the solutions to (\ref{eq:ab1.f}) and (\ref{eq:ac1.f}) respectively in the two-dimensional space of $(y_1,y_0)$.
First, the set of solutions to (\ref{eq:ab1.f}) form a strictly increasing relationship from the strict monotonicity of $F_{Y_1|\mC^1_{0,1}}$ and $F_{Y_0|\mC^0_{1,0}}$.
To see this, suppose that both $(y^-_1,y^-_0)$ and $(y^+_1,y^+_0)$ are solutions to (\ref{eq:ab1.f}).
If we have $y^-_1 < y^+_1$, we also have $F_{Y_1|\mC^1_{0,1}}(y^-_1)<F_{Y_1|\mC^1_{0,1}}(y^+_1)$ from the strict monotonicity of $F_{Y_1|\mC^1_{0,1}}$.
Then, because $F_{Y_1|\mC^1_{0,1}}$ and $F_{Y_0|\mC^0_{1,0}}$ are on the left and right-hand sides of (\ref{eq:ab1.f}), respectively, we need $F_{Y_0|\mC^0_{1,0}}(y^-_0)<F_{Y_0|\mC^0_{1,0}}(y^+_0)$ for (\ref{eq:ab1.f}) to hold, and this implies that $y^-_0 < y^+_0$ from the strict monotonicity of $F_{Y_0|\mC^0_{1,0}}$.
Next, from an analogous argument with the strict monotonicity of $F_{Y_1|\mC^1_{2,0}}$ and $F_{Y_0|\mC^0_{2,0}}$, the set of solutions to (\ref{eq:ac1.f}) form a strictly decreasing relationship, where both $F_{Y_1|\mC^1_{2,0}}$ and $F_{Y_0|\mC^0_{2,0}}$ are on the right-hand side of (\ref{eq:ac1.f}).
Hence, because strictly increasing and strictly decreasing relationships intersect only once, (\ref{eq:ab1.f}) and (\ref{eq:ac1.f}) have a unique solution at $(\phi_{2,1}(y^f),\phi_{2,0}(y^f))$.

\begin{figure}[htbp]
\centering
\scalebox{1.0}{
\begin{tikzpicture}
\coordinate(O)at(0,0);
\coordinate(XS)at(0,0);
\coordinate(XL)at(6,0);
\coordinate(YS)at(0,0);
\coordinate(YL)at(0,6);
\draw[->,>=stealth,semithick](XS)--(XL)node[below left]{$y_1$};
\draw[->,>=stealth,semithick](YS)--(YL)node[below left]{$y_0$};

\draw[red,thick,domain=0.5:5.5]plot(\x,\x)node[above]{Solutions to (\ref{eq:ab1.f})};
\draw[blue,thick,domain=0.5:5.5]plot(\x,-\x+6)node[above right]{Solutions to (\ref{eq:ac1.f})};

\coordinate(P)at(3,3);
\fill(P)circle(0.075);
\draw[dashed,thick]($(XS)!(P)!(XL)$)node[below]{$\phi_{2,1}(y^f)$}--(P)--($(YS)!(P)!(YL)$)node[left]{$\phi_{2,0}(y^f)$};
\coordinate(P-)at(2,2);
\draw[dashed,thick]($(XS)!(P-)!(XL)$)node[below]{$y^-_1$}--(P-)--($(YS)!(P-)!(YL)$)node[left]{$y^-_0$};
\coordinate(P+)at(4,4);
\draw[dashed,thick]($(XS)!(P+)!(XL)$)node[below]{$y^+_1$}--(P+)--($(YS)!(P+)!(YL)$)node[left]{$y^+_0$};
\end{tikzpicture}
}
\caption{Identification image when $k=2$}
\label{fig:ex1.3}
\end{figure}

Next, we generalize the above argument for identifying the counterfactual mappings to the case of arbitrary $k$.
With the sign treatments, we employ the following assumption and assume that $k$ different types of monotonicity relationships exist.
\begin{assumption}[Existence of different types of monotonicity relationships]\label{asp:apb}
The monotonicity subset $\Lambda$ contains $k$ pairs of instrument values $\lambda_1,\ldots,\lambda_k$ such that the 
following condition holds:
\begin{enumerate}
\item[] 
For $i=1,\ldots,k$, there uniquely exists $t({\lambda_i})\in\mT$ with $\lambda_i=(\lambda_{i,1},\lambda_{i,2})$ and $t({\lambda_i})\neq t({\lambda_j})$ for $i\neq j$ such that 
\[
D_{t({\lambda_i})}({\lambda_{i,1}})\geq D_{t({\lambda_i})}({\lambda_{i,2}})\text{ and }
D_j({\lambda_{i,1}})\leq D_j({\lambda_{i,2}})\,\, \text{ for }j\in \mT\setminus\{t({\lambda_i})\}
\]
hold almost surely. 
\end{enumerate}
\end{assumption}
\noindent Under Assumption \ref{asp:apb}, the monotonicity subset $\Lambda$ contains $k$ pairs of instrument values $\lambda_1,\ldots,\lambda_k$ such that each pair $\lambda_i$ has a sign treatment, and the sign treatments $t({\lambda_i})$'s are all different. 
Then, each $t({\lambda_i})$ characterizes a type of monotonicity relationship of $\lambda_i$. 
As in the case of $k=2$, Assumption \ref{asp:apb} holds in Example I.
From Table \ref{tab:mtc2} in Section \ref{subsec:3.2}, for $i=1,\ldots,k$, $(i,0)$ has a sign treatment $t{(i,0)}=i$. Then, for $i\neq j$, $(i,0)$ and $(j,0)$ induce different types of monotonicity relationships.

Under Assumption \ref{asp:apb}, we obtain key relationships between the compliers essential for identifying the counterfactual mappings.
For Example I, the proof of Lemma \ref{lem:apb} in Appendix \ref{sec:a} derives the following key relationships between the compliers:
\begin{equation}\label{eq:4.3.4}
F_{Y_i|\mC^i_{0,i}}(\phi_{k,i}(y))
=\frac{\sum_{j\neq i}F_{Y_j|\mC^j_{i,0}}(\phi_{k,j}(y))P(\mC^j_{i,0})}{P(\mC^i_{0,i})}\text{ for }y\in\mY^\circ\text{ and }i=1,\ldots,k.
\end{equation}
The $k$ relationships of \eqref{eq:4.3.4} correspond to \eqref{eq:ab1} and \eqref{eq:ac1} in the case of $k=2$.
Applying Lemma \ref{lem:idcpl}, all the functions in (\ref{eq:4.3.4}) except for the counterfactual mappings are identified.
When we fix (\ref{eq:4.3.4}) at $y^{f}\in\mY^\circ$, the $k$ equations of (\ref{eq:4.3.4}) constitute nonlinear simultaneous equations of $\phi_{k,i}(y^{f})$ for $j=0,\ldots,k-1$.
We identify these values by solving the $k$ equations of (\ref{eq:4.3.4}) simultaneously at $y^{f}$.

The following lemma shows the identification of the counterfactual mappings under Assumption \ref{asp:apb}:
\begin{lemma}[Identification of counterfactual mappings from monotonicity]\label{lem:apb}
Suppose that Assumptions \ref{asp:ivqr}-\ref{asp:apb} hold. Then, $\phi_{s,t}(y)$ for $y\in\mY^\circ$ and $s,t\in\mT$ are identified.
\end{lemma}
\begin{remark}[]\label{rem:apb}
In the proof of Lemma \ref{lem:apb} in Appendix \ref{sec:a}, we do not derive the closed-form expressions of $\phi_{s,t}$'s for identifying them. 
Appendix \ref{sec:a.25} provides the closed-form expressions of $\phi_{s,t}$'s for the general $k\in\mT$.
\end{remark}
\noindent As discussed in Section \ref{subsec:2.2}, $E[Y_s]$ and $Q_{Y_s} (\tau )$ for $\tau\in(0,1)$ are identified if $\phi_{s,t}(\cdot)$ for $t\in\mT$ are identified. 
Hence, we obtain the following theorem:
\begin{theorem}[Identification of potential outcome cdfs and means from monotonicity]\label{thm:apb}
Suppose that Assumptions \ref{asp:ivqr}-\ref{asp:apb} hold. Then, $E[Y_s]$ and $Q_{Y_s} (\tau )$ for $\tau\in(0,1)$ and $s\in\mT$ are identified.
\end{theorem}
This result is interesting because the proposed sufficient condition is economically interpretable.
We do not need to interpret the numerical conditions on the distribution of the outcome variable when this condition is satisfied.
This fact may be helpful when designing a social experiment for which the outcome data will be collected later.

\subsection{Comparison with Chernozhukov and Hansen (2005) and violation of our assumptions}\label{subsec:4.4}
In this section, we compare Assumption \ref{asp:apb} (or Assumption \ref{asp:apb.2} in the case of $k=2$) with the identification condition of \cite{CH2005} and discuss the case where Assumption \ref{asp:apb} is violated.
For the binary treatment case, the closed-form expression of \cite{Wuthrich2019} is also valid when the full rank condition in \cite{CH2005} holds for all the values in $\mY$.
\cite{VX2017} establish an identification condition weaker than the monotonicity assumption for the binary treatment case and show that the identification condition of \cite{CH2005} implies that condition.

In the multi-valued treatment setting, we can show that Assumption \ref{asp:apb} implies the full rank conditions in \cite{CH2005} under Assumptions \ref{asp:ivqr} and \ref{asp:mt} and some differentiability assumptions.\footnote{Appendix \ref{sec:e} proves this statement.}
On the other hand, our identification condition does not require any differentiability assumption.
Assumption \ref{asp:apb} may be violated even when the identification condition of \cite{CH2005} is satisfied.\footnote{Appendix \ref{sec:e.2} provides a numerical example where the full rank condition in \cite{CH2005} holds but Assumption \ref{asp:apb} is violated.}
The full rank condition in \cite{CH2005} only requires that the moment equations \eqref{eq:ti} are uniquely solved.
Assumption \ref{asp:apb} clarifies the behavioral patterns of each individual for identifying the treatment effects.
Assumption \ref{asp:apb} also implies additional testable restrictions compared with \cite{CH2005}.\footnote{Appendix \ref{sec:e.3} discusses these additional restrictions. These restrictions are equally difficult to check statistically compared with the full rank conditions in \cite{CH2005} for the binary treatment case.}

\begin{remark}\label{rem:4.4}
The full rank conditions in \cite{CH2005} requires the instrument to take the same number of values as the treatment, which is also required in our settings.
\cite{Feng2019} and \cite{caetano2020} establish the identification of treatment effects using observed covariates when the instrument has smaller support than the treatment.
Our results do not rely on the existence of observed covariates.
\end{remark}

Next, we discuss the case where Assumption \ref{asp:apb} is violated.
Even when the unordered monotonicity assumption holds, Assumption \ref{asp:apb} is violated if the sign treatments are the same for all the instrument pairs.
We consider the $k=2$ case and let the monotonicity inequalities summarized in Table \ref{tab:mtn} hold almost surely.
In Example I, this may happen when individuals are randomly assigned to one of the three groups, where the cost of choosing field $2$ is decreased for all groups; Group 2 has the largest decrease in cost, and Group 0 has the smallest.
\begin{table}[htb]
\caption{Monotonicity inequalities when sign treatments are all the same}
\label{tab:mtn}
\begin{center}
\begin{tabular}{|c|c c c c|} \hline
\multicolumn{1}{|c|}{} & \multicolumn{4}{|c|}{$\mT$} \\ \hline
\multicolumn{1}{|c|}{} & & $0$ & $1$ & $2$ \\ 
& $(1,0)$ & $\,\,\,\,\,\,\,\,D_0({1})\leq D_0({0})\,\,\,\,\,\,\,\,$ & $\,\,\,\,\,\,\,\,D_1({1})\leq D_1({0})\,\,\,\,\,\,\,\,$ & $\,\,\,\,\,\,\,\,D_2({1})\textcolor{red}{\geq} D_2({0})\,\,\,\,\,\,\,\,$ \\ 
$\mP$ & $(2,0)$ & $\,\,D_0({2})\leq D_0({0})\,\,$ & $\,\,D_1({2})\leq D_1({0})\,\,$ & $\,\,D_2({2})\textcolor{red}{\geq} D_2({0})\,\,$ \\ 
& $(2,1)$ & $\,\,D_0({2})\leq D_0({1})\,\,$ & $\,\,D_1({2})\leq D_1({1})\,\,$ & $\,\,D_2({2})\textcolor{red}{\geq} D_2({1})\,\,$ \\ \hline
\end{tabular}
\end{center}
\end{table}
Then, Assumptions 1 and 2 hold, and the unordered monotonicity assumption holds because Assumption 2 (ii) holds for all the instrument pairs.
However, the sign treatments of the instrument pairs are $t{(1,0)}=t{(2,0)}=t{(2,1)}=2$, and Assumption \ref{asp:apb.2} does not hold because all the instrument pairs induce the same type of monotonicity relationship.

Under this monotonicity assumption, we obtain the following two equations from the monotonicity relationships of $(1,0)$ and $(2,0)$ as we obtain (\ref{eq:ab1}) and (\ref{eq:ac1}) under Assumption \ref{asp:apb.2}:
\begin{eqnarray}
F_{Y_2|\mC^2_{0,1}}(y)
&=& \frac{F_{Y_0|\mC^0_{1,0}}(\phi_{2,0}(y))P(\mC^0_{1,0})+F_{Y_1|\mC^1_{1,0}}(\phi_{2,1}(y))P(\mC^1_{1,0})}{P(\mC^2_{0,1})}, \label{eq:all1} \\
F_{Y_2|\mC^2_{0,2}}(y)
&=& \frac{F_{Y_0|\mC^0_{2,0}}(\phi_{2,0}(y))P(\mC^0_{2,0})+F_{Y_1|\mC^1_{2,0}}(\phi_{2,1}(y))P(\mC^1_{2,0})}{P(\mC^2_{0,2})}. \label{eq:all2}
\end{eqnarray}
Whether $\phi_{2,1}(y)$ and $\phi_{2,0}(y)$ are identified from the two equations (\ref{eq:all1}) and (\ref{eq:all2}) depends on the numerical conditions on the conditional cdfs.
For each $y^f\in\mY^\circ$, $\phi_{2,1}(y^f)$ and $\phi_{2,0}(y^f)$ are the solutions to the following nonlinear simultaneous equations of two unknown variables $y_1$ and $y_0$:
\begin{eqnarray}
F_{Y_2|\mC^2_{0,1}}(y^f)
&=& \frac{F_{Y_0|\mC^0_{1,0}}(y_0)P(\mC^0_{1,0})+F_{Y_1|\mC^1_{1,0}}(y_1)P(\mC^1_{1,0})}{P(\mC^2_{0,1})}, \label{eq:all1.f} \\
F_{Y_2|\mC^2_{0,2}}(y^f)
&=& \frac{F_{Y_0|\mC^0_{2,0}}(y_0)P(\mC^0_{2,0})+F_{Y_1|\mC^1_{2,0}}(y_1)P(\mC^1_{2,0})}{P(\mC^2_{0,2})}. \label{eq:all2.f}
\end{eqnarray}
As in Section \ref{subsec:4.1}, we compare the set of solutions to each of these two equations, and Figure \ref{fig:ex1.4} visualizes the solutions to (\ref{eq:all1.f}) and (\ref{eq:all2.f}), respectively, in the two-dimensional space of $(y_1,y_0)$.
Notice that, from the strict monotonicity of the conditional cdfs, the solutions to (\ref{eq:all1.f}) and (\ref{eq:all2.f}) have both strictly decreasing relationships between $y_1$ and $y_0$, where both the conditional cdfs of $Y_1$ and $Y_0$ are on the same side of (\ref{eq:all1.f}) and (\ref{eq:all2.f}).
Hence, when these two strictly decreasing relationships intersect only once, 
the two equations (\ref{eq:all1.f}) and (\ref{eq:all2.f}) have a unique solution at $(\phi_{2,1}(y^f),\phi_{2,0}(y^f))$.
In this case, we can show that the full rank conditions in \cite{CH2005} hold under Assumptions \ref{asp:ivqr} and \ref{asp:mt} and some differentiability assumptions.\footnote{Appendix \ref{sec:e.2} proves this statement.}

\begin{figure}[htbp]
\centering
\scalebox{1.0}{
\begin{tikzpicture}
\coordinate(O)at(0,0);
\coordinate(XS)at(0,0);
\coordinate(XL)at(6,0);
\coordinate(YS)at(0,0);
\coordinate(YL)at(0,6);
\draw[->,>=stealth,semithick](XS)--(XL)node[below left]{$y_1$};
\draw[->,>=stealth,semithick](YS)--(YL)node[below left]{$y_0$};

\draw[red,thick,domain=1.25:4.75]plot(\x,-1.5*\x+7.5)node[above right]{Solutions to (\ref{eq:all1.f})};
\draw[blue,thick,domain=0.5:5.5]plot(\x,-0.5*\x+4.5)node[above right]{ Solutions to (\ref{eq:all2.f})};

\coordinate(P)at(3,3);
\fill(P)circle(0.075);
\draw[dashed,thick]($(XS)!(P)!(XL)$)node[below]{$\phi_{2,1}(y^f)$}--(P)--($(YS)!(P)!(YL)$)node[left]{$\phi_{2,0}(y^f)$};

\end{tikzpicture}
}
\caption{Identification image when sign treatments are all the same}
\label{fig:ex1.4}
\end{figure}

\subsection{Identification of the local treatment effects}\label{subsec:4.5}
In this section, we identify the local treatment effects in the multi-valued treatment setting.
Suppose that the monotonicity inequalities hold for $(z,z')$, and that $D_t(z)\leq D_t(z')$ holds almost surely for treatment level $t$. 
For two different treatment levels $t$ and $t'$ and instrument values $z$ and $z'$, the subpopulation $\{D_{t'}(z)=1,D_t(z')=1\}$ changes the treatment choice from $t'$ to $t$ if the instrument value changes from $z$ to $z'$.
The LATE that compares treatment states $t$ and $t'$ conditional on $\{D_{t'}(z)=1,D_t(z')=1\}$ is $E[Y_t|D_{t'}(z)=1,D_t(z')=1]-E[Y_{t'}|D_{t'}(z)=1,D_t(z')=1]$.
The local quantile treatment effect (LQTE) conditional on $\{D_{t'}(z)=1,D_t(z')=1\}$ is $Q_{Y_t|D_{t'}(z)=1,D_t(z')=1}(\tau)-Q_{Y_{t'}|D_{t'}(z)=1,D_t(z')=1}(\tau)$ where $\tau\in(0,1)$. 

First, we briefly review the identification studies of local treatment effects in multi-valued treatment settings.
From the definition, local treatment effects are treatment effects conditional on the compliers when the treatment is binary.
As shown in \cite{Imbens1994}, LATEs are identified as the standard two-stage least squares (2SLS) estimands under the monotonicity assumption.
\cite{Heckman2018} establish the identification of the conditional means of the potential outcomes on the compliers.
However, identifying the local treatment effects is more complex in multi-valued treatment settings because the relationship between the compliers becomes more complicated than in the binary treatment case.
\cite{Angrist1995} show that the 2SLS estimand generally represents only a weighted average of LATEs with a binary instrument, even under certain monotonicity assumption.\footnote{\cite{kline2016evaluating} and \cite{hull2018isolateing} also derive related results in the case of a binary instrument.}
\cite{kirkeboen2016field} and \cite{Mountjoy2019} demonstrate that a similar result holds for 2SLS estimands even when there are as many instruments as treatments. 
\cite{Heckman2006} identify the LATEs generalized to compare treatment state $t$ and the set of other states with continuous instruments.
\cite{Heckman2006} also establish identification on more various subpopulations that are identified with continuous instruments.
For a more general class of selection models,
\cite{lee2018} and \cite{Mountjoy2019} show similar identification results for the ATEs that compare two different treatment states, $t$ and $t'$.

This section establishes the identification of the LATEs and LQTEs that compare two different treatment states, $t$ and $t'$, under our monotonicity assumption.
We require that the outcome variable is continuously distributed with additional assumptions, such as the rank similarity assumption, necessary for identifying the counterfactual mappings.

Note that, if $(z,z')$ has a sign treatment $t_{(z,z')}=t'$, then $P(\mC_{z,z'}^t)=P(D_{t'}(z)=1,D_t(z')=1)$ holds for $t\in\mT\setminus\{t'\}$ as we obtain (\ref{eq:case1}) for Example I in Section \ref{subsec:4.1}.
Under the rank similarity assumption, identifying the counterfactual mappings leads to identifying the treatment effects conditional on the compliers. 
The following lemma shows the identification of the conditional distribution of $Y_{t'}$ given $\mC_{z,z'}^t$ as well as that of $Y_{t}$ given $\mC_{z,z'}^t$: 

\begin{lemma}[Identification of potential outcome conditional cdfs and means given the compliers]\label{lem:apb2}
Suppose that Assumptions \ref{asp:ivqr}-\ref{asp:apb} hold, and that $P(D_t(z)\leq D_t({z'}))=1$ and $P(\mC^t_{z,z'})>0$ hold for $t\in\mT$ and $(z,z')\in\mP$. 
Then, for all $t'\in\mT$, $F_{Y_{t'}|\mC_{z,z'}^t }(y)$ for $y\in\mY^\circ$ and $E[Y_{t'}|\mC_{z,z'}^t]$ can be expressed as
\begin{equation}\label{eq:l4.1}
F_{Y_{t'}|\mC_{z,z'}^t }(y)=F_{Y_{t}|\mC_{z,z'}^t }(\phi_{t',t}(y))
\end{equation}
and
\begin{equation}\label{eq:l4.2}
E[Y_{t'}|\mC_{z,z'}^t]=E[\phi_{t,t'}(Y_t)|\mC_{z,z'}^t],
\end{equation}
and $E[Y_{t'}|\mC_{z,z'}^t]$ and $Q_{Y_{t'}|\mC_{z,z'}^t } (\tau )$ for $\tau\in(0,1)$ are identified.
\end{lemma}

\noindent With Lemma \ref{lem:apb2} at hand, we obtain the following theorem that shows the identification of local treatment effects under our assumptions:

\begin{theorem}[Identification of local potential outcome cdfs and means]\label{thm:apb2}
Assume that Assumptions \ref{asp:ivqr}-\ref{asp:apb} hold. Then, for each $t,t'\in\mT$, there exists $(z,z')\in\mP$ such that  $E[Y_s|D_{t'}(z)=1,D_t(z')=1]$ and $Q_{Y_s|D_{t'}(z)=1,D_t(z')=1} (\tau )$ for $\tau\in(0,1)$ and $s\in\{t,t'\}$ are identified.
\end{theorem}

\subsection{Identification in real-world examples}\label{sec:5}
In this section, we discuss the content of our assumptions and the applicability of our identification result for Examples II and III in Section \ref{subsec:2.2ex}.
In each example, we confirm the assumptions for Theorems \ref{thm:apb} and \ref{thm:apb2} with monotonicity inequalities.
Suppose that the treatment $T$ and the instrument $Z$ are sufficiently correlated, and we assume Assumption \ref{asp:ivqr} and Assumptions \ref{asp:mt} (i), (iii), and (iv) unless stated otherwise.

\subsubsection{Example II}\label{subsec:3.3.2}
We take an approach similar to \cite{Mountjoy2019} for the model and economic analysis.
\cite{Mountjoy2019} also employs the discrete choice index model to motivate his monotonicity assumption for continuous instruments.
Let $Y$ denote student outcome.
Let the treatment $T$ denote college decision that takes values on $\{0,2,4\}$, where $T=0$ denotes no college, while $T=2$ and $T=4$ denote starting college at a two-year institution and a four-year institution, respectively.
Let the two binary instruments $Z_2$ and $Z_4$ that take values on $\{0,1\}$ represent the distance to the nearest college, and $Z_2$ equals 1 when the nearest two-year college is within a $d_2$ km radius, whereas $Z_4$ equals 1 when the nearest four-year college is within a $d_4$ km radius.
As discussed in Section \ref{subsec:3.2}, the discrete choice index model generates the monotonicity inequalities.
As in \cite{Mountjoy2019}, we define the indirect utilities for each treatment option as follows:
\begin{eqnarray*}
I_0 &=& 0, \\
I_2 &=& V_2 - \mu_2(Z_2), \\
I_4 &=& V_4 - \mu_4(Z_4),
\end{eqnarray*}
where we assume strict monotonicity for the cost functions as follows:
\begin{eqnarray}
\mu_2(0) >\mu_2(1) \quad\text{and}\quad \mu_4(0) >\mu_4(1). \label{eq:cft1} 
\end{eqnarray}
This requirement is natural because students who live near a college will have a smaller cost for choosing that college.
The strict monotonicity of the cost functions generates the monotonicity inequalities summarized in Table \ref{tab:tyc} for $z_2,z_4\in\{0,1\}$.
\begin{table}[htb]
\caption{Monotonicity inequalities of Example II}\label{tab:tyc}
\begin{center}
\scalebox{0.95}{
\begin{tabular}{|c|c c c c|} \hline
\multicolumn{1}{|c|}{} & \multicolumn{4}{|c|}{$\mT$} \\ \hline
\multicolumn{1}{|c|}{} & & $0$ & $2$ & $4$ \\ 
& $\{(1,z_4),(0,z_4)\}$ & $\,D_0(0,z_4)\leq D_0(1,z_4)\,$ & $\,D_2(0,z_4)\textcolor{red}{\geq} D_2(1,z_4)\,$ & $\,D_4(0,z_4)\leq D_4(1,z_4)\,$ \\ 
$\mP$ & $\{(z_2,1),(z_2,0)\}$ & $\,D_0(z_2,0)\leq D_0(z_2,1)\,$ & $\,D_2(z_2,0)\leq D_2(z_2,1)\,$ & $\,D_4(z_2,0)\textcolor{red}{\geq} D_4(z_2,1)\,$ \\ \hline
\end{tabular}
}
\end{center}
\end{table}
\cite{Mountjoy2019} provides the same type of inequalities as the monotonicity inequalities in Table \ref{tab:tyc} with continuous instruments.
The monotonicity inequalities in Table \ref{tab:tyc} with $z_2=z_4=0$ are the same as those for $(1,0)$ and $(2,0)$ in Example I in Section \ref{subsec:3.2} if we replace the values that $T$ and $Z$ take from $\{0,2,4\}$ to $\{0,1,2\}$ for $T$, and from $\{(0,0),(1,0),(0,1)\}$ to $\{0,1,2\}$ for $Z$.
Example II may have more monotonicity inequalities than Example I because $z_2$ and $z_4$ take values of either 0 or 1. 

From Table \ref{tab:tyc}, Assumption \ref{asp:mt} (ii) holds for $\{(1,z_4),(0,z_4)\},\{(z_2,1),(z_2,0)\}\in\Lambda$.
From Table \ref{tab:tyc}, $\{(0,z_4),(1,z_4)\}$ and $\{(z_2,0),(z_2,1)\}$ have different sign treatments $t{\{(0,z_4),(1,z_4)\}}=2$ and $t{\{(z_2,0),(z_2,1)\}}=4$, respectively. 
Therefore, Assumption \ref{asp:apb} holds, and from Theorems \ref{thm:apb} and \ref{thm:apb2}, the treatment effects are identified as closed-form expressions.

\subsubsection{Example III}\label{subsec:3.3.3}
We take an approach similar to \cite{Pinto2015} for the model and economic analysis.
Let $Y$ denote the outcome of interest that is continuously distributed.
Let the treatment $T$ denote the relocation decision at the intervention onset, where $T=0$ denotes no relocation, which is equivalent to choosing high poverty neighborhood, $T=1$ denotes medium-poverty neighborhood relocation, and $T=2$ denotes low-poverty neighborhood relocation. 
Let the instrument $Z$ represent voucher assignment that takes values on $\mZ=\{a,b,c\}$, where $Z=a$ denotes no voucher (control group), 
$Z=b$ denotes the Section 8 voucher, and $Z=c$ denotes the experimental voucher.

As discussed in Section \ref{subsec:3.2}, the discrete choice index model generates the monotonicity inequalities.
As in Example I, we assume the following relationships for the cost functions under additive separability of the utility functions:
\begin{eqnarray}
\mu_1(a) &=& \mu_1(c)>\mu_1(b), \label{eq:cfm1} \\
\mu_2(a) &>& \mu_2(b)=\mu_1(c). \label{eq:cfm2}
\end{eqnarray}
Relationships (\ref{eq:cfm1}) and (\ref{eq:cfm2}) can be interpreted in the same way as (\ref{eq:cf1}) and (\ref{eq:cf2}) in Example I.
Applying these restrictions to the cost functions generates the monotonicity inequalities summarized in Table \ref{tab:mto}.
\begin{table}[htb]
\caption{Monotonicity inequalities of Example III}\label{tab:mto}
\begin{center}
\begin{tabular}{|c|c c c c|} \hline
\multicolumn{1}{|c|}{} & \multicolumn{4}{|c|}{$\mT$} \\ \hline
\multicolumn{1}{|c|}{} & & $0$ & $1$ & $2$ \\ 
& $(c,a)$ & $\,\,\,\,\,\,\,\,D_0({c})\leq D_0({a})\,\,\,\,\,\,\,\,$ & $\,\,\,\,\,\,\,\,D_1({c})\leq D_1({a})\,\,\,\,\,\,\,\,$ & $\,\,\,\,\,\,\,\,D_2({c})\textcolor{red}{\geq} D_2({a})\,\,\,\,\,\,\,\,$ \\ 
$\mP$ & $(b,c)$ & $\,\,D_0({b})\leq D_0({c})\,\,$ & $\,\,D_1({b})\textcolor{red}{\geq} D_1({c})\,\,$ & $\,\,D_2({b})\leq D_2({c})\,\,$ \\ \hline
\end{tabular}
\end{center}
\end{table}
\cite{Pinto2015} provides similar relationships as (\ref{eq:cfm1}) and (\ref{eq:cfm2}) with budget sets and choice restrictions equivalent to the inequalities in Table \ref{tab:mto}.
From Table \ref{tab:mto}, Assumption \ref{asp:mt} (ii) holds for $\Lambda=\{(c,a),(b,c)\}$.\footnote{\cite{Pinto2015} further assumes that a neighborhood is a normal good and generates monotonicity inequalities in addition to those in Table \ref{tab:mto}. 
Assumptions of \cite{Pinto2015} lead to part (A2) of Assumption \ref{asp:mt} for all the pairs in $\mP$, and the unordered monotonicity assumption holds.
Our assumptions are weaker than those in \cite{Pinto2015} but are sufficient to identify the treatment effects.}
From Table \ref{tab:mto}, $(c,a)$ and $(b,c)$ have different sign treatments $t{(c,a)}=2$ and $t{(b,c)}=1$, respectively. 
Therefore, Assumption \ref{asp:apb} holds, and from Theorems \ref{thm:apb} and \ref{thm:apb2}, the treatment effects are identified as closed-form expressions.

%% file: conclusion.tex
\section{Conclusion}\label{sec:6}
In this paper, we establish sufficient conditions for the identification of the treatment effects when the treatment is discrete and endogenous. We show that an appropriately constructed monotonicity assumption is sufficient, and this condition is economically interpretable. We also derive closed-form expressions of the identified treatment effects.

For the estimation procedure, \cite{Wuthrich2019} constructs an estimator in the binary case by semiparametric estimation of observable conditional cdfs, qfs, and probabilities and plugging them into his closed-form expression. A similar approach could be applied to our closed-form expressions.

Alternatively, especially for the estimation of the QTE, we can apply the existing estimation methods under structural quantile models based on the GMM objective function after checking our identification conditions.
For estimation based on the GMM objective function, reliable and practically useful methods are developed, particularly for parametric structural quantile models. See \cite{CH2006}, \cite{ChenLee2018}, \cite{Zhu2018}, and \cite{KaidoWuthrich2018} for linear-in-parameters quantile models, and see \cite{ChernozhukovHong2003} and \cite{deCastroetal2019} for nonlinear quantile models. Nonparametric estimation approaches are studied by \cite{CIN2007}, \cite{HorowitzLee2007}, \cite{ChenPouzo2009,ChenPouzo2012}, and \cite{GS2012}.

%% file: proofs.tex
\section{Proofs of the results in the main text}\label{sec:a}
Proofs in this section use some auxiliary results (Lemmas \ref{lem:3b}-\ref{lem:3a}) collected in Appendix \ref{sec:a.5}.

\begin{proof}[Proof of Lemma \ref{lem:clfm}]
Observe that, for each $s\in\mT$, we have
\begin{equation}\label{eq:6}
\begin{split}
F_{Y_s}(y)=&F_{Y_s|Z}(y|z)
=\sum_{t=0}^k F_{Y_s|TZ}(y|t,z) p_t(z)\text{ for }y\in\mY^\circ
\end{split}
\end{equation}
and
\begin{equation}\label{eq:7}
\begin{split}
E[Y_s]=&E[Y_s|Z=z]
=\sum_{t=0}^k E[Y_s|T=t,Z=z] p_t(z),
\end{split}
\end{equation}
where the first equalities in (\ref{eq:6}) and (\ref{eq:7}) hold from Assumption \ref{asp:ivqr} (ii).

Take any $y^*\in\mY^\circ$. Then, there exists $\tau^*\in(0,1)$ such that $y^*=Q_{Y_s}(\tau^*)$ holds from Assumption \ref{asp:ivqr} (i). 
This $\tau^*$ can be expressed as $\tau^*=F_{Y_s}(y^*)$.

First, we show part (a). From \eqref{eq:6}, the required result (\ref{eq:clfm}) holds if we show
\begin{equation}\label{eq:32}
F_{Y_s|TZ}(y^*|t,z)=F_{Y_t|TZ}(\phi_{s,t}(y^*)|t,z)
\end{equation}
and
\begin{equation}\label{eq:32.5}
F_{Y_t|TZ}(\phi_{s,t}(y^*)|t,z)=F_{Y|TZX}(\phi_{s,t}(y^*)|t,z).
\end{equation}
\eqref{eq:32.5} holds because the observed outcome is identical with the potential outcome under treatment state $t$ when the treatment choice is $T=t$.
We proceed to show (\ref{eq:32}).
Under the rank similarity assumption, we have 
\begin{equation}\label{eq:29}
F_{U_{s}|TZ}(\tau^*|t,z)=F_{U_{t}|TZ}(\tau^*|t,z).
\end{equation}
Observe that the following equations hold:
\begin{equation}\label{eq:30}
\{U_s\leq\tau^*\}=\{F_{Y_s}(Y_s)\leq F_{Y_s}(y^*)\}=\{Y_s\leq y^*\}
\end{equation}
and
\begin{equation}\label{eq:31}
\{U_t\leq\tau^*\}=\{F_{Y_t}(Y_t)\leq F_{Y_t}(\phi_{s,t}(y^*))\}=\{Y_t\leq \phi_{s,t}(y^*)\}.
\end{equation}
The first equality in (\ref{eq:30}) holds from the definitions of $U_s$ and $\tau^*$. 
The first equality in (\ref{eq:31}) holds from the definitions of $U_t$, $\tau^*$, and $\phi_{s,t}$, as well as because $F_{Y_t}(Q_{Y_t}(\tau^*))=\tau^*$ holds from Assumption \ref{asp:ivqr} (i). 
The second equalities in (\ref{eq:30}) and (\ref{eq:31}) hold because $F_{Y_t}(y)$ for $t\in\mT$ are strictly increasing in $y\in\mY^\circ$ by Lemma \ref{lem:3b}. 
Therefore, applying (\ref{eq:30}) and (\ref{eq:31}) to (\ref{eq:29}) leads to (\ref{eq:32}), and (\ref{eq:clfm}) holds for $y\in\mY^\circ$.

Next, we show part (b). 
From \eqref{eq:7}, the required result (\ref{eq:clfmasf}) holds if we show the following equations:
\begin{equation}\label{eq:34}
E[Y_s|T=t,Z=z]=E[\phi_{t,s}(Y_t)|T=t,Z=z]
\end{equation}
and
\begin{equation}\label{eq:34.5}
E[\phi_{t,s}(Y_t)|T=t,Z=z]=E[\phi_{t,s}(Y)|T=t,Z=z].
\end{equation}
\eqref{eq:34.5} holds because the observed outcome is identical with the potential outcome under treatment state $t$ when the treatment choice is $T=t$.
We proceed to show (\ref{eq:34}). Observe that
\begin{equation}\label{eq:33}
\{U_t\leq \tau^*\}=\{Q_{Y_s}(U_t)\leq Q_{Y_s}(\tau^*)\}=\{\phi_{t,s}(Y_t)\leq y^*\}.
\end{equation}
The first equality in (\ref{eq:33}) holds because $Q_{Y_s}(\tau)$ is strictly increasing in $\tau\in(0,1)$ from Lemma \ref{lem:3c}. 
The second equality in (\ref{eq:33}) holds from the definition of $U_t$ and $\phi_{t,s}$.
Then, applying (\ref{eq:30}) and (\ref{eq:33}) to (\ref{eq:29}) leads to
\begin{equation}\label{eq:33.5}
F_{Y_s|TZ}(y^*|t,z)=F_{\phi_{t,s}(Y_t)|TZ}(y^*|t,z).
\end{equation}
Because $U_t\sim U(0,1)$, we have $\phi_{t,s}(Y_t)=Q_{Y_s}(U_t)\deq Y_s$, and $F_{\phi_{t,s}(Y_t)}(\cdot)$ is continuous. 
Hence, $F_{Y_s|TZ}(\cdot|t,z)$ and $F_{\phi_{t,s}(Y_t)|TZ}(\cdot|t,z)$ are also continuous. 
Then, from the assumption that the closure of $\mY^\circ$ is equal to $\mY$, we have 
\begin{equation}\label{eq:33.75}
F_{Y_s|TZ}(y|t,z)=F_{\phi_{t,s}(Y_t)|TZ}(y|t,z)\text{ for }y\in\mY.
\end{equation}
It follows that $Y_s\deq \phi_{t,s}(Y_t)$ conditional on $(T,Z)=(t,z)$; hence, we have (\ref{eq:34}), and (\ref{eq:clfmasf}) holds.
\end{proof}

\begin{proof}[Proof of Lemma \ref{lem:idcpl}]
Observe that we have  $\mC^t_{z,z'}=\{D_t({z'})=1\}\setminus\{D_t(z)=D_t({z'})=1\}$ from the definition of $\mC^t_{z,z'}$, and that $P(D_t(z)\leq D_t({z'}))=1$ implies $p_t(z)=P(D_t(z)=D_t({z'})=1)$. 
Hence, (\ref{eq:15}) holds and $P(\mC^t_{z,z'})>0$ implies $p_t(z')>p_t(z)$. 
For $y\in\mY$, an analogous argument gives
\begin{equation}\label{eq:36}
P(Y_t\leq y,\mC^t_{z,z'})=F_{Y|TZ}(y|t,z')p_t(z')-F_{Y|TZ}(y|t,z)p_t(z),
\end{equation}
and (\ref{eq:16}) follows from (\ref{eq:15}) and (\ref{eq:36}). 
\end{proof}

\begin{proof}[Proof of Lemmas \ref{lem:apb.2} and \ref{lem:apb}]
Lemma \ref{lem:apb.2} is the same as Lemma \ref{lem:apb} when $k=2$. 
We give a proof of Lemma \ref{lem:apb} for any $k$.
Suppose that the monotonicity subset $\Lambda\subset\mP$ contains $k$ pairs of instrument values $\lambda_1,\ldots,\lambda_k$ such that each sign treatment is $t({\lambda_i})=i$ for $i=1,\ldots,k$. 
For notation simplicity, let $\lambda_i=(i,0)$ and $D_{i}({i})\geq D_{i}({0})$ hold almost surely. 
Then, the monotonicity inequalities correspond to those of Example I in Section \ref{subsec:3.2}, and the types of monotonicity relationships on $(i,0)$ and $(j,0)$ for $i\neq j$ differ. 
The proof does not rely on this assumption, and we can similarly prove the lemma without this assumption.

It suffices to show that $\phi_{k,0},\ldots,\phi_{k,k-1}$ are identified on $\mY^\circ$ because the other counterfactual mappings are identified from $\phi_{k,0},\ldots,\phi_{k,k-1}$ under Assumption \ref{asp:ivqr} (v).
To see this, observe that $\phi_{s,t}^{-1}$ exists on $\mY^\circ$ and $\phi_{t,s}(y)=\phi_{s,t}^{-1}(y)$ holds for $y\in\mY^\circ$
because $\phi_{s,t}$ is strictly increasing on $\mY^\circ$ from Assumption \ref{asp:ivqr} (i) and Lemmas \ref{lem:3b} and \ref{lem:3c},
and $\phi_{s,t}(\mY^\circ)=\mY^\circ$ holds from
Assumption \ref{asp:ivqr} (v).
Furthermore, for $s,t,r\in\mT$, $\phi_{s,r}$ is identified on $\mY^\circ$ if $\phi_{s,t}$ and $\phi_{t,r}$ are identified on $\mY^\circ$
because $\phi_{s,t}(\mY^\circ)=\mY^\circ$ holds from Assumption \ref{asp:ivqr} (v),
and $\phi_{s,r}=\phi_{t,r}\circ\phi_{s,t}$ follows from $F_{Y_t}(Q_{Y_t}(\tau))=\tau$ for $\tau\in(0,1)$ by Assumption \ref{asp:ivqr} (i).
Therefore, the identification of $\phi_{k,0},\ldots,\phi_{k,k-1}$ suffices for the identification of all the counterfactual mappings.

We proceed to show that $\phi_{k,0},\ldots,\phi_{k,k-1}$ are identified on $\mY^\circ$.
We divide the proof into parts (i) and (ii).
Part (i) shows (\ref{eq:a.4.1}) and (\ref{eq:4.3.4}), and part (ii) shows that $\phi_{k,0}(y_k),\ldots,\phi_{k,k-1}(y_k)$ are identified in (\ref{eq:4.3.4}) for each $y_k\in\mY^\circ$.
We do not derive the closed-form expressions of the counterfactual mappings in this proof. 
See Appendices \ref{sec:a.75} and \ref{sec:a.25} for the derivation of the closed-form expressions of $\phi_{s,t}^x$'s for the case of $k=2$ and general $k\in\mT$, respectively.
\bigskip

\noindent\textit{Part (i).} In this part, for $(1,0),\ldots,(k,0)\in\Lambda$, we show that (\ref{eq:4.3.4}) hold.
We first show
\begin{equation}\label{eq:a.4.1}
P(U_i\leq\tau,\mC^i_{0,i})
=\sum_{j\neq i}P(U_i\leq\tau,\mC^j_{i,0})\text{ for }\tau\in(0,1)\text{ and }i=1,\ldots,k.
\end{equation}
Take any $\tau\in(0,1)$. Observe that, from the definition, we have
\begin{equation}\label{eq:f.1}
P(U_i\leq\tau,\mC^i_{0,i})=
\sum_{j\neq i}P(U_i\leq\tau,D_i({i})=1,D_j({0})=1)
\text{ for }i=1,\ldots,k
\end{equation}
and
\begin{equation}\label{eq:f.2}
P(U_i\leq\tau,\mC^j_{i,0})=
\sum_{l\neq j}P(U_i\leq\tau,D_j({0})=1,D_l({i})=1)
\text{ for }j\in\mT\setminus\{i\}.
\end{equation}
Note that, for $j,l\in\mT\setminus\{i\}$ and $j\neq l$, we have
\[
\begin{split}
P(U_i\leq\tau,D_j({0})=1,D_l({i})=1)
\leq&P(U_i\leq\tau,D_l({0})=0,D_l({i})=1)\\
\leq&P(D_l({0})=0,D_l({i})=1)=0
\end{split}
\]
because $\{D_j({0})=1,D_l({i})=1\}$ is contained in $\{D_l({0})=0,D_l({i})=1\}$, and $D_l({i})\leq D_l({0})$ holds almost surely from Assumption \ref{asp:apb}. 
Hence, we have
\begin{equation}\label{eq:f.3}
P(U_i\leq\tau,\mC^j_{i,0})=P(U_i\leq\tau,D_i({i})=1,D_j({0})=1)\text{ for }i=1,\ldots,k\text{ and }j\in\mT\setminus\{i\}.
\end{equation}
Therefore, (\ref{eq:a.4.1}) follows from (\ref{eq:f.3}) and (\ref{eq:f.1}).

With (\ref{eq:a.4.1}) at hand, we show (\ref{eq:4.3.4}).
Take any $y^*\in\mY^\circ$ and $i\in\{1,\ldots,k\}$. 
Then, from Assumption \ref{asp:ivqr} (i), there exists $\tau^*\in(0,1)$ such that $y^*=Q_{Y_i}(\tau^*)$ holds. 
Observe that 
\begin{equation}\label{eq:44.5}
F_{U_i|\mC^j_{i,0}}(\tau^*)=F_{U_j|\mC^j_{i,0}}(\tau^*)
\text{ for }j\in\mT\setminus\{i\}
\end{equation}
holds because $\{U_s\}_{s=0}^k$ are identically distributed conditional on each $\mC^j_{i,0}$ for $j\in\mT\setminus\{i\}$ from Lemma \ref{lem:3a}.
Then, from (\ref{eq:a.4.1}) and (\ref{eq:44.5}), we have
\begin{equation}\label{eq:44}
F_{U_i|\mC^i_{0,i}}(\tau^*)
=\frac{\sum_{j\neq i}F_{U_j|\mC^j_{i,0}}(\tau^*)P(\mC^j_{i,0})}{P(\mC^i_{0,i})}.
\end{equation}
From (\ref{eq:30}) and (\ref{eq:31}) in the proof of Lemma \ref{lem:clfm},  $\{U_t\leq\tau^*\}=\{Y_t\leq \phi_{k,t}(y^*)\}$ for $t\in\mT$ hold,
and hence (\ref{eq:4.3.4}) holds by applying this to (\ref{eq:44}).
\bigskip

\noindent\textit{Part (ii).} In this part, take any $y_k\in\mY^\circ$, and consider the following simultaneous equations of $(y_0,\ldots,y_{k-1})$:
\begin{equation}\label{eq:a.4.3}
F_{Y_i|\mC^i_{0,i}}(y_i)
=\frac{\sum_{j\neq i}F_{Y_j|\mC^j_{i,0}}(y_j)P(\mC^j_{i,0})}{P(\mC^i_{0,i})}\text
{ for }i=1,\ldots,k\text
{ and }y_i\in\mY^\circ.
\end{equation}
We show that $(y_0,\ldots,y_{k-1})=(\phi_{k,0}(y_{k}),\ldots,\phi_{k,k-1}(y_{k}))$ uniquely solves (\ref{eq:a.4.3}).\footnote{Note that $y\in\mY^\circ$ implies $\phi_{s,t}(y)\in\mY^\circ$ because $\phi_{s,t}(\mY^\circ)=\mY^\circ$ holds from Assumption \ref{asp:ivqr} (v).}
Then, $\phi_{k,0}(y),\ldots,\phi_{k,k-1}(y)$ are identified for $y\in\mY^\circ$ in (\ref{eq:4.3.4}) because all the functions in (\ref{eq:a.4.3}) are identified by Lemma \ref{lem:idcpl}.

Suppose we have a solution 
$(y_0',\ldots,y_{k-1}')$ different from
$(\phi_{k,0}(y_{k}),\ldots,\phi_{k,k-1}(y_{k}))$ that also satisfies (\ref{eq:a.4.3}).
Let $y_k'=y_k$ for notation simplicity.
We first consider the case of $y_0'<\phi_{k,0}(y_k)$. 
Then, there exists $j\in\mT\setminus\{0,k\}$ such that $y_j'>\phi_{k,j}(y_k)$ holds.
To see this, suppose that $y_j'\leq \phi_{k,j}(y_k)$ holds for all $j\in\mT\setminus\{0,k\}$.
Then, $(y_0',\ldots,y_{k-1}')$ cannot be a solution of (\ref{eq:a.4.3}) because substituting $(y_0',\ldots,y_{k-1}')$ to the right hand side of (\ref{eq:a.4.3}) with $i=k$ and noting that $F_{Y_0|\mC^0_{k,0}}$ is strictly increasing on $\mY^\circ$ gives
\begin{equation}\label{eq:a.4.1.1}
F_{Y_k|\mC^k_{0,k}}(y_k)
>\frac{\sum_{j\neq k}F_{Y_j|\mC^j_{k,0}}(y_j')P(\mC^j_{k,0})}{P(\mC^k_{0,k})}.
\end{equation}
Without loss of generality, suppose $j=k-1$, so that $y_{k-1}'>\phi_{k,k-1}(y_k)$ holds.
Because $\phi_{k,k-1}$ is strictly increasing on $\mY^\circ$, there exists $y^{(1)}>y_k$ such that $y_{k-1}'=\phi_{k,k-1}(y^{(1)})$ holds.
This implies that there exists $j\in\mT\setminus\{0,k-1,k\}$ such that $y_j'>\phi_{k,j}(y^{(1)})$ holds.
To see this, suppose that $y_j'\leq \phi_{k,j}(y^{(1)})$ holds for all $j\in\mT\setminus\{0,k-1,k\}$. 
Then, substituting $(y_0',\ldots,y_{k-2}')$ to the right hand side of (\ref{eq:a.4.3}) with $i=k-1$ and noting that $F_{Y_0|\mC^0_{k-1,0}}$ and $\phi_{k,0}$ are strictly increasing on $\mY^\circ$ gives
\begin{equation}\label{eq:a.4.1.1.5}
\frac{\sum_{j\neq k-1}F_{Y_j|\mC^j_{k-1,0}}(\phi_{k,j}(y^{(1)}))P(\mC^j_{k-1,0})}{P(\mC^{k-1}_{0,k-1})}
>\frac{\sum_{j\neq k-1}F_{Y_j|\mC^j_{k-1,0}}(y_j')P(\mC^j_{k-1,0})}{P(\mC^{k-1}_{0,k-1})}.
\end{equation}
From (\ref{eq:4.3.4}) with $i=k-1$ at $y^{(1)}$, the left hand side of \eqref{eq:a.4.1.1.5} equals $F_{Y_{k-1}|\mC^{k-1}_{0,k-1}}(\phi_{k,k-1}(y^{(1)}))$,
and $(y_0',\ldots,y_{k-1}')$ cannot be a solution of (\ref{eq:a.4.3}) because $y_{k-1}'=\phi_{k,k-1}(y^{(1)})$ holds.
Without loss of generality, suppose $j=k-2$, so that $y_{k-2}'>\phi_{k,k-2}(y^{(1)})$ holds.
Because $\phi_{k,k-2}$ is strictly increasing on $\mY^\circ$, there exists $y^{(2)}>y^{(1)}$ such that $y_{k-2}'=\phi_{k,k-2}(y^{(2)})$ holds.
Then, by repeating similar arguments, we can show from (\ref{eq:a.4.3}) with $i=2,\ldots,k$ that, without loss of generality, there exists 
$y_k<y^{(1)}<\cdots<y^{(k-2)}<y^{(k-1)}$
such that $y_{j}'=\phi_{k,j}(y^{(k-j)})$ holds for $j\in\mT\setminus\{0,k\}$.
Substituting $(y_0',y_2',\ldots,y_{k-1}')$ to the right hand side of (\ref{eq:a.4.3}) with $i=1$ and noting that $F_{Y_j|\mC^j_{1,0}}$ and $\phi_{k,j}$ are strictly increasing on $\mY^\circ$ gives
\begin{equation}\label{eq:a.4.1.2.5}
\frac{\sum_{j\neq 1}F_{Y_j|\mC^j_{1,0}}(\phi_{k,j}(y^{(k-1)}))P(\mC^j_{1,0})}{P(\mC^{1}_{0,1})}
>\frac{\sum_{j\neq 1}F_{Y_j|\mC^j_{1,0}}(y_j')P(\mC^j_{1,0})}{P(\mC^{1}_{0,1})}.
\end{equation}
From (\ref{eq:4.3.4}) with $i=1$ at $y^{(k-1)}$, the left hand side of (\ref{eq:a.4.1.2.5}) equals $F_{Y_{1}|\mC^{1}_{0,1}}(\phi_{k,1}(y^{(k-1)}))$.
Therefore, $(y_0',\ldots,y_{k-1}')$ cannot be a solution of (\ref{eq:a.4.3}) because $y_{1}'=\phi_{k,1}(y^{(k-1)})$ holds.

We can show contradiction similarly for other cases. For the case of $y_0'>\phi_{k,0}(y_k)$, we can show contradiction similarly by using reverse signs of inequality. For the case of $y_0'=\phi_{k,0}(y_k)$, consider the case of $y_j'\neq\phi_{k,j}(y_k)$ for some $j\in\mT\setminus\{0,k\}$, and we can show a contradiction in the same way as the case of $y_0'\neq\phi_{k,0}(y_k)$.

Therefore, $(y_0,\ldots,y_{k-1})=(\phi_{k,0}(y_{k}),\ldots,\phi_{k,k-1}(y_{k}))$ uniquely solves (\ref{eq:a.4.3}), and $\phi_{k,0}(y),\ldots,\phi_{k,k-1}(y)$ are identified for $y\in\mY^\circ$ in (\ref{eq:4.3.4}).
Other counterfactual mappings are also identified on $\mY^\circ$ because they are inversions or compositions of $\phi_{k,0},\ldots,\phi_{k,k-1}$.
\end{proof}

\begin{proof}[Proof of Theorem \ref{thm:apb}]
It suffices to show that, for each $s\in\mT$, the conditional distribution of $Y_s$ is identified.
Because $\phi_{s,t}(y)$ for $y\in\mY^\circ$ is identified from Lemma \ref{lem:apb}, $F_{Y_s}(y)$ for $y\in\mY^\circ$ is identified from (\ref{eq:clfm}) in Lemma \ref{lem:clfm}.
Then, from Assumption \ref{asp:ivqr} (v),
$F_{Y_s}(y)$ for $y\in\mY$ is identified, and the required result follows.
\end{proof}

\begin{proof}[Proof of Lemma \ref{lem:apb2}]
First, we show (\ref{eq:l4.1}) and (\ref{eq:l4.2}).
Because $U_{t}$ and $U_{t'}$ are identically distributed conditional on $\mC_{z,z'}^t$ by Lemma \ref{lem:3a}, we have
\begin{equation}\label{eq:l4.3}
F_{U_{t'}|\mC_{z,z'}^t}(\tau)=F_{U_{t}|\mC_{z,z'}^t}(\tau)\text{ for }\tau\in(0,1).
\end{equation}
Then, similar to the derivations of (\ref{eq:31}) and (\ref{eq:33})-(\ref{eq:33.75}) in the proof of Lemma \ref{lem:clfm}, we have
\begin{equation}\label{eq:l4.4}
\{U_t\leq\tau\}=\{Y_t\leq \phi_{t',t}(y)\}\text{ for }y\in\mY^\circ
\end{equation}
and 
\begin{equation}\label{eq:l4.5}
Y_{t'}\deq \phi_{t,t'}(Y_t)\text{ conditional on }\mC_{z,z'}^t.
\end{equation}
Applying (\ref{eq:l4.4}) to (\ref{eq:l4.3}) leads to (\ref{eq:l4.1}), and (\ref{eq:l4.2}) follows from (\ref{eq:l4.5}).

The proof is completed by showing that, for each $t'\in\mT$, $F_{Y_{t'}|\mC_{z,z'}^t}(y)$ for $y\in\mY$ is identified.
Because $\phi_{t',t}(y)$ for $y\in\mY^\circ$ is identified from Lemma \ref{lem:apb}, $F_{Y_{t'}|\mC_{z,z'}^t}(y)$ for $y\in\mY^\circ$ is identified from (\ref{eq:l4.1}).
Then, from Assumption \ref{asp:ivqr} (v),
$F_{Y_{t'}|\mC_{z,z'}^t}(y)$ for $y\in\mY$ is identified, and the stated result follows.
\end{proof}

\begin{proof}[Proof of Theorem \ref{thm:apb2}]
From Assumption \ref{asp:apb}, for each $t,t'\in\mT$, there exists $(z,z')\in\mP$ such that $t{(z,z')}\in\{t,t'\}$. Let $t{(z,z')}=t'$ without loss of generality. 
From (\ref{eq:f.3}) in the proof of Lemma \ref{lem:apb}, we have $P(\mC_{z,z'}^t)=P(D_{t'}(z)=1,D_{t}(z')=1)$. Therefore, the stated result follows from Lemma \ref{lem:apb2}.
\end{proof}

%% file: otherproofs.tex
\section{Derivation of (\ref{eq:ac3,4})}\label{sec:a.75}
In this section, we derive \eqref{eq:ac3,4} as the unique solution to (\ref{eq:ab1.f}) and (\ref{eq:ac1.f}).
We divide the proof into parts (i) and (ii).
The proof in this section uses an auxiliary result (Lemma \ref{lem:3b}) introduced in Appendix \ref{sec:a.5}.
\bigskip

\noindent\textit{Part (i).} 
Define $\phi_{1,0}^{y_f}$ as in (\ref{eq:A}).
In this part, we show that $\phi_{1,0}^{y_f}$ satisfies
\begin{equation}\label{eq:ac3}  
\phi_{2,0}(y^f)=\phi_{1,0}^{y^f}(\phi_{2,1}(y^f)).
\end{equation}
Observe that $\phi_{1,0}^{y_f}$ is the identified function that satisfies
\begin{equation}\label{eq:ac4}
F_{Y_1|\mC^1_{0,1}}(y)
=\frac{F_{Y_0|\mC^0_{1,0}}(\phi_{1,0}^{y_f}(y))P(\mC^0_{1,0})+F_{Y_2|\mC^2_{1,0}}(y^f)P(\mC^2_{1,0})}{P(\mC^1_{1,0})}\text{ for }y\in\mY^f.
\end{equation}
This is because $F_{Y_0|\mC^0_{1,0}}$ is continuous on $\mY$ from Assumptions \ref{asp:ivqr} and \ref{asp:mt} and Lemma \ref{lem:idcpl},
and $F_{Y_0|\mC^0_{1,0}}(Q_{Y_0|\mC^0_{1,0}}(\tau))=\tau$ holds for $\tau\in(0,1)$.
From (\ref{eq:ab1.f}) and (\ref{eq:ac4}) at $\phi_{2,1}(y^f)$, we have 
\begin{equation}\label{eq:B}  
F_{Y_0|\mC^0_{1,0}}(\phi_{2,0}(y^f))=F_{Y_0|\mC^0_{1,0}}(\phi_{1,0}^{y^f}(\phi_{2,1}(y^f))).
\end{equation}
Observe that $\phi_{2,0}(y^f)$ is contained in $\mY^\circ$ because $\phi_{2,0}(\mY^\circ)=\mY^\circ$ holds from Assumption \ref{asp:ivqr} (v), and
$F_{Y_0|\mC^0_{1,0}}$ is strictly increasing on $\mY^\circ$ from Lemma \ref{lem:3b}.
Therefore, (\ref{eq:ac3}) holds by taking the inverse of $F_{Y_0|\mC^0_{1,0}}$ in (\ref{eq:B}).

\bigskip
\noindent\textit{Part (ii).} In this part, we show that (\ref{eq:ac3,4}) is the unique solution to (\ref{eq:ab1.f}) and (\ref{eq:ac1.f}).
First, we plug in (\ref{eq:ac3}) to (\ref{eq:ac1.f}) and obtain
\begin{equation}\label{eq:ab3}
F_{Y_2|\mC^2_{0,2}}(y^f)=\frac{F_{Y_0|\mC^0_{2,0}}(\phi_{1,0}^{y^f}(\phi_{2,1}(y^f)))P(\mC^0_{2,0})+F_{Y_1|\mC^1_{2,0}}(\phi_{2,1}(y^f))P(\mC^1_{2,0})}{P(\mC^2_{0,2})}.
\end{equation}
Using $G_{1,2}^{y^f}(\cdot)$ defined in \ref{eq:46}, we can write (\ref{eq:ab3}) as $F_{Y_2|\mC^2_{(0,2)}}(y^f)=G_{1,2}^{y^f}(\phi_{2,1}(y^f))$.
Observe that $G_{1,2}^{y^f}$ is strictly increasing in $\mY^{f}\cap\mY^{\circ}$ because $F_{Y_0|\mC^0_{2,0}}$, $F_{Y_1|\mC^1_{2,0}}$, and $\phi_{1,0}^{y^f}$ are strictly increasing on $\mY^{f}\cap\mY^{\circ}$. $\phi_{2,1}(y^f)$ is contained in $\mY^{f}\cap\mY^{\circ}$ from Assumption \ref{asp:ivqr} (v), and we can solve (\ref{eq:ab3}) for $\phi_{2,1}(y^f)$ by taking the inverse of $G_{1,2}^{y^f}$.
Hence, $\phi_{2,1}(y)$ is identified at each $y^f\in\mY^\circ$ as in (\ref{eq:ac3,4}).

\section{Auxiliary results}\label{sec:a.5}
The following lemmas are used in the proofs in Appendices \ref{sec:a} and \ref{sec:a.75}.

\begin{lemma}[Strict monotonicity on the interior of the support]\label{lem:3b}
Let $W$ be a scalar-valued random variable whose support is $\mW$. Then, $F_W$ is strictly increasing on $\mW^\circ$.
\end{lemma}

\begin{proof}[Proof of Lemma \ref{lem:3b}]
Because $\mW$ is the support of $W$, 
we have 
\begin{equation}\label{eq:3.b.1}
\mW=\{w\in\R:F_W(w+\eps)-F_W(w-\eps)>0\text{ for all }\eps>0\}.
\end{equation}
Consider $w_1,w_2\in\mW^\circ$ with $w_1<w_2$. Then, there exists $\delta>0$ such that $\eta\leq\delta\To w_1+\eta\in\mW$ holds. 
First, suppose $w_2-w_1\leq\delta$. Then, because $(w_1+w_2)/2\in\mW$, we have 
$F_W(w_1)<F_W(w_2)$ from (\ref{eq:3.b.1}).
Second, suppose $w_2-w_1>\delta$. Then, because $w_1+\delta/2\in\mW$, we have 
$F_W(w_1)<F_W(w_1+\delta)$ from (\ref{eq:3.b.1}). Hence, we have 
$F_W(w_1)<F_W(w_2)$ because 
$w_1+\delta<w_2$ implies 
$F_W(w_1+\delta)\leq F_W(w_2)$.
Therefore, the stated result follows. 
\end{proof}

\begin{lemma}[Strict monotonicity of the qf]\label{lem:3c}
Let $W$ be a scalar-valued random variable. Assume that $F_W$ is continuous. Then $Q_W$ is strictly increasing on $(0,1)$.
\end{lemma}

\begin{proof}[Proof of Lemma \ref{lem:3c}]
Consider $\tau_1,\tau_2\in(0,1)$ with $\tau_1<\tau_2$. Suppose that $Q_W(\tau_1)=Q_W(\tau_2)$ holds.
Because $F_W$ is continuous, $F_W(W)\sim U(0,1)$ holds.
Then, from $Q_W(\tau_1)\leq w\toto\tau\leq F_W(w)$ for $\tau\in(0,1)$ and $w\in\R$, we have 
\[
1-\tau_i=P(F_W(W)\geq\tau_i)=P(W\geq Q_W(\tau_i))\text{ for }i=1,2.
\]
Hence, we have $\tau_1=\tau_2$, which is a contradiction.
Therefore, the stated result follows. 
\end{proof}

\begin{lemma}[Rank similarity on the compliers]\label{lem:3a}
Suppose that Assumptions \ref{asp:ivqr} and \ref{asp:mt} hold, and that $P(D_t(z)\leq D_t({z'}))=1$ and $P(\mC^t_{z,z'})>0$ hold for $(z,z')\in\mP$ and $t\in\mT$. 
Then, $\{U_s\}_{s=0}^k$ are identically distributed conditional on $\mC^t_{z,z'}$.
\end{lemma}

\begin{proof}[Proof of Lemma \ref{lem:3a}]
As we show (\ref{eq:16}) of Lemma \ref{lem:idcpl}, we can show that for $\tau\in(0,1)$ and $t'\in\mT$,
\begin{equation}\label{eq:16a}
F_{U_{t'}|\mC^t_{z,z'} }(\tau)
=\frac{F_{U_{t'}|TZ}(\tau|t,z')p_t(z')-F_{U_{t'}|TZ}(\tau|t,z)p_t(z)}{p_t(z')-p_t(z)}
\end{equation}
holds. Under rank similarity, for $z^+\in\mZ$, we have
\begin{equation}\label{eq:16b}
F_{U_{t'}|TZ}(\tau|t,z^+)=F_{U_{t}|TZ}(\tau|t,z^+).
\end{equation}
Combining (\ref{eq:16a}) with (\ref{eq:16b}) leads to $F_{U_{t'}|\mC^t_{z,z'} }(\tau)=F_{U_{t}|\mC^t_{z,z'} }(\tau)$, and the stated result follows. 
\end{proof}

%% file: abstract_supp.tex
This Supplemental Appendix is organized as follows. Appendix \ref{sec:a.25} derives the closed-form expressions of the counterfactual mappings for the general discrete treatment case.
Appendix \ref{sec:e} shows that Assumption \ref{asp:apb} implies the full rank conditions in \cite{CH2005} under Assumptions \ref{asp:ivqr} and \ref{asp:mt} and some differentiability assumptions.
Appendix \ref{sec:e.2} discusses the case where Assumption \ref{asp:apb.2} is violated, although the full rank conditions in \cite{CH2005} hold for the $k=2$ case. We provide a numerical example for this case.
Appendix \ref{sec:e.3} provides additional testable restrictions compared with \cite{CH2005} under our assumptions.
Appendix \ref{sec:g} relaxes Assumption \ref{asp:ivqr} (v), which is assumed for simplicity in the main paper, and precisely derives the closed-form expressions of the treatment effects. 

%% file: closedform.tex
\section{Closed-form expressions of the counterfactual mappings in Lemma \ref{lem:apb}}\label{sec:a.25}
In this section, we derive the closed-form expressions of $\phi_{s,t}$'s in Lemma \ref{lem:apb} for general $k\in\mT$.
Take any $y^{f,k}\in\mY^\circ$. Define $G_{k-1,k}^{y^f}(y)$ for $y\in\mY^f$ as
\begin{equation}\label{eq:a.4.10}
G_{k-1,k}^{y^f}(y):=\frac{\sum_{j=0}^{k-2}F_{Y_j|\mC^j_{k,0}}(\phi_{k-1,j}^{y^f}(y))P(\mC^j_{k,0})
+F_{Y_{k-1}|\mC^{k-1}_{k,0}}(y)P(\mC^{k-1}_{k,0})}{P(\mC^h_{0,k})},
\end{equation}
where $\phi_{k-1,j}^{y_f}$ for $j=0,\ldots,k-2$ with their domain $\mY^f\subset\mY$ are constructed to satisfy the following equations with $\phi_{k-1,k-1}^{y^f}(y)=y$:\footnote{The domain $\mY^f$ contains $\phi_{k,k-1}(y^{f,k})$. When $k=2$, (\ref{eq:ac4}) and (\ref{eq:4.3.10}) are the same.}
\begin{equation}\label{eq:4.3.10}
\begin{split}
&F_{Y_i|\mC^i_{0,i}}(\phi_{k-1,i}^{y^f}(y))\\
=&\frac{\sum_{j\neq i,k,k-1}F_{Y_j|\mC^j_{i,0}}(\phi_{k-1,j}^{y^f}(y))P(\mC^j_{i,0})+F_{Y_{k-1}|\mC^{k-1}_{i,0}}(y)P(\mC^{k-1}_{i,0})+F_{Y_k|\mC^k_{i,0}}(y^{f,k})P(\mC^k_{i,0})}{P(\mC^i_{0,i})}\\
\text{ for }&i=1,\ldots,k-1\text{ and }y\in\mY^f.
\end{split}
\end{equation}
The closed-form expressions of $\phi_{k,j}(y^{f,k})$ for $j=0,\ldots,k-1$ are derived from (\ref{eq:4.3.4}) as follows:
\begin{equation}\label{eq:4.3.12}
\begin{split}
\phi_{k,k-1}(y^{f,k})
=&\sup\left\{y\in\mY^f:G_{k-1,k}^{y^f}(y)\leq F_{Y_k|\mC^h_{0,k}}(y^{f,k})\right\}\\
=&\inf\left\{y\in\mY^f:G_{k-1,k}^{y^f}(y)\geq F_{Y_k|\mC^h_{0,k}}(y^{f,k})\right\},
\end{split}
\end{equation}
\begin{equation}\label{eq:4.3.11}  
\phi_{k,j}(y^{f,k})=\phi_{k-1,j}^{y^f}(\phi_{k,k-1}(y^{f,k}))\text{ for }j=0,\ldots,k-2.
\end{equation}
When $k=2$, \eqref{eq:a.4.10}-\eqref{eq:4.3.11} are the same as \eqref{eq:46}, \eqref{eq:ac4}, and the two expressions of \eqref{eq:ac3,4}, respectively.
Note that the first expression of \eqref{eq:ac3,4} can be expressed in two ways as in \eqref{eq:4.3.12}.
This is because $G_{1,2}^{y^f}(y)$, defined in \eqref{eq:46}, is increasing in $y\in\mY^f$ when $k=2$.
$\phi_{k,k-1}(y^{f,k})$ has these two expressions for general $k\in\mT$ because $G_{k-1,k}^{y^f}(y)$ is not larger (or smaller) than $F_{Y_k|\mC^h_{0,k}}(y^{f,k})$ when $y$ is smaller (or larger) than $\phi_{k,k-1}(y^{f,k})$.
This property simplifies the process of finding $\phi_{k,k-1}(y^{f,k})$ from $G_{k-1,k}^{y^f}$.
As discussed in proof of Lemma \ref{asp:apb} in Appendix \ref{sec:a},
other counterfactual mappings are also identified as closed-form expressions on $\mY^\circ$ because they are inversions or compositions of $\phi_{k,0},\ldots,\phi_{k,k-1}$.

We proceed to show that (\ref{eq:4.3.12}) and (\ref{eq:4.3.11}) are the unique solutions to (\ref{eq:4.3.4}) at $y^{f,k}$. 
We divide the proof into parts (i)-(iii). 
Part (i) constructs $\phi_{k-1,j}^{y_f}$ for $j=0,\ldots,k-2$ that satisfy (\ref{eq:4.3.10}),
part (ii) shows that $\phi_{k-1,j}^{y_f}$ for $j=0,\ldots,k-2$ satisfy (\ref{eq:4.3.11}), and
part (iii) shows that (\ref{eq:4.3.12}) and (\ref{eq:4.3.11}) uniquely solve (\ref{eq:4.3.4}) at $y^{f,k}$.
\bigskip

\noindent\textit{Part (i).} In this part, we construct $\phi_{k-1,j}^{y_f}$ for $j=0,\ldots,k-2$ that satisfy (\ref{eq:4.3.10}), 
where $\mY^f\subset\mY$ is the domain of $\phi_{k-1,j}^{y_f}$ for $j=0,\ldots,k-2$, and $\mY^f$ contains $\phi_{k,k-1}(y^{f,k})$.
By comparing (\ref{eq:4.3.4}) at $y^{f,k}$ and (\ref{eq:4.3.10}), these functions exist at $\phi_{k,k-1}(y^{f,k})$. 

Consider the case of $k=h$. 
We can construct these functions for any $h$ by starting the following discussion from the case of $h=3$ and inductively repeating the discussion for $h=4,5,\ldots$.

Suppose we know how to construct $\phi_{k-1,j}^{y_f}$ for $j=0,\ldots,k-2$ in the case of $k=2,\ldots,h-1$. From (\ref{eq:4.3.10}) with $k=h$, for $y^{f,h-1}\in\mY^f$, we can construct $\phi_{h-1,j}^{y_f}$ with its domain $\mY^{f,2}\subset\mY$ as follows: 
\begin{equation}\label{eq:a.4.2.6}
\phi_{h-1,h-2}^{y^f}(y^{f,h-1})
\in\left\{y\in\mY^{f,2}:G_{h-2,h-1}^{y^f}(y)
= F_{Y_{h-1}|\mC^{h-1}_{0,h-1}}(y^{f,h-1})\right\},
\end{equation}
\begin{equation}\label{eq:a.4.2.7}  
\phi_{h-1,j}^{y^f}(y^{f,h-1})
=\phi_{h-2,j}^{y^f}(\phi_{h-1,h-2}^{y^f}(y^{f,h-1}))\text{ for }j=0,\ldots,h-3,
\end{equation} 
where $G_{h-2,h-1}^{y^f}$ is defined as
\[
G_{h-2,h-1}^{y^f}(y):=\frac{\sum_{j\neq h,h-1}F_{Y_j|\mC^j_{h-1,0}}(\phi_{h-2,j}^{y^f}(y))P(\mC^j_{h-1,0})+F_{Y_h|\mC^h_{h-1,0}}(y^{f,h})P(\mC^h_{h-1,0})}{P(\mC^{h-1}_{0,h-1})},
\]
and $\phi_{h-2,j}^{y_f}$ for $j=0,\ldots,h-3$ are the functions that satisfy the following equations:
\begin{equation}\label{eq:a.4.2.8}
\begin{split}
&F_{Y_i|\mC^i_{0,i}}(\phi_{h-2,i}^{y^f}(y))\\
=&\frac{\sum_{j\neq i,h,h-1}F_{Y_j|\mC^j_{i,0}}(\phi_{h-2,j}^{y^f}(y))P(\mC^j_{i,0})+\sum_{j=h-1}^h F_{Y_j|\mC^j_{i,0}}(y^{f,j})P(\mC^h_{i,0})}{P(\mC^i_{0,i})}\\
\text{ for }&i=1,\ldots,h-2\text{ and }y\in\mY^{f,2},
\end{split}
\end{equation}
We can construct $\phi_{h-2,j}^{y_f}$ for $j=0,\ldots,h-3$ in the same way as the construction of $\phi_{h-2,j}^{y_f}$ for $j=0,\ldots,h-3$ for the case of $k=h-1$. 
\bigskip

\noindent\textit{Part (ii).} In this part, we show that $\phi_{k-1,j}^{y_f}$ for $j=0,\ldots,k-2$ satisfy (\ref{eq:4.3.11}). 
Let $y_{k-1}=\phi_{k,k-1}(y^{f,k})$ and consider the following simultaneous equations of $(y_0,\ldots,y_{k-2})$:
\begin{equation}\label{eq:a.4.5}
\begin{split}
&F_{Y_i|\mC^i_{0,i}}(y_i)\\
=&\frac{\sum_{j\neq i,k}F_{Y_j|\mC^j_{i,0}}(y_j)P(\mC^j_{i,0})+F_{Y_k|\mC^k_{i,0}}(y^{f,k})P(\mC^k_{i,0})}{P(\mC^i_{0,i})}\\
\text{ for }&y_i\in\mY^\circ\text
{ and }i=1,\ldots,k-1.
\end{split}
\end{equation}
We show that $(y_0,\ldots,y_{k-2})=(\phi_{k,0}(y^{f,k}),\ldots,\phi_{k,k-2}(y^{f,k}))$ uniquely solves (\ref{eq:a.4.5}) with $y_{k-1}=\phi_{k,k-1}(y^{f,k})$.
Then, (\ref{eq:4.3.11}) follows from
(\ref{eq:4.3.4}) at $y^{f,k}$ and (\ref{eq:4.3.10}) at $\phi_{k,k-1}(y^{f,k})$.

Suppose we have a solution
$(y_0',\ldots,y_{k-2}')$ different from $(\phi_{k,0}(y^{f,k}),\ldots,\phi_{k,k-2}(y^{f,k}))$ that also satisfies (\ref{eq:a.4.5}) with $y_{k-1}=\phi_{k,k-1}(y^{f,k})$. 
Thus, as we show the contradiction for the cases of $y_0'<\phi_{k,0}(y_k)$, $y_0'>\phi_{k,0}(y_k)$, and $y_0'=\phi_{k,0}(y_k)$ in part (ii) of the proof of Lemma \ref{lem:apb}, we can show that $(y_0',\ldots,y_{k-2}')$ cannot be a solution of (\ref{eq:a.4.5}) with $y_{k-1}=\phi_{k,k-1}(y^{f,k})$ for the cases of $y_0'<\phi_{k,0}(y^{f,k})$, $y_0'>\phi_{k,0}(y^{f,k})$, and $y_0'=\phi_{k,0}(y^{f,k})$. 

Therefore, $(y_0,\ldots,y_{k-1})=(\phi_{k,0}(y^{f,k}),\ldots,\phi_{k,k-1}(y^{f,k}))$ uniquely solves (\ref{eq:a.4.5}) with $y_{k-1}=\phi_{k,k-1}(y^{f,k})$.
Moreover, (\ref{eq:4.3.11}) holds regardless of the uniqueness of $\phi_{k-1,j}^{y_f}$ for $j=0,\ldots,k-2$ to satisfy (\ref{eq:4.3.10}).
\bigskip

\noindent\textit{Part (iii).} In this part, we show that (\ref{eq:4.3.12}) and (\ref{eq:4.3.11}) uniquely solve (\ref{eq:4.3.4}) at $y^{f,k}$.
First, we plug in (\ref{eq:4.3.11}) to (\ref{eq:4.3.4}) with $i=k$:
\begin{equation}\label{eq:a.4.7}
\begin{split}
&F_{Y_k|\mC^h_{0,k}}(y^{f,k})\\
=&\frac{\sum_{j=0}^{k-2}F_{Y_j|\mC^j_{k,0}}(\phi_{k-1,j}^{y^f}(\phi_{k,k-1}(y^{f,k})))P(\mC^j_{k,0})
+F_{Y_{k-1}|\mC^{k-1}_{k,0}}(\phi_{k,k-1}(y^{f,k}))P(\mC^{k-1}_{k,0})}{P(\mC^k_{0,k})}.
\end{split}
\end{equation}
Next, we solve (\ref{eq:a.4.7}) for $\phi_{k,k-1}(y^{f,k})$. 
We first show that,
for $j=0,\ldots,k-2$ and $y\in\mY^f\cap\mY^\circ$, $\phi_{k-1,j}^{y^f}$ satisfies the following properties:
\begin{equation}\label{eq:a.4.8}
y>\phi_{k,k-1}(y^{f,k})\To\phi_{k-1,j}^{y^f}(y)>\phi_{k-1,j}^{y^f}(\phi_{k,k-1}(y^{f,k}))
\end{equation}
and
\begin{equation}\label{eq:a.4.9}
y<\phi_{k,k-1}(y^{f,k})\To\phi_{k-1,j}^{y^f}(y)<\phi_{k-1,j}^{y^f}(\phi_{k,k-1}(y^{f,k})).
\end{equation}
Given (\ref{eq:a.4.8}) and (\ref{eq:a.4.9}), for $y\in\mY^f\cap\mY^\circ$, $G_{k-1,k}^{y^f}$, defined in (\ref{eq:a.4.10}), satisfies a property similar to (\ref{eq:a.4.8}) and (\ref{eq:a.4.9}) of $\phi_{k-1,j}^{y^f}$.

We show (\ref{eq:a.4.8}). 
(\ref{eq:a.4.9}) follows from a similar argument using the reverse signs of inequalities.
Suppose $y^+>\phi_{k,k-1}(y^{f,k})$ satisfies $\phi_{k-1,0}^{y^f}(y^+)\leq \phi_{k-1,0}^{y^f}(\phi_{k,k-1}(y^{f,k}))$. 
From (\ref{eq:4.3.11}), we have $\phi_{k-1,0}^{y^f}(y^+)\leq\phi_{k,0}(y^{f,k})$. Then, as we show the contradiction for the case of $y_0'<\phi_{k,0}(y_k)$ in part (ii) of the proof of Lemma \ref{lem:apb}, we can show that $\phi_{k-1,j}^{y_f}(y^+)$ for $j=k-1,\ldots,0$ cannot satisfy (\ref{eq:4.3.10}) at $y^+$. 
For the case of $y^+>\phi_{k,k-1}(y^{f,k})$ satisfying $\phi_{k-1,0}^{y^f}(y^+)> \phi_{k-1,0}^{y^f}(\phi_{k,k-1}(y^{f,h}))$, consider the case of $\phi_{k-1,j}^{y^f}(y^+)\leq \phi_{k-1,j}^{y^f}(\phi_{k,k-1}(y^{f,h}))$ for some $j\in\setminus\{0,k-1,k\}$.
We can show a contradiction in the same way as in the case of $y_0'=\phi_{k,0}(y_k)$ and $y_j'\neq\phi_{k,j}(y_k)$ for some $j\in\mT\setminus\{0,k\}$ in part (ii) of the proof of Lemma \ref{lem:apb}.
Hence, (\ref{eq:a.4.8}) holds.

Therefore, we can solve (\ref{eq:a.4.7}) for $\phi_{k,k-1}(y^{f,k})$, and the closed-form expression of $\phi_{k,k-1}(y)$ at each $y^{f,k}$ is derived as (\ref{eq:4.3.12}). Closed-form expressions of $\phi_{k,j}$ for $j=0,\ldots,k-2$ are derived by plugging in the closed-form expression of $\phi_{k,k-1}$ to (\ref{eq:4.3.11}).

%% file: fullrank.tex
\section{The full rank conditions in Chernozhukov and Hansen (2005) under Assumption \ref{asp:apb}}\label{sec:e}
This section shows that Assumption \ref{asp:apb} implies the full rank conditions in \cite{CH2005} under Assumptions \ref{asp:ivqr} and \ref{asp:mt} and some differentiability assumptions.
For simplicity, let $\mZ$ contain $k+1$ values and $\mZ=\{z_0,z_1\ldots,z_k\}$.
The following argument does not rely on this assumption.

Assume that Assumptions \ref{asp:ivqr}-\ref{asp:apb} hold.
Fix some small constants $\overline \delta>0$ and $\underline{f}>0$.
For a given quantile $\tau\in(0,1)$ and $\delta\geq 0$, define $\mL(\delta)$ as a set containing all vectors $(y_0,\ldots,y_k)$ that satisfy $\sum_{t=0}^k F_{Y|TZ}(y_t|t,z)p_t(z)\in[\tau-\delta,\tau+\delta]$ for each $z\in\mZ$.
Assume that, for each $t\in\mT$ and $z\in\mZ$, $F_{Y|TZ}(\cdot|t,z)$ is continuously differentiable on $\mL_0$.
Define $\mL$ as a subset of $\mL_0$ containing all vectors $(y_0,\ldots,y_k)$ that satisfy the following restriction:
\begin{itemize}
\item[] For each $t\in\mT$, $y_t$ is contained in a set containing all values $y\in\mY$ that satisfy $f_{Y|TZ}(y|t,z)\geq \underline{f}$ for all $z\in\mZ$ with $p_t(z)>0$ and $f_{Y|\mC^t_{z,z'}}(y)\geq \underline{f}$ for all $(z,z')\in\mP$ with $D^t_z\leq D^t_{z'}$ almost surely and  $P(\mC^t_{z,z'})>0$.
\end{itemize}
Assume that $(Q_{Y_0}(\tau),\ldots,Q_{Y_k}(\tau))$ is contained in $\mL$.
Define $\Pi:\R^{k+1}\to\R^{k+1}$ as
\begin{equation}\label{eq:ch1}
\Pi(y_0,\ldots,y_k):=
\begin{pmatrix}
\sum_{t=0}^k F_{Y|TZ}(y_t|t,z_0)p_t(z_0)-\tau \\
\vdots \\
\sum_{t=0}^k F_{Y|TZ}(y_t|t,z_k)p_t(z_k)-\tau
\end{pmatrix}.
\end{equation}
Then $\Pi'(y_0,\ldots,y_k)$, defined in \eqref{eq:ch2}, is the Jacobian of $\Pi$ with respect to $(y_0,\ldots,y_k)$:
\begin{equation*}
\Pi'(y_0,\ldots,y_k):=
\begin{pmatrix}
f_{Y|TZ}(y_0|0,z_0)p_0(z_0) & \cdots & f_{Y|TZ}(y_k|k,z_0)p_k(z_0) \\
\vdots & \ddots & \vdots \\
f_{Y|TZ}(y_0|0,z_k)p_0(z_k) & \cdots & f_{Y|TZ}(y_k|k,z_k)p_k(z_k)
\end{pmatrix}.
\end{equation*}
We show that $\Pi'(y_0,\ldots,y_k)$ is full rank for each $(y_0,\ldots,y_k)\in\mL\cap\mL(0)$.
In this case, $\Pi'(y_0,\ldots,y_k)$ is also full rank for each $(y_0,\ldots,y_k)\in\mL\cap\mL(\delta^*)$ for some $\delta^*>0$ because the determinant of $\Pi'(y_0,\ldots,y_k)$ is continuous in $(y_0,\ldots,y_k)$.

The restrictions define $\mL$ as a set of potential solutions to the moment equations \eqref{eq:ti} where the solutions are contained in $\mY$.
Both $f_{Y|TZ}(\cdot|t,z)$ and $f_{Y|\mC^t_{z,z'}}(\cdot)$ are required to be positive on $\mL$ because
we assume that the support of the conditional distribution of $Y_t$ given $\mC^t_{z,z'}$ is $\mY$ in Assumption \ref{asp:mt} (iv).
Note that, for each $s,t\in\mT$, the derivative of the counterfactual mapping $\phi_{s,t}'(\cdot)$ is also positive on $\mL$.
This is because, for $(y_0,\ldots,y_k)\in\mL$, the definition implies
\begin{equation}\label{eq:e.00}
\phi_{s,t}'(y_s)=\frac{f_{Y_s}(y_s)}{f_{Y_t}(y_t)}>0.
\end{equation}

Suppose that the monotonicity subset $\Lambda\subset\mP$ contains $k$ pairs of instrument values $\lambda_1,\ldots,\lambda_k$ such that each sign treatment is $t({\lambda_i})=i$ for $i=1,\ldots,k$. 
For notation simplicity, let $\lambda_i=(i,0)$ and $D_{i}({i})\geq D_{i}({0})$ hold almost surely. 
Then the monotonicity inequalities correspond to those of Example I in Section \ref{subsec:3.2}, and the types of monotonicity relationships on $(i,0)$ and $(j,0)$ for $i\neq j$ differ. 
The proof does not rely on this assumption.
Similarly, we can prove the stated result without this assumption.

For notation simplicity, let $f_{Y_i\mC^i_{0,i}}(\cdot):=f_{Y_i|\mC^i_{0,i}}(\cdot)P(\mC^i_{0,i})$ and
$f_{i,j}(\cdot):=f_{Y|TZ}(\cdot|j,0)p_j(i)$.
Observe that Lemma \ref{lem:idcpl} implies $f_{Y_t\mC^t_{z,z'}}(y)=f_{z',t}(y)-f_{z,t}(y)$, and the determinant of the Jacobian $\Pi'(y_0,\ldots,y_k)$ can be expressed as follows:
\begin{equation}\label{eq:e.1}
\begin{split}
&\Pi'(y_0,\ldots,y_k)\\
=&\begin{vmatrix}
f_{0,0}(y_0) & \cdots & f_{0,k}(y_k)\\
\vdots & \ddots & \vdots \\
f_{k,0}(y_0) & \cdots & f_{k,k}(y_k)
\end{vmatrix}\\
=&\begin{vmatrix}
f_{0,0}(y_0) & f_{0,1}(y_1) & f_{0,2}(y_2) & \cdots & f_{0,k-1}(y_{k-1}) & f_{0,k}(y_k) \\
-f_{Y_0\mC^0_{1,0}}(y_0) & f_{Y_1\mC^1_{0,1}}(y_1) & -f_{Y_2\mC^2_{1,0}}(y_2) & \cdots & -f_{Y_{k-1}\mC^{k-1}_{1,0}}(y_{k-1}) & -f_{Y_k\mC^k_{1,0}}(y_k)\\
\vdots & \vdots & \vdots & \ddots & \vdots \\
-f_{Y_0\mC^0_{k,0}}(y_0) & -f_{Y_1\mC^1_{k,0}}(y_1) & -f_{Y_2\mC^2_{k,0}}(y_2) & \cdots & -f_{Y_{k-1}\mC^{k-1}_{k,0}}(y_{k-1}) & f_{Y_k\mC^k_{0.k}}(y_k)
\end{vmatrix}.
\end{split}
\end{equation}
Define
\begin{equation}\label{eq:e.2}
\begin{split}
&\Pi^{\wedge}(y_0,\ldots,y_k)\\
:=&\begin{pmatrix}
f_{0,0}(y_0) & f_{0,1}(y_1) & f_{0,2}(y_2) & \cdots & f_{0,k-1}(y_{k-1}) & f_{0,k}(y_k) \\
-f_{Y_0\mC^0_{1,0}}(y_0) & f_{Y_1\mC^1_{0,1}}(y_1) & -f_{Y_2\mC^2_{1,0}}(y_2) & \cdots & -f_{Y_{k-1}\mC^{k-1}_{1,0}}(y_{k-1}) & -f_{Y_k\mC^k_{1,0}}(y_k)\\
\vdots & \vdots & \vdots & \ddots & \vdots \\
-f_{Y_0\mC^0_{k,0}}(y_0) & -f_{Y_1\mC^1_{k,0}}(y_1) & -f_{Y_2\mC^2_{k,0}}(y_2) & \cdots & -f_{Y_{k-1}\mC^{k-1}_{k,0}}(y_{k-1}) & f_{Y_k\mC^k_{0.k}}(y_k)
\end{pmatrix}.
\end{split}
\end{equation}
It suffices to show that $\Pi^{\wedge}(y_0,\ldots,y_k)$ is full rank.

For $l=2,\ldots,k$, let $\mM^l$ be a set of $l\times l$ submatrices of $\Pi^{\wedge}(y_0,\ldots,y_k)$ whose diagonal elements are positive and all other entries are negative.
Let $\mM^l_{i,j}$ be a set of $l\times l$ submatrices of $\Pi^{\wedge}(y_0,\ldots,y_k)$ such that all entries in row $i$ and column $j$ are negative, and the $(l-1)\times(l-1)$ submatrix obtained by removing row $i$ and column $j$ is contained in $\mM^{l-1}$.
Let $\Pi^l\in\mM^l$ and $\Pi^l_{i,j}\in\mM^l_{i,j}$.
Using induction, we show that the following holds for $l=2,\ldots,k$:
\begin{equation}\label{eq:e.3}
\begin{cases}
|\Pi^l_{i,j}|>0 & \text{ if }(i+j)\text{ is odd}\\
|\Pi^l_{i,j}|<0 & \text{ otherwise}
\end{cases}
\end{equation}
and
\begin{equation}\label{eq:e.4}
|\Pi^l|\geq0.
\end{equation}

For $l=2$, \eqref{eq:e.3} follows from the definition.
Observe that, for $i_1<\cdots<i_l$, $\Pi^l$ can be expressed as 
\begin{equation}\label{eq:e.5}
\Pi^l
=\begin{pmatrix}
f_{Y_{i_1}\mC^{i_1}_{0,i_1}}(y_{i_1}) & -f_{Y_{i_2}\mC^{i_2}_{i_1,0}}(y_{i_2}) & \cdots & -f_{Y_{i_{l-1}}\mC^{i_{l-1}}_{i_1,0}}(y_{i_{l-1}}) & -f_{Y_{i_l}\mC^{i_l}_{i_1,0}}(y_{i_l})\\
\vdots & \vdots & \ddots & \vdots \\
-f_{Y_{i_1}\mC^{i_1}_{i_l,0}}(y_{i_1}) & -f_{Y_{i_2}\mC^{i_2}_{i_l,0}}(y_{i_2}) & \cdots & -f_{Y_{i_{l-1}}\mC^{i_{l-1}}_{i_l,0}}(y_{i_{l-1}}) & f_{Y_{i_l}\mC^{i_l}_{0,i_l}}(y_{i_l})
\end{pmatrix}.
\end{equation}
Note that (\ref{eq:4.3.4}) holds, which is shown in the proof of Lemma \ref{lem:apb}, and (\ref{eq:4.3.4}) implies 
\begin{equation}\label{eq:e.0}
f_{Y_i\mC^i_{0,i}}(y_i)\phi_{k,i}'(y_k)
=\sum_{j\neq i}f_{Y_j\mC^j_{i,0}}(y_j)\phi_{k,j}'(y_k)\text{ for }y\in\mY^\circ\text{ and }i=1,\ldots,k.
\end{equation}
From \eqref{eq:e.0}, we obtain
\begin{equation}\label{eq:e.6}
\Pi^l
\begin{pmatrix}
\phi_{k,i_1}'(y_k)\\
\vdots\\
\phi_{k,i_l}'(y_k)
\end{pmatrix}
=\begin{pmatrix}
\sum_{j\neq i_1,\ldots,i_l}f_{Y_j\mC^j_{i_1,0}}(y_j)\phi_{k,j}'(y_k)\\
\vdots\\
\sum_{j\neq i_1,\ldots,i_l}f_{Y_j\mC^j_{i_l,0}}(y_j)\phi_{k,j}'(y_k)
\end{pmatrix}.
\end{equation}
Then, if $\Pi^l$ is full rank, we have 
\begin{equation}\label{eq:e.7}
\begin{pmatrix}
\phi_{k,i_1}'(y_k)\\
\vdots\\
\phi_{k,i_l}'(y_k)
\end{pmatrix}
=\frac{1}{|\Pi^l|}\tilde \Pi^l
\begin{pmatrix}
\sum_{j\neq i_1,\ldots,i_l}f_{Y_j\mC^j_{i_1,0}}(y_j)\phi_{k,j}'(y_k)\\
\vdots\\
\sum_{j\neq i_1,\ldots,i_l}f_{Y_j\mC^j_{i_l,0}}(y_j)\phi_{k,j}'(y_k)
\end{pmatrix},
\end{equation}
where $\tilde \Pi^l$ denotes the adjugate of $\Pi^l$.
Note that, from \eqref{eq:e.00}, all entries on the left-hand side of \eqref{eq:e.7} are positive.
From the definition,
\[
\tilde \Pi^2
=\begin{pmatrix}
f_{Y_{i_2}\mC^{i_2}_{0,i_2}}(y_{i_2}) & f_{Y_{i_2}\mC^{i_2}_{i_1,0}}(y_{i_2})\\
f_{Y_{i_1}\mC^{i_1}_{i_2,0}}(y_{i_1}) &  f_{Y_{i_1}\mC^{i_1}_{0,i_1}}(y_{i_1})
\end{pmatrix},
\]
Thus, $|\Pi^2|$ is positive because all entries of $\tilde \Pi^2$ are positive.
Hence, \eqref{eq:e.4} follows for $l=2$.

Suppose that $\Pi^{l-1}\in\mM^{l-1}$ and $\Pi^{l-1}_{i,j}\in\mM^{l-1}_{i,j}$ satisfy
\begin{equation}\label{eq:e.8}
\begin{cases}
|\Pi^{l-1}_{i,j}|>0 & \text{ if }(i+j)\text{ is odd}\\
|\Pi^{l-1}_{i,j}|<0 & \text{ otherwise}
\end{cases}\quad\text{and}\quad |\Pi^{l-1}|\geq0.
\end{equation}
First, we show \eqref{eq:e.3} under \eqref{eq:e.8}.
Let $\pi^l_{i,m}$ be the $(i,m)$ entry of $\Pi^l_{i,j}$.
Then, from the definition, there exist $(l-1)\times(l-1)$ matrices $\Pi^{l-1}_{i,m}\in\mM^{l-1}_{i,m}$ and 
$\Pi^{l-1}\in\mM^{l-1}$ such that the Laplace expansion of the determinant of $\Pi^l_{i,j}$ can be expressed as follows:
\begin{equation}\label{eq:e.9}
|\Pi^l_{i,j}|=\sum_{m=1}^{j-1}(-1)^{i+m}\pi^l_{i,m}
|\Pi^{l-1}_{m,j-1}|+(-1)^{i+j}\pi^l_{i,j}|\Pi^{l-1}|
+\sum_{m=j+1}^{l}(-1)^{i+m}\pi^l_{i,m}|\Pi^{l-1}_{m-1,j}|.
\end{equation}
Observe that, using \eqref{eq:e.8}, we obtain, for $m=1,\ldots,j-1$,
\begin{equation}\label{eq:e.10}
\begin{cases}
(-1)^{i+m}\pi^l_{i,m}|\Pi^{l-1}_{m,j-1}|>0 & 
\text{ if }(i+j)\text{ is odd}\\
(-1)^{i+m}\pi^l_{i,m}|\Pi^{l-1}_{m,j-1}|<0 & \text{ otherwise}
\end{cases},
\end{equation}
\begin{equation}\label{eq:e.11}
\begin{cases}
(-1)^{i+j}\pi^l_{i,j}|\Pi^{l-1}|>0 & 
\text{ if }(i+j)\text{ is odd}\\
(-1)^{i+j}\pi^l_{i,j}|\Pi^{l-1}|<0 & \text{ otherwise}
\end{cases},
\end{equation}
and for $m=j+1,\ldots,l$,
\begin{equation}\label{eq:e.12}
\begin{cases}
(-1)^{i+m}\pi^l_{i,m}|\Pi^{l-1}_{m-1,j}|>0 & 
\text{ if }(i+j)\text{ is odd}\\
(-1)^{i+m}\pi^l_{i,m}|\Pi^{l-1}_{m-1,j}|<0 & \text{ otherwise}
\end{cases}.
\end{equation}
Hence, \eqref{eq:e.3} follows from \eqref{eq:e.9}-\eqref{eq:e.12}.

Next, we show \eqref{eq:e.4} under \eqref{eq:e.8}.
Let $\tilde \pi^l_{j,m}$ be the $(j,m)$ cofactor of $\Pi^l$.
Then, from the definition, there exist $(l-1)\times(l-1)$ matrices $\Pi^{l-1}_{i,m}\in\mM^{l-1}_{i,m}$ and 
$\Pi^{l-1}\in\mM^{l-1}$ such that the cofactors of $\Pi^l$ can be expressed as
\begin{equation}\label{eq:e.13}
\tilde \pi^l_{j,m}=
\begin{cases}
(-1)^{j+m}|\Pi^{l-1}_{m,j-1}| & \text{ if }m=1,\ldots,j-1\\
(-1)^{j+m}|\Pi^{l-1}| & \text{ if }m=j\\
(-1)^{j+m}|\Pi^{l-1}_{m-1,j}| & \text{ if }m=j+1,\ldots,l.
\end{cases}
\end{equation}
Observe that, using \eqref{eq:e.8}, we obtain $(-1)^{j+m}|\Pi^{l-1}_{m,j-1}|>0$ for $m=1,\ldots,j-1$,
$(-1)^{2j}|\Pi^{l-1}|\geq0$, and
$(-1)^{j+m}|\Pi^{l-1}_{m-1,j}|>0$ for $m=j+1,\ldots,l$.
Then, based on \eqref{eq:e.7}, $|\Pi^l|$ is positive when $|\Pi^l|$ is full rank because each $(j,m)$ cofactor of $\Pi^l$ satisfies $\tilde \pi^l_{j,m}\geq0$.
Hence, \eqref{eq:e.4} follows.

Therefore, using induction, \eqref{eq:e.3} and \eqref{eq:e.4} hold.
From the definition, there exist $k\times k$ matrices $\Pi^{k}_{i,m}\in\mM^{k}_{i,m}$ and 
$\Pi^{k}\in\mM^{k}$ such that the Laplace expansion of the determinant of $\Pi^{\wedge}(y_0,\ldots,y_k)$ can be expressed as follows:
\begin{equation}\label{eq:e.14}
|\Pi^+(y_0,\ldots,y_k)|
=f_{0,0}(y_0)|\Pi^{k}|
+\sum_{m=2}^{k+1}(-1)^{1+m}f_{0,m-1}(y_{m-1})|\Pi^{k}_{m-1,1}|.
\end{equation}
Then $|\Pi^{\wedge}(y_0,\ldots,y_k)|$ is positive because, using \eqref{eq:e.3} and \eqref{eq:e.4} for $l=k$, we obtain 
$f_{0,0}(y_0)|\Pi^{k}|\geq0$, and
$(-1)^{1+m}f_{0,m-1}(y_{m-1})|\Pi^{k}_{m-1,1}|>0$ for $m=2,\ldots,l$.
Therefore, $\Pi^{\wedge}(y_0,\ldots,y_k)$ and hence $\Pi'(y_0,\ldots,y_k)$ are full rank.

\section{The full rank conditions in Chernozhukov and Hansen (2005) where Assumption \ref{asp:apb} is violated}\label{sec:e.2}
This section discusses the case where Assumption \ref{asp:apb.2} is violated, although the full rank conditions in \cite{CH2005} hold for the $k=2$ case.
We also present a numerical example for this case. 
Similar to Section \ref{subsec:4.4}, we consider a case where the unordered monotonicity assumption holds but Assumption \ref{asp:apb} is violated because the sign treatments are the same for all instrument pairs.
First, we show that when the two sets of solutions to (\ref{eq:all1.f}) and (\ref{eq:all2.f}) intersect only once, as shown in Figure \ref{fig:ex1.4}, the full rank conditions in \cite{CH2005} hold under Assumptions \ref{asp:ivqr} and \ref{asp:mt} and some differentiability assumptions.
We assume the monotonicity assumption as given in Table \ref{tab:mtn}.
We also define the sets and notations and assume the differentiability assumptions as shown in Appendix \ref{sec:e}.
We define the two functions $\phi_{1,0}^{y_f}$ and $\phi_{2,0}^{y_f}$ with their domain $\mY^f\subset\mY$ as 
\begin{equation}\label{eq:allf1}
\phi_{1,0}^{y_f}(y)
:=Q_{Y_0|\mC^0_{1,0}}\left(\frac{-F_{Y_1|\mC^1_{1,0}}(y)P(\mC^1_{1,0})+F_{Y_2|\mC^2_{0,1}}(y^f)P(\mC^2_{0,1})}{P(\mC^0_{1,0})}\right),
\end{equation}
\begin{equation}\label{eq:allf2}
\phi_{2,0}^{y_f}(y)
:=Q_{Y_0|\mC^0_{2,0}}\left(\frac{-F_{Y_1|\mC^1_{2,0}}(y)P(\mC^1_{2,0})+F_{Y_2|\mC^2_{0,2}}(y^f)P(\mC^2_{0,2})}{P(\mC^0_{2,0})}\right)\text{ for }y\in\mY^f.
\end{equation}
Note that these functions are visualized as the two curves in Figure \ref{fig:ex1.4}, and the sets of solutions to (\ref{eq:all1.f}) and (\ref{eq:all2.f}) are expressed as $\{(y_0,y_1):y_1\in\mY^f\text{ and }y_0=\phi_{1,0}^{y^f}(y_1)\}$ and $\{(y_0,y_1):y_1\in\mY^f\text{ and }y_0=\phi_{2,0}^{y^f}(y_1)\}$, respectively.
Then the derivatives of these functions for each solution are as follows:
\begin{equation*}
\phi_{1,0}^{y_f'}(y_1) = -\frac{f_{Y_1\mC^1_{1,0}}(y_1)}{f_{Y_0\mC^0_{1,0}}(y_0)} \text{ for }(y_0,y_1)\in\{(y_0,y_1):y_1\in\mY^f\text{ and }y_0=\phi_{1,0}^{y^f}(y_1)\},
\end{equation*}
\begin{equation*}
\phi_{2,0}^{y_f'}(y_1) = -\frac{f_{Y_1\mC^1_{2,0}}(y_1)}{f_{Y_0\mC^0_{2,0}}(y_0)} \text{ for }(y_0,y_1)\in\{(y_0,y_1):y_1\in\mY^f\text{ and }y_0=\phi_{2,0}^{y^f}(y_1)\}.
\end{equation*}
When two functions $\phi_{1,0}^{y_f}$ and $\phi_{2,0}^{y_f}$ intersect only once, we have $\phi_{1,0}^{y_f'}(y_1)\neq \phi_{2,0}^{y_f'}(y_1)$ at the intersection point.
We show that, when we have $\phi_{1,0}^{y_f'}(y_1)< \phi_{2,0}^{y_f'}(y_1)$, the Jacobian $\Pi'(y_0,\ldots,y_k)$ is full rank for each $(y_0,\ldots,y_k)\in\mL\cap\mL(0)$.
Hence $\Pi'(y_0,\ldots,y_k)$ is also full rank for each $(y_0,\ldots,y_k)\in\mL\cap\mL(\delta^*)$ for some $\delta^*>0$.
For the case of $\phi_{1,0}^{y_f'}(y_1)> \phi_{2,0}^{y_f'}(y_1)$, we can show this using an analogous argument.

Under the monotonicity assumption of Table \ref{tab:mtn}, the determinant of the Jacobian $\Pi'(y_0,\ldots,y_k)$ can be expressed as follows:
\begin{equation}\label{eq:e.2.1}
\begin{split}
\Pi'(y_0,y_1,y_2)=\begin{vmatrix}
f_{0,0}(y_0) & f_{0,1}(y_1) & f_{0,2}(y_2) \\
-f_{Y_0\mC^0_{1,0}}(y_0) & -f_{Y_1\mC^1_{1,0}}(y_1) & f_{Y_2\mC^2_{0,1}}(y_2)\\
-f_{Y_0\mC^0_{2,0}}(y_0) & -f_{Y_1\mC^1_{2,0}}(y_1) & f_{Y_2\mC^2_{0,2}}(y_2)
\end{vmatrix}.
\end{split}
\end{equation}
Define
\begin{equation}\label{eq:e.2.2}
\begin{split}
\Pi^{\vee}(y_0,y_1,y_2)
:=\begin{pmatrix}
f_{0,0}(y_0) & f_{0,1}(y_1) & f_{0,2}(y_2) \\
-f_{Y_0\mC^0_{1,0}}(y_0) & -f_{Y_1\mC^1_{1,0}}(y_1) & f_{Y_2\mC^2_{0,1}}(y_2)\\
-f_{Y_0\mC^0_{2,0}}(y_0) & -f_{Y_1\mC^1_{2,0}}(y_1) & f_{Y_2\mC^2_{0,2}}(y_2)
\end{pmatrix}.
\end{split}
\end{equation}
It suffices to show that $\Pi^{\vee}(y_0,y_1,y_2)$ is full rank.
Let $\Pi^{\vee}_{i,j}$ be the $2\times2$ submatrix of $\Pi^{\vee}$ obtained by removing row $i$ and column $j$.
Note that (\ref{eq:all1}) and (\ref{eq:all2}) imply
\begin{equation}\label{eq:e.2.0}
f_{Y_2\mC^2_{0,i}}(y_2)
=f_{Y_0\mC^0_{i,0}}(y_0)\phi_{2,0}'(y_2)
+f_{Y_1\mC^1_{i,0}}(y_1)\phi_{2,1}'(y_2)\text{ for }y\in\mY^\circ\text{ and }i=1,2.
\end{equation}
From \eqref{eq:e.2.0}, we obtain
\begin{equation}\label{eq:e.2.6}
\begin{pmatrix}
f_{Y_0\mC^0_{1,0}}(y_0) & f_{Y_1\mC^1_{1,0}}(y_1)\\
f_{Y_0\mC^0_{2,0}}(y_0) & f_{Y_1\mC^1_{2,0}}(y_1)
\end{pmatrix}
\begin{pmatrix}
\phi_{2,0}'(y_2)\\
\phi_{2,1}'(y_2)
\end{pmatrix}
=\begin{pmatrix}
f_{Y_2\mC^2_{0,1}}(y_2)\\
f_{Y_2\mC^2_{0,2}}(y_2)
\end{pmatrix}.
\end{equation}
Then, because $\phi_{1,0}^{y_f'}(y_1)< \phi_{2,0}^{y_f'}(y_1)$ implies that $\Pi^{\vee}_{1,3}$ is full rank with 
\begin{equation*}
|\Pi^{\vee}_{1,3}|
=f_{Y_0\mC^0_{1,0}}(y_0)f_{Y_1\mC^1_{2,0}}(y_1)-f_{Y_0\mC^0_{2,0}}(y_0)f_{Y_1\mC^1_{1,0}}(y_1)>0,
\end{equation*}
we have 
\begin{equation}\label{eq:e.2.7}
\begin{pmatrix}
\phi_{2,0}'(y_2)\\
\phi_{2,1}'(y_2)
\end{pmatrix}
=\frac{1}{\Pi^{\vee}_{1,3}}
\begin{pmatrix}
f_{Y_1\mC^1_{2,0}}(y_1) & -f_{Y_1\mC^1_{1,0}}(y_1)\\
-f_{Y_0\mC^0_{2,0}}(y_0) & f_{Y_0\mC^0_{1,0}}(y_0)
\end{pmatrix}
\begin{pmatrix}
f_{Y_2\mC^2_{0,1}}(y_2)\\
f_{Y_2\mC^2_{0,2}}(y_2)
\end{pmatrix}.
\end{equation}
Note that, from \eqref{eq:e.00}, all entries on the left-hand side of \eqref{eq:e.2.7} are positive.
Thus, \eqref{eq:e.2.7} implies
\begin{equation*}
\begin{pmatrix}
f_{Y_1\mC^1_{2,0}}(y_1) & -f_{Y_1\mC^1_{1,0}}(y_1)
\end{pmatrix}
\begin{pmatrix}
f_{Y_2\mC^2_{0,1}}(y_2)\\
f_{Y_2\mC^2_{0,2}}(y_2)
\end{pmatrix}
=|\Pi^{\vee}_{1,1}|>0,
\end{equation*}
\begin{equation*}
\begin{pmatrix}
-f_{Y_0\mC^0_{2,0}}(y_0) & f_{Y_0\mC^0_{1,0}}(y_0)
\end{pmatrix}
\begin{pmatrix}
f_{Y_2\mC^2_{0,1}}(y_2)\\
f_{Y_2\mC^2_{0,2}}(y_2)
\end{pmatrix}
=-|\Pi^{\vee}_{1,2}|<0.
\end{equation*}
Hence, from the Laplace expansion of the determinant of $\Pi^{\vee}$, we have
\begin{equation*}
|\Pi^{\vee}| = \sum_{j=1}^3 f_{0,j-1}(y_{j-1})(-1)^{1+j}|\Pi^{\vee}_{1,j}|>0.
\end{equation*}
Therefore, $\Pi^{\vee}(y_0,y_1,y_2)$ and hence $\Pi'(y_0,y_1,y_2)$ are full rank.

Next, we provide a numerical example where the full rank conditions in \cite{CH2005} hold under the monotonicity assumption of Table \ref{tab:mtn}.
Suppose that $Y_t$ is normally distributed with mean $t$ and variance 1, and $Y_t$ is independent of $T(z)$ under Assumption \ref{asp:mt} (i).
Let the distributions of $T(z)$ 's be similar to those in \eqref{eq:4.4.1}
The monotonicity inequalities summarized in Table \ref{tab:mtn} hold almost surely.
\begin{equation}\label{eq:4.4.1}
\begin{pmatrix}
P(D_0(0)=1) & P(D_1(0)=1) & P(D_2(0)=1) \\
P(D_0(1)=1) & P(D_1(1)=1) & P(D_2(1)=1) \\
P(D_0(2)=1) & P(D_1(2)=1) & P(D_2(2)=1)
\end{pmatrix}=
\begin{pmatrix}
0.3125 & 0.4375 & 0.25 \\
0.25 & 0.375 & 0.375 \\
0.125 & 0.3125 & 0.5625
\end{pmatrix}
\end{equation}
In this case, Assumption \ref{asp:apb} is violated but the full rank condition in \cite{CH2005} holds for all $y_t\in\mY$.
This is because the determinant of the Jacobian defined in \eqref{eq:ch2} is shown to be negative as follows:
\begin{equation}\label{eq:4.4.2}
\begin{split}
|\Pi'(y_0,y_1,y_2)|=&
\begin{vmatrix}
f_{Y|TZ}(y_0|0,0)p_0(0) & f_{Y|TZ}(y_1|1,0)p_1(0) & f_{Y|TZ}(y_2|2,0)p_2(0) \\
f_{Y|TZ}(y_0|0,1)p_0(1) & f_{Y|TZ}(y_1|1,1)p_1(1) & f_{Y|TZ}(y_2|2,1)p_2(1) \\
f_{Y|TZ}(y_0|0,2)p_0(2) & f_{Y|TZ}(y_1|1,2)p_1(2) & f_{Y|TZ}(y_2|2,2)p_2(2)
\end{vmatrix}\\
=&\begin{vmatrix}
f_{Y_0}(y_0)P(D_0(0)=1) & f_{Y_1}(y_1)P(D_1(0)=1) & f_{Y_2}(y_2)P(D_2(0)=1) \\
f_{Y_0}(y_0)P(D_0(1)=1) & f_{Y_1}(y_1)P(D_1(1)=1) & f_{Y_2}(y_2)P(D_2(1)=1) \\
f_{Y_0}(y_0)P(D_0(2)=1) & f_{Y_1}(y_1)P(D_1(2)=1) & f_{Y_2}(y_2)P(D_2(2)=1)
\end{vmatrix}\\
=&f_{Y_0}(y_0)f_{Y_1}(y_1)f_{Y_2}(y_2)\begin{vmatrix}
P(D_0(0)=1) & P(D_1(0)=1) & P(D_2(0)=1) \\
P(D_0(1)=1) & P(D_1(1)=1) & P(D_2(1)=1) \\
P(D_0(2)=1) & P(D_1(2)=1) & P(D_2(2)=1)
\end{vmatrix}\\
=&-f_{Y_0}(y_0)f_{Y_1}(y_1)f_{Y_2}(y_2)0.00390625<0.
\end{split}
\end{equation}

\section{Testable restrictions under Assumption \ref{asp:apb}}\label{sec:e.3}
Assumptions \ref{asp:mt} and \ref{asp:apb} imply additional testable restrictions compared with \cite{CH2005}.
Suppose that Assumptions \ref{asp:ivqr}-\ref{asp:apb} hold.
Then there exists $k$ pairs of instrument values $\lambda_1,\ldots,\lambda_k\in\mP$ such that the 
following condition holds:
\begin{enumerate}
\item[] 
For $i=1,\ldots,k$, there uniquely exists $t({\lambda_i})\in\mT$ with $\lambda_i=(\lambda_{i,1},\lambda_{i,2})$ $t({\lambda_i})\neq t({\lambda_j})$ for $i\neq j$ such that
\begin{equation}\label{eq:e.3.1}
F_{Y|TZ}(y|t({\lambda_i}),\lambda_{i,1})p_{t({\lambda_i})}(\lambda_{i,1})-F_{Y|TZ}(y|t({\lambda_i}),\lambda_{i,2})p_{t({\lambda_i})}(\lambda_{i,2})
\end{equation}
and
\begin{equation}\label{eq:e.3.2}
F_{Y|TZ}(y|j,\lambda_{i,2})p_j(\lambda_{i,2})-F_{Y|TZ}(y|j,\lambda_{i,1})p_j(\lambda_{i,1})
\end{equation}
are strictly increasing on $\mY$ for $j\in \mT\setminus\{t({\lambda_i})\}$.
\end{enumerate}
This result follows from the fact that the supports of  $F_{Y_{t({\lambda_i})}|\mC^{t({\lambda_i})}_{\lambda_{i,2},\lambda_{i,1}} }(y)$ and $F_{Y_{j}|\mC^{j}_{\lambda_{i,1},\lambda_{i,2}}}(y)$
are $\mY$ from Assumption \ref{asp:mt} (iv), and 
$F_{Y_{t({\lambda_i})}|\mC^{t({\lambda_i})}_{\lambda_{i,2},\lambda_{i,1}} }(y)
P(\mC^{t({\lambda_i})}_{\lambda_{i,2},\lambda_{i,1}})$ and $F_{Y_{j}|\mC^{j}_{\lambda_{i,1},\lambda_{i,2}}}(y)P(\mC^{j}_{\lambda_{i,1},\lambda_{i,2}})$ are identified as \eqref{eq:e.3.1} and \eqref{eq:e.3.2} respectively from Lemma \ref{lem:idcpl}.

Under some differentiability assumptions, the above restrictions may be expressed in terms of densities. 
Assume $F_{Y|TZ}(\cdot|t,z)$ is continuously differentiable on $\mY$ for each $t\in\mT$ and $z\in\mZ$.
Then, there exist $k$ pairs of instrument values $\lambda_1,\ldots,\lambda_k\in\mP$ such that the 
following condition holds:
\begin{enumerate}
\item[] 
For $i=1,\ldots,k$, there uniquely exist $t({\lambda_i})\in\mT$ with $t({\lambda_i})\neq t({\lambda_j})$ for $i\neq j$ such that
\begin{equation}\label{eq:e.3.3}
f_{Y|TZ}(y|t({\lambda_i}),\lambda_{i,1})p_{t({\lambda_i})}(\lambda_{i,1})-f_{Y|TZ}(y|t({\lambda_i}),\lambda_{i,2})p_{t({\lambda_i})}(\lambda_{i,2})>0
\end{equation}
and
\begin{equation}\label{eq:e.3.4}
f_{Y|TZ}(y|j,\lambda_{i,2})p_j(\lambda_{i,2})-f_{Y|TZ}(y|j,\lambda_{i,1})p_j(\lambda_{i,1})>0
\end{equation}
on $\mY$ for $j\in \mT\setminus\{t({\lambda_i})\}$.
\end{enumerate}
For the binary case, \cite{Wuthrich2019} shows that Assumption \ref{asp:mt} implies full rank conditions in \cite{CH2005}.
They are expressed as 
\eqref{eq:e.3.3} and \eqref{eq:e.3.4} with $(t_{\lambda^i},j)=(0,1) \text{ or } (1,0)$.
\eqref{eq:e.3.3} and \eqref{eq:e.3.4} are equally difficult to check statistically compared to the full rank conditions in \cite{CH2005} for the binary case.  

%% file: without.tex
\section{Identification as closed-form expressions without Assumption \ref{asp:ivqr} (v)}\label{sec:g}

In this section, we relax Assumption \ref{asp:ivqr} (v)
and precisely derive the closed-form expressions of the treatment effects.
First, we define a subset of $\mY$ that is sufficiently large and contains $\mY^\circ$.
Next, we derive the closed-form expressions of $\phi_{s,t}$ and the conditional c.d.f.'s of $Y_t$ on this subset of $\mY$.
Proofs of lemmas, propositions, and theorems are collected at the end of this section.

Before we define the required subset of $\mY$, we introduce some useful preliminary results in this section.
Let $V$ and $W$ be scalar-valued random variables whose supports are $\mV$ and $\mW$. Define $\mV^*:=\{v\in\R:F_V(v)-F_V(v-\eps)>0\text{ for all }\eps>0\}$ and $\mW^*:=\{w\in\R:F_W(w)-F_W(w-\eps)>0\text{ for all }\eps>0\}$. The following lemmas show some useful properties of $\mV^*$ and $\mW^*$.  

\begin{lemma}[Size of $\mW^*$]\label{lem:g0}
$\mW$ contains $\mW^*$, and $\mW^*$ contains $\mW^\circ$.
\end{lemma}

\begin{lemma}[Covering the quantiles]\label{lem:g2}
\begin{enumerate}
\item[(a)] For all $\tau\in(0,1)$, $Q_W(\tau)$ is contained in $\mW^*$.

\item[(b)] If $F_W$ is continuous, then, for all $\tau\in(0,1)$, there exists $w\in\mW^*$ such that $F_W(w)=\tau$ holds.
\end{enumerate}

\end{lemma}

\begin{lemma}[Identical supports]\label{lem:g1}
\begin{enumerate}
\item[(a)] If $\mV^*\supset\mW^*$ holds, then $\mV\supset\mW$ holds.

\item[(b)] If $F_W$ is continuous and $\mV\supset\mW$ holds, then $\mV^*\supset\mW^*$ holds.
\end{enumerate}

\end{lemma}

\begin{lemma}[Strict monotonicity on $\mW^*$]\label{lem:g3}
$F_W$ is strictly increasing on $\mW^*$.
\end{lemma}

\begin{lemma}[Identical distributions]\label{lem:g4}
If $F_V(w)=F_W(w)$ holds for all $w\in\mW^*$, and $F_W$ is continuous, then $F_V(w)=F_W(w)$ holds for all $w\in\R$.
\end{lemma}

Lemma \ref{lem:g0} shows that $\mW^*$ is a subset of $\mW$ and contains $\mW^\circ$.
Lemma \ref{lem:g2} shows that $\mW^*$ contains all quantiles of the distribution of $W$.
Lemma \ref{lem:g1} shows that $\mV^*=\mW^*$ implies $\mV=\mW$, and $\mV^*=\mW^*$ and $\mV=\mW$ are equivalent when $F_V$ and $F_W$ are continuous.
Lemma \ref{lem:g3} shows that $F_W$ is strictly increasing on $\mW^*$ and $\mW^\circ$.
Lemma \ref{lem:g4} shows that $\mV^*$ and $\mW^*$ are sufficiently large such that $V$ and $W$ follow the same continuous distribution if they are identically distributed on $\mV^*$ or $\mW^*$.

Now we define the required subset of $\mY$.
For $t\in\mT$, define $\mY^*:=\{y\in\R:F_{Y_t}(y)-F_{Y_t}(y-\eps)>0\text{ for all }\eps>0\}$.\footnote{$\mY^*$ does not depend on $t\in\mT$ from Lemma \ref{lem:g1}. 
We assume that the support of the conditional distribution of $Y_t$ is $\mY$.}
From Lemma \ref{lem:g0}, $\mY^*$ is a subset of $\mY$ and contains $\mY^\circ$.
The following proposition shows that $\phi_{s,t}$ and the conditional cdfs of $Y_t$ are identified on $\mY^*$ and $\mY^\circ$:

\begin{proposition}\label{prop:1}
Suppose that Assumptions \ref{asp:ivqr}-\ref{asp:apb} hold.  
Define $p_t(z)$ as in (\ref{eq:pt}). 
Then the following implication holds:
\begin{enumerate}
\item[(a)] For each $s\in\mT$, $F_{Y_s}(y)$ for $y\in\mY^*$ can be expressed as (\ref{eq:clfm}).

\item[(b)] Assume that $P(D_t(z)\leq D_t({z'}))=1$ and $P(\mC^t_{z,z'})>0$ hold for $t\in\mT$ and $(z,z')\in\mP$. 
Then, for all $t'\in\mT$, $F_{Y_{t'}|\mC_{z,z'}^t}(y)$ for $y\in\mY^*$ can be expressed as (\ref{eq:l4.1}).

\item[(c)] $\phi_{s,t}(y)$ for $y\in\mY^*$ and $s,t\in\mT$ is identified.
\end{enumerate}
\end{proposition}

\begin{remark}\label{rem.g1}
Proposition \ref{prop:1} modifies the results in the main paper that hold on $\mY^\circ$.
Proposition \ref{prop:1} (a) modifies Lemma \ref{lem:clfm} (a).
Proposition \ref{prop:1} (b) modifies Lemma \ref{lem:apb2}.
Proposition \ref{prop:1} (c) modifies Lemma \ref{lem:apb}.
The closed-form expressions derived in Appendices \ref{sec:a.75} and \ref{sec:a.25} also hold on $\mY^*$.
\end{remark}

In the proof of Lemmas \ref{lem:clfm}-\ref{lem:apb2}, we use strict monotonicity of $F_{Y_t}(y)$ and $F_{Y_t|\mC^t_{z,z'}}(y)$ in $y\in\mY^\circ$ to derive the closed-form expressions of $\phi_{s,t}$ and the conditional cdfs of $Y_t$ on $\mY^\circ$.
From Lemma \ref{lem:g3}, $F_{Y_t|X}(y|x)$ is strictly increasing on $\mY^*$ under Assumption \ref{asp:ivqr} (i). 
From Lemmas \ref{lem:g1} and \ref{lem:g3}, $F_{Y_t|\mC^t_{z,z'}}(y)$ for each $(z,z')\in\mP$ is also strictly increasing on $\mY^*$ under Assumption \ref{asp:mt} (iv).
Therefore, applying Lemmas \ref{lem:g1} and \ref{lem:g3} instead of Lemma \ref{lem:3b} leads to the closed-form expressions of $\phi_{s,t}$ and the conditional cdfs of $Y_t$ on $\mY^*$.

Given Proposition \ref{prop:1}, the following proposition shows that the treatment effects are identified as closed-form expressions:

\begin{proposition}\label{prop:2}
Suppose that Assumptions \ref{asp:ivqr}-\ref{asp:apb} hold. Then the following holds:
\begin{enumerate}
\item[(a)] For each $s\in\mT$, $Q_{Y_s}(\tau)$ for $\tau\in(0,1)$ and $E[Y_s]$ can be expressed as
\begin{equation}\label{eq:p2.1}
Q_{Y_t}(\tau)=\inf\{y\in\mY^*:F_{Y_t}(y)\geq\tau\}
\end{equation}
and (\ref{eq:clfmasf}) respectively.

\item[(b)] Assume that $P(D_t(z)\leq D_t({z'}))=1$ and $P(\mC^t_{z,z'})>0$ hold for $t\in\mT$ and $(z,z')\in\mP$. 
Then, for all $t'\in\mT$, $Q_{Y_{t'}|\mC_{z,z'}^t }(\tau)$ for $\tau\in(0,1)$ and $E[Y_{t'}|\mC_{z,z'}^t]$ can be expressed as follows:
\begin{equation}\label{eq:p2.2}
Q_{Y_{t'}|\mC_{z,z'}^t }(\tau)=\inf\{y\in\mY^*:F_{Y_{t'}|\mC_{z,z'}^t }(y)\geq\tau\}
\end{equation}
and (\ref{eq:l4.2}) respectively.
\end{enumerate}
\end{proposition}

\begin{remark}\label{rem.g2}
In the main paper, we show (\ref{eq:clfmasf}) and (\ref{eq:l4.2}) under Assumption \ref{asp:ivqr} (v) in Lemmas \ref{lem:clfm} (b) and \ref{lem:apb2}, respectively. 
In the proof of Proposition \ref{prop:2}, we show (\ref{eq:clfmasf}) and (\ref{eq:l4.2}) without Assumption \ref{asp:ivqr} (v).
\end{remark}

Proposition \ref{prop:2} follows from the fact that $\mY^*$ is sufficiently large such that identification of the distribution on $\mY^*$ leads to the identification of the entire distribution.
From Lemma \ref{lem:g2}, $\mY^*$ contains all the quantiles of the distribution of $Y_t$.
Lemma \ref{lem:g4} implies that $F_{Y_s}(\cdot)$ on $\R$ is specified when $F_{Y_s}(\cdot)$ on $\mY^*$ is specified.
Hence, the identification of $F_{Y_s}(\cdot)$ on $\mY^*$ leads to identification of $Q_{Y_s}(\cdot)$ on $(0,1)$ and $E[Y_s]$.

Furthermore, we can relax Assumption \ref{asp:ivqr} (v) in Theorems \ref{thm:apb} and \ref{thm:apb2} using 
Propositions \ref{prop:1} and \ref{prop:2}.
The assumption that the closure of $\mY^\circ$ is equal to $\mY$ is used in the main paper to assure that identification on the interior of the support leads to identification on the entire support.
However, $\mY^*$ is sufficiently large and does not require such a condition to identify the conditional distributions of $Y_t$ on $\mY$.
The assumption that $F_{Y_t}(\mY^\circ)$ does not depend on $t\in\mT$ is used in the main paper to ensure the existence of an inverse mapping $\phi_{s,t}^{-1}$ on $\mY^\circ$.
However, $\phi_{s,t}^{-1}$ exists on $\mY^*$ because $F_{Y_t}(\mY^*)=(0,1)$ holds and $F_{Y_t}(\mY^*)$ does not depend on $t\in\mT$. 
$F_{Y_t|X}(\mY^*)=(0,1)$ holds because $\mY^*$ contains all quantiles of the distribution of $Y_t$ from Lemma \ref{lem:g2}.

Finally, we show Lemmas \ref{lem:g0}-\ref{lem:g4} and Propositions \ref{prop:1} and \ref{prop:2}.
We modify the proofs of Theorems \ref{thm:apb} and \ref{thm:apb2} using 
Propositions \ref{prop:1} and \ref{prop:2}. 
The proofs are as follows:

\begin{proof}[Proof of Lemma \ref{lem:g0}]
First, we show that $\mW$ contains $\mW^*$. For $w\in\mW^*$, $F_W(w+\eps)>F_W(w-\eps)$ holds for all $\eps>0$ because $F_W(w+\eps)\geq F_W(w)$ holds. Hence, $w$ is also contained in $\mW$, and $\mW$ contains $\mW^*$.
Next, we show that $\mW^*$ contains $\mW^\circ$. Let $w$ be a point not contained in $\mW^*$. Then, there exists $\eps>0$ such that $F_W(w)=F_W(w-\eps)$ holds. Suppose that $W$ is contained in $\mW^\circ$. Then, there exists $\delta>0$ such that $\eta\leq\delta\To w-\eta\in\mW$ holds. If this $\eta$ also satisfies $\eta\leq\eps/2$, then $F_W(w)>F_W(w-2\eta)$ holds because $w-\eta$ is contained in $\mW$. 
However, this fact contradicts the fact that
$F_W(w-2\eta)\geq F_W(w-\eps)$ holds because $F_W(w)=F_W(w-\eps)$ holds. Hence, $W$ is not contained in $\mW^\circ$, and $\mW^*$ contains $\mW^\circ$.
Therefore, the stated result follows.
\end{proof}

\begin{proof}[Proof of Lemma \ref{lem:g2}]
First, we show part (a).
Suppose that there exists $\tau'\in(0,1)$ such that $Q_W(\tau')$ is not contained in $\mW^*$. Then, there exists $\eps>0$ such that $F_W(Q_W(\tau'))=F_W(Q_W(\tau')-\eps)$ holds. However, this fact contradicts the fact that $F_W(Q_W(\tau'))\geq\tau'$ holds, according to the definition of $Q_W$. Hence, $Q_W(\tau)$ is contained in $\mW^*$ for all $\tau\in(0,1)$.

Next, we show part (b).
Observe that $F_W(Q_W(\tau))=\tau$ holds for all $\tau\in(0,1)$ because $F_W$ is continuous. 
Therefore, the stated result follows by taking $w=Q_W(\tau)$ for each $\tau$.
\end{proof}

\begin{proof}[Proof of Lemma \ref{lem:g1}]
First, we show part (a).
Let $w\in\mW$ satisfy $w\notin\mV$.
Then there exists $\eps>0$ such that $F_V(w+\eps)=F_V(w-\eps)$. This holds because $F_V(w+\eps)\geq F_V(w-\eps)$ holds.
This implies that $(w-\eps,w+\eps)$ is not contained in $\mW^*$ because $\mV^*\supset\mW^*$ and $\mV^*\subset\mV$ hold from Lemma \ref{lem:g0}.
Thus, there exists $\tau\in(0,1)$ such that $F_W(\eta)=\tau$ holds for all $\eta\in(w-\eps,w+\eps)$.
We proceed to show this statement.
Suppose that $\eta,\eta'\in(w-\eps,w+\eps)$ satisfies $\eta<\eta'$ and $F_W(\eta)<F_W(\eta')$.
Then, $\eta<Q_W(F_W(\eta'))\leq\eta'$ holds, and
$Q_W(F_W(\eta'))$ is contained in $\mW^*$ from Lemma \ref{lem:g2}.
However, this contradicts the fact that $(w-\eps,w+\eps)$ is not contained in $\mW^*$.
Hence, the above statement follows.
This implies that $w$ is not contained in $\mW$.
However, this contradicts $w\in\mW$.
Therefore, $\mV\supset\mW$ holds.

Next, we show part (b).
Let $w\in\mW^*$ satisfy $w\notin\mV^*$. 
Then there exists $\eps>0$ such that $F_V(w)=F_V(w-\eps)$ holds.
This is because $F_V(w)\geq F_V(w-\eps)$ holds. 
This implies that $(w-\eps,w)$ is not contained in $\mW$ because $\mV\supset\mW$ holds.
Thus, there exists $\tau\in(0,1)$ such that $F_W(w-\eta)=\tau$ and $\tau<F_W(w)$ hold for all $0<\eta<\eps$.
We proceed to show this statement.
Suppose that $0<\eta<\eta'<\eps$ satisfies $F_W(w-\eta)<F_W(w-\eta')$.
Then, $w-\eta<Q_W(F_W(w-\eta'))\leq w-\eta'$ holds, and $Q_W(F_W(w-\eta'))$ is contained in $\mW^*$ from Lemma \ref{lem:g2}.
However, this contradicts the fact that $(w-\eps,w)$ is not contained in $\mW$.
This is because $Q_W(F_W(w-\eta'))$ is also contained in $\mW$ from Lemma \ref{lem:g0}.
Hence, the above statement follows, which contradicts the continuity of $F_W$. 
Therefore, $\mV^*\supset\mW^*$ holds.
\end{proof}

\begin{proof}[Proof of Lemma \ref{lem:g3}]
Let $w_1,w_2\in\mW^*$ satisfy $w_1<w_2$. Then, because $w_2$ is contained in $\mW^*$, $F_W(w_2)>F_W(w_1)$ holds. Therefore, the stated result follows.
\end{proof}

\begin{proof}[Proof of Lemma \ref{lem:g4}]
We show that $F_V(w)\geq F_W(w)$ holds for all $w\in\R$. 
Suppose that there exists $w\in\R$ such that $F_V(w)< F_W(w)$ holds. Then, for $F_V(w)<\tau'< F_W(w)$, there exists $w'\in\mW^*$ such that $w'<w$ and $F_V(w')=F_W(w')=\tau'$ hold from Lemma \ref{lem:g2}. However, this implies that $\tau'\leq F_V(w)$ holds and contradicts with the fact that $F_V(w)< F_W(w)$ holds. Hence, $F_V(w)\geq F_W(w)$ holds for all $w\in\R$.
Similarly, we can show that $F_V(w)\leq F_W(w)$ holds for all $w\in\R$. Therefore, the stated result follows.
\end{proof}

\begin{proof}[Proof of Proposition \ref{prop:1}]
Part (a) follows from an argument similar to the proof of Lemma \ref{lem:clfm} (a), using Lemmas \ref{lem:g1} and \ref{lem:g3} instead of Lemma \ref{lem:3b}.
Part (b) follows from an argument similar to the derivation of (\ref{eq:l4.1}) in the proof of Lemma \ref{lem:apb2}, using Lemmas \ref{lem:g1} and \ref{lem:g3} instead of Lemma \ref{lem:3b}.

For part (c), it suffices to show that $\phi_{k,j}$ for $j=0,\ldots,k-1$ is identified on $\mY^*$ because the other counterfactual mappings are identified from $\phi_{k,j}$ for $j=0,\ldots,k-1$.
To see this, observe that $\phi_{s,t}^{-1}$ exists on $\mY^*$ and $\phi_{t,s}(y)=\phi_{s,t}^{-1}(y)$ holds for $y\in\mY^*$.
This is because $\phi_{s,t}$ is strictly increasing on $\mY^\circ$ from Assumption \ref{asp:ivqr} (i) and Lemma \ref{lem:g3},
and $\phi_{s,t}(\mY^*)=\mY^*$ holds from Lemma \ref{lem:g2}.
Note that $F_{Y_t}(\mY^*)=(0,1)$ holds and $F_{Y_t}(\mY^*)$ does not depend on $t\in\mT$.
This is because $\mY^*$ contains all quantiles of the distribution of $Y_t$ by Lemma \ref{lem:g2}, which leads to $\phi_{s,t}(\mY^*)=\mY^*$.
Furthermore, for $s,t,r\in\mT$, $\phi_{s,r}$ is identified on $\mY^*$ if $\phi_{s,t}$ and $\phi_{t,r}$ are identified on $\mY^*$.
This is because $\phi_{s,t}(\mY^*)=\mY^*$ holds and
$\phi_{s,r}=\phi_{t,r}\circ\phi_{s,t}$ follows from $F_{Y_t}(Q_{Y_t}(\tau))=\tau$ for $\tau\in(0,1)$ by Assumption \ref{asp:ivqr} (i).
Therefore, identification of $\phi_{k,j}$ for $j=0,\ldots,k-1$ suffices to identify all the counterfactual mappings.

Identification of $\phi_{k,j}$ for $j=0,\ldots,k-1$ on $\mY^*$ follows from an argument similar to the proof of Lemma \ref{lem:apb}, using Lemmas \ref{lem:g1} and \ref{lem:g3} instead of Lemma \ref{lem:3b}.
Note that $y\in\mY^*$ implies $\phi_{s,t}(y)\in\mY^*$ because $\phi_{s,t}(\mY^*)=\mY^*$ holds from Lemma \ref{lem:g2}.
Other counterfactual mappings are also identified on $\mY^*$ because they are inversions or compositions of $\phi_{k,j}$ for $j=0,\ldots,k-1$.
The closed-form expressions of $\phi_{s,t}$ on $y\in\mY^*$ follow from an argument similar to that in Appendices \ref{sec:a.75} and \ref{sec:a.25}, using Lemmas \ref{lem:g1} and \ref{lem:g3} instead of Lemma \ref{lem:3b}.
\end{proof}

\begin{proof}[Proof of Proposition \ref{prop:2}]
We show part (a). Part (b) follows from a similar argument.
Observe that (\ref{eq:p2.1}) holds because $F_{Y_t}(Q_{Y_t}(\tau))=\tau$ and 
$Q_{Y_t}(\tau)\in\mY^*$ hold for all $\tau\in(0,1)$ from Lemma \ref{lem:g2}.
Similar to the derivation of \eqref{eq:33.5} in the proof of Lemma \ref{lem:clfm} (b), we obtain
\begin{equation}\label{eq:g.1}
F_{Y_s|TZ}(y|t,z)=F_{\phi_{t,s}(Y_t)|TZ}(y|t,z)\text{ for }y\in\mY^*.
\end{equation}
Applying Lemma \ref{lem:g4} to (\ref{eq:g.1}) gives
\begin{equation}\label{eq:g.2}
F_{Y_s|TZ}(y|t,z)=F_{\phi_{t,s}(Y_t)|TZ}(y|t,z)\text{ for }y\in\R.
\end{equation}
Hence, $Y_s\deq \phi_{t,s}(Y_t)$ conditional on $(T,Z)=(t,z)$ holds.
Therefore, we obtain (\ref{eq:clfmasf}) as in the proof of Lemma \ref{lem:clfm} (b).
\end{proof}

\begin{proof}[Modified proof of Theorem \ref{thm:apb}]
Observe that $F_{Y_s}(y)$ for $y\in\mY^*$ is identified from Proposition \ref{prop:1} (a) because $\phi_{s,t}(y)$ for $y\in\mY^*$ is identified from Proposition \ref{prop:1} (c).
Hence, $Q_{Y_s}(\tau)$ for $\tau\in(0,1)$ and $E[Y_s]$ are identified from Proposition \ref{prop:2} (a).
\end{proof}

\begin{proof}[Modified proof of Theorem \ref{thm:apb2}]
First, we modify the proof of Lemma \ref{lem:apb2} and show that $Q_{Y_{t'}|\mC_{z,z'}^t }(\tau)$ for $\tau\in(0,1)$ and $E[Y_{t'}|\mC_{z,z'}^t]$ are identified for all $t'\in\mT$ if $P(D_t(z)\leq D_t({z'}))=1$ and $P(\mC^t_{z,z'})>0$ hold for $t\in\mT$, $(z,z')\in\mP$, and $x\in\mX$.
Observe that $F_{Y_{t'}|\mC_{z,z'}^t }(y)$ for $y\in\mY^*$ is identified from Proposition \ref{prop:1} (b) because $\phi_{s,t}(y)$ for $y\in\mY^*$ is identified from Proposition \ref{prop:1} (c).
Hence, $Q_{Y_{t'}|\mC_{z,z'}^t }(\tau)$ for $\tau\in(0,1)$ and $E[Y_{t'}|\mC_{z,z'}^t]$ are identified from Proposition \ref{prop:2} (b).
The stated result follows from applying this result instead of Lemma \ref{lem:apb2} to the proof in the main paper.
\end{proof}